\documentclass[a4paper,11pt]{article}
\pdfoutput=1 
\usepackage{amsthm,graphicx,amsmath,amssymb,amsfonts,framed,upgreek,physics,setspace}   
\usepackage{slashed} 
\usepackage{fullpage,xcolor}

\usepackage{hyperref}
\hypersetup{colorlinks=true, linkcolor=blue, citecolor=red, urlcolor=blue}

\usepackage[nottoc,notlot,notlof]{tocbibind}

\usepackage[utf8]{inputenc}
\usepackage[T1]{fontenc}
\begin{document} 
\title{\vspace{-1.5truecm}\bf Lorentzian fermionic action \\ by twisting euclidean spectral triples}
\author{Pierre Martinetti,\textsuperscript{\dag,\ddag} 
		Devashish Singh\textsuperscript{\dag} \\[5pt] 
\textsuperscript{\dag}\emph{DIMA, Universit\`a di Genova},
\textsuperscript{\ddag}\emph{INFN sezione di Genova,} \\[0pt]  
\emph{via Dodecaneso, 16146 Genova Italy.} \\[6pt] 
\emph{E-mail:}	\href{mailto:martinetti@dima.unige.it}{martinetti@dima.unige.it}, 
				\href{mailto:devashish@dima.unige.it}{devashish@dima.unige.it}}

\newtheorem{theorem}{Theorem}[section] 
\newtheorem{corollary}{Corollary}[theorem] 
\newtheorem{lemma}[theorem]{Lemma} 
\newtheorem{proposition}[theorem]{Proposition} 
\newtheorem{definition}[theorem]{Definition} 
\newtheorem{remark}[theorem]{Remark}

\newcommand{\A}{{\cal A}}
\newcommand{\D}{{\cal D}}
\newcommand{\HH}{{\cal H}}
\newcommand{\M}{{\cal M}}
\newcommand{\J}{\mathcal{J}}
\newcommand{\R}{\mathcal{R}}
\newcommand{\dmu}{\text{d}\upmu}


\maketitle

\begin{abstract}
We show how the twisting of spectral triples induces a
transition from an euclidean to
        a lorentzian noncommutative geometry, at the level fo the fermionic action.
       More specifically,
        we compute the fermionic action for the twisting of a closed
        euclidean manifold, then that of a two-sheet 
	euclidean manifold, and finally the twisting of the spectral triple of 
	electrodynamics in euclidean signature. We obtain the Weyl and the Dirac equations in lorentzian 
	signature (and in the temporal gauge). 
The twisted fermionic action is then shown to be invariant under an action
of the Lorentz group. This permits to interprete the field of $1$-form
that parametrizes the twisted fluctuation of a manifold as the (dual)
of the energy momentum $4$-vector. 
\end{abstract}


\setcounter{tocdepth}{1}
\tableofcontents

\newpage
\section{Introduction}
\label{sec:intro}

Noncommutative geometry~\cite{Connes:1994kx} offers various ways to build models
  beyond the \emph{standard model of elementary particles} (SM), 
 recently reviewed in~\cite{Chamseddine:2019aa,Devastato:2019ab}. 
  One of them~\cite{a2,a4}
  consists in twisting the spectral
  triple of SM by an algebra automorphism, 
  in the sense of Connes, Moscovici~\cite{CM}. This provides a
  mathematical justification to the extra scalar
  field introduced in~\cite{Chamseddine:2012fk} to both fit the mass of the
  Higgs and stabilise the electroweak
  vacuum. A significant difference from the construction based on
  spectral triples without  first order 
  condition~\cite{Chamseddine:2013fk,Chamseddine:2013uq} is that the twist does
  not only yield an extra scalar field, but also 
  a supplementary $1$-form field {\footnote{In \cite{a4} this field
      was improperly called \emph{vector field}.}, whose meaning was rather unclear so far.

Connes' theory of spectral triples provides a spectral characterization of compact riemannian
manifolds~\cite{connesreconstruct} along with the tools for their noncommutative
generalisation~\cite{Connes:1996fu}. Extending this program to the pseudo-riemannian case is
 notoriously difficult.
Although several 
interesting results in this context have been obtained recently, see 
e.g.~\cite{Besnard:2016ab, Rennie12, Franco2014, CQG2013}, there is no
reconstruction theorem for pseudo-riemannian manifolds in view, and
it is still unclear how the spectral action should be handled in a
pseudo-riemannian signature. 

Quite unexpectedly, the twist of the SM, which has been introduced in a purely 
riemannian context, has something to do with the transition from
  the euclidean signature to the lorentzian one. In fact, the inner product
  induced by the twist on the Hilbert space of euclidean spinors
  on a
  four-dimensional manifold $\M$, 
  coincides with the Krein product of lorentzian spinors~\cite{a5}.
This is not so surprising, for the twist $\rho$ coincides
with the automorphism that exchanges the two
eigenspaces of the grading operator  (in physicist's words: that exchanges the
left and the right components of spinors). 
And this is nothing but the inner
automorphism induced by the first Dirac matrix
$\gamma^0=c(dx^0)$. This explains why, by twisting, one is somehow able
to single out the $x_0$ direction among the four riemannian dimensions of
$\M$. However, the promotion of this $x_0$ to a ``time direction'' is not
fully accomplished, at least not in the sense of Wick rotation~\cite{b5}.
Indeed, regarding the Dirac matrices, the inner automorphism induced
by $\gamma^0$ does not
implement the Wick rotation (which maps the spatial Dirac matrices
$\gamma^j$ to $W(\gamma^j):=i\gamma^j$)  but actually its square:
\begin{equation}
\rho(\gamma^j)=\gamma^0\gamma^j\gamma^0 =-\gamma^j = W^2(\gamma^j), \qquad \text{for} \quad j=1,2,3.
\label{eq:1}
 \end{equation}

Nevertheless, a transition from the euclidean to the lorentzian does
occur, and the $x_0$ direction gets promoted to a time
direction, but this happens at the level of the fermionic action. This
is the main result of this paper, summarised in propositions
\ref{Prop:Weyl}, \ref{prop:Diracfinal}, and their lorentz invariant
version propositions \ref{Prop:Weylboosted} and \ref{Prop:Diracboost}.

More specifically, starting with the twisting of an
\emph{euclidean} manifold, then that of a
  two-sheet euclidean manifold, and finally the twisting of the spectral triple of
  electrodynamics in euclidean signature~\cite{W1}; we show how the fermionic action for
  twisted spectral triples, proposed in~\cite{a5}, actually
  yields the Weyl and
the Dirac equations in \emph{lorentzian
  signature}. In addition, the  extra $1$-form field acquires a clear interpretation as
 the dual of the  energy-momentum $4$-vector.

The following three aspects 
of the twisted fermionic action 
explain  the change of signature:
\begin{itemize}
\item First, in order to guarantee that the fermionic action is
  symmetric when evaluated on Gra{\ss}mann variables (which is an
  important requirement for the whole physical interpretation of the
  action formula, also in the non-twisted case~\cite{CCM}), 
  one restricts the bilinear form
 that defines the action to the
  $+1$-eigenspace $\HH_\R$ of the unitary
  operator $\cal R$ that implements the twist;  whereas in the
  non-twisted case, the restriction is to the $+1$-eigenspace of the
  grading, in order to solve the fermion doubling problem.  This
  different choice of eigenspace had been
  noticed in~\cite{a5} but the physical consequences were not
  drawn. As already
  emphasised above, in the models relevant for physics, ${\cal R}=\gamma^0$, 
  and once restricted to $\HH_\R$, the bilinear
  form no longer involves derivative in the $x_0$
  direction. In other words, the restriction to $\HH_\R$ projects the
  euclidean fermionic action to what will constitute its spatial
  part in lorentzian signature.

\item Second, the twisted fluctuations of the Dirac operator of a
four-dimensional riemannian manifold are not necessarily zero~\cite{a4,LM1}, 
in contrast
with the non-twisted case where those fluctuations always vanish. These
are parametrised by the above-mentioned $1$-form field. By
interpreting the zeroth component of this field as an energy, one recovers a derivative in
the $x_0$ direction, but now in a lorentzian signature.

\item Third, we show that the twisted fermionic action is invariant
  under an actin of the Lorentz group. From that follows the interpretation of the whole $1$-form field (not only its
zeroth component) as the dual of the energy-momentum $4$-vector.
\end{itemize}

All this is detailed as follows. In section~\ref{sec:secfermion}, we review 
known material regarding twisted spectral
triples, their compatibility with the real
structure~(\S\ref{subsec:twistedsp}) and the new
inner product they induce on the initial Hilbert space~(\S\ref{subsec:rho}). We
discuss what
a covariant Dirac operator is in the twisted context, and the corresponding gauge
invariant fermionic action it defines (\S\ref{subsec:fermionaction}). We
finally recall how to associate a
twisted partner to  graded spectral
triples~(\S\ref{sec:2.2}).

In section \ref{sec:Weyl}, we investigate the fermionic action for the
minimal twist of a closed euclidean manifold, that is, the twisted
spectral triple having the same Hilbert space and Dirac operator as
the canonical triple of the manifold, but whose algebra is doubled in
order to make the twisting possible~(\S\ref{subsec:minimaltwistriema}). 
In~\S\ref{lem:3.1}, we show
that twisted fluctuations of the Dirac operator are parametrised by
a $1$-form field of components $X_\mu$, first discovered in~\cite{a4}. In~\S\ref{sec:3.1.2}, we recall 
how to deal with gauge transformations in a twisted context, along the lines 
of~\cite{LM2}. We then compute the twisted fermionic action 
in~\S\ref{subsec:fermioncactionmanif} and show that it yields a lagrangian
density similar to that of the Weyl equations in lorentzian
  signature, as soon as one interprets the zeroth component of $X_\mu$ as the time component of the energy-momentum
$4$-vector of fermions. However, there are not enough spinor degrees-of-freedom to deduce the Weyl equations for this lagrangian density.

 That is why in section~\ref{sec:doublem} we double the twisted
manifold~(\S\ref{subsec:minimaltwistdoublema}), compute the
twisted-covariant Dirac operator~(\S\ref{subsec:twistfluc2man}), 
and obtain Weyl equations from the fermionic action~(\S\ref{subsec:Weyl}).

In section~\ref{sec:electrody}, we apply the same construction to the
spectral triple of
electrodynamics proposed in~\cite{W1}. Its minimal twist is written 
in~\S\ref{subsec:mintwised}, the twisted fluctuations are calculated 
in~\S\ref{sec:4.3}, for both the free part and the finite parts of the Dirac
operator. The gauge transformations are studied 
in~\S\ref{subsec:gaugetransformED} and, finally, the Dirac equation in lorentzian
signature (and in the temporal gauge) is obtained in \S
\ref{sec:Dirtaceq}. 

Section  \ref{sec:Lorentzinv} deals with Lorentz invariance.

We conclude with some outlook and perspective. The appendices contain
all the required notations for the Dirac matrices and for the Weyl \& Dirac
equations.
\medskip

The Lorentz metric is $(+1,-1,-1,-1)$. We use Einstein convention for
summing on
alternate (up/down) indices: for instance $\gamma^\mu\partial_\mu$
stands for
$\sum_\mu \gamma^\mu \partial_\mu$.  

\bigskip
\newpage

\section{Fermionic action for twisted spectral geometry}
\label{sec:secfermion}
\setcounter{equation}{0} 

After an introduction to twisted spectral triples (\S\ref{subsec:twistedsp}), 
we recall how the inner product induced by the twist
on the Hilbert space  (\S\ref{subsec:rho}) permits building a fermionic action (\S\ref{subsec:fermionaction}). The key difference
with the usual (i.e.~non twisted)
case is that one
no longer restricts to the positive eigenspace of the grading $\Gamma$, but 
rather to that of
the unitary $\R$~implementing~the~twist. 
Finally, we  emphasise the
twist-by-grading procedure, that associates a twisted 
partner to any graded spectral triple whose representation is
sufficiently faithful~(\S\ref{sec:2.2}).

\subsection{Real twisted spectral triples}
\label{subsec:twistedsp}

Twisted spectral triples have been introduced to build noncommutative
geometries from type III algebras~\cite{CM}. Later, they found applications in
high energy physics describing extensions of Standard Model, such
as the Grand symmetry model~\cite{a2,a4}.

\begin{definition}[from~\cite{CM}]
  A \emph{twisted spectral triple} $({\cal A, H, D})_{\rho}$ is a unital *-algebra ${\cal A}$
  that acts faithfully on a Hilbert space ${\cal H}$ as bounded
  operators,\footnote{Wherever applicable, we use $a$ to mean its representation $\pi(a)$. Thus, 
    $a^*$ denotes $\pi(a^*) = \pi(a)^\dag$, where $*$ is the involution of 
    ${\cal A}$ and $\dag$ is the Hermitian conjugation on ${\cal H}$.}
  along with a self-adjoint operator 
  $\cal D$ on ${\cal H}$ with compact resolvent, called the
  \emph{Dirac operator}, and an automorphism $\rho$ of ${\cal A}$ such that the
  \emph{twisted commutator}, defined as
  \begin{equation}
  \label{twistcom}
    [{\cal D},a]_{\rho} := {\cal D}a - \rho(a){\cal D}, 
  \end{equation}
  is bounded for any $a\in\A$ (that is $[{\cal D},a]_{\rho}$ is well
  defined on the domain of $\cal D$, and extends to a bounded operator
  on $\HH$).
\end{definition}

A \emph{graded} twisted spectral triple is one endowed with a self-adjoint operator
$\Gamma$ on ${\cal H}$ such~that
\begin{equation}
	\Gamma^2 = \mathbb{I},
	\qquad \Gamma {\cal D} + {\cal D}\Gamma = 0,
	\qquad \Gamma a = a\Gamma,
	\quad \forall a \in {\cal A}.
\end{equation}
The real structure \cite{Connes:1995kx}  easily
adapts to the twisted case \cite{LM1}: as in the non-twisted case, one
considers an antilinear isometry 
$J : {\cal H} \rightarrow {\cal H}$, such that
\begin{equation}
\label{eq:KOdim}
	J^2 = \epsilon\mathbb{I}, \qquad 
	J{\cal D} = \epsilon'{\cal D}J, \qquad 
	J\Gamma = \epsilon''\Gamma J,
\end{equation}
where the signs $\epsilon, \epsilon', \epsilon'' \in \{\pm1\}$ determine the 
$KO$-dimension of the twisted spectral triple.
In addition, $J$ is required to implement an
isomorphism between $\A$ and its opposite algebra~$\A^\circ$,
\begin{equation}
\label{eq:7}
	b \mapsto b^{\circ} := Jb^*J^{-1}, \qquad \forall b \in {\cal A}.
\end{equation}
One requires this action of $\A^\circ$ on $\HH$ to commute with that of $\cal A$ (the \emph{order zero condition}),
\begin{equation}
\label{eq:order0}
	[a, b^{\circ}] = 0, \qquad \forall a, b \in {\cal A},
\end{equation}
in order to define a right representation of
${\cal A}$ on ${\cal H}$:
\begin{equation}
	\psi a := a^{\circ} \psi = Ja^*J^{-1}\psi, \qquad 
	\forall \psi \in {\cal H}.
\end{equation}

The part of the real structure that is modified is
the \emph{first order condition}. In the non-twisted case, it 
reads: $[[D,a], b^\circ]=0, \;\forall a,b\in\A$; while in the
twisted case, it becomes~\cite{a4,LM1}
\begin{equation}
\label{eq:order1}
	[[{\cal D},a]_{\rho}, b^{\circ} ]_{\rho^{\circ}} := 
	[{\cal D},a]_{\rho}b^{\circ} - \rho^{\circ}(b^{\circ})[{\cal
          D},a]_{\rho} = 0, \qquad \forall a, b \in\A,
\end{equation}
where $\rho^\circ$ is the automorphism induced by $\rho$ on the opposite algebra:
\begin{equation}
\rho^{\circ}(b^{\circ}) = \rho^{\circ}(Jb^*J^{-1})
:= J\rho(b^*)J^{-1}.
\label{eq:6}
\end{equation}
\begin{definition}[from~\cite{LM1}]
A \emph{real twisted spectral triple} is a graded twisted spectral triple, along with a real structure $J$
  satisfying~\eqref{eq:KOdim}, the zeroth and the first order conditions~\eqref{eq:order0},~\eqref{eq:order1}.
\end{definition}
In case the automorphism $\rho$ coincides with an inner automorphism
of ${\cal B}({\cal H})$, that is
\begin{equation}
  \label{eq:14}
  \pi(\rho(a))={\cal R} \pi(a)  {\cal R}^\dag, \qquad \forall a\in\A,
\end{equation}
where $\cal R\in{\cal B}(\HH)$ is unitary, then $\rho$ is said \emph{compatible
with the real structure} $J$, as soon as
\begin{equation} 
\label{eq:2.11}
	J{\cal R} = \epsilon'''\,{\cal R}J, \qquad \text{ for } \quad
        \epsilon'''=\pm.
\end{equation}
The inner automorphism, hence the unitary $\cal R$, are not necessarily
unique. In that case, $\rho$ is compatible with the real
structure if there exists at least one $\R$ satisfying the above conditions.
\begin{remark}
\label{rem:autmodul}
  In the original definition \cite[(3.4)]{CM}, the automorphism is not required to be a
  $*$-automorphism, but rather to satisfy
   the regularity condition $\rho(a^*)=\rho^{-1}(a)^*$.
    If, however, one requires $\rho$ to be a $*$-automorphism, then
  the regularity condition implies that
  \begin{equation}
    \rho^2= {\sf Id}.
    \label{eq:17}
  \end{equation}
\end{remark}
Other modifications of spectral triples by twisting the real
structure have been proposed~\cite{T.-Brzezinski:2016aa}. 
Interesting relations with the above real twisted spectral triples have been worked 
out in~\cite{Brzezinski:2018aa}.

\subsection{Twisted inner product}
\label{subsec:rho}

Given a Hilbert space  $\left( {\cal H}, \langle \cdot, \cdot \rangle
\right)$ and an automorphism $\rho$ of $\mathcal{B}({\cal H})$,
a $\rho$-\emph{product} $\langle \cdot , \cdot \rangle_{\rho}$ is an inner 
product satisfying
\begin{equation} \label{3.2}
	\langle\phi,  \mathcal{O}\xi \rangle_{\rho}
	= \langle {\rho(\mathcal{O})}^\dag\phi, \xi \rangle_{\rho},
        \qquad \forall \, \mathcal{O} \in \mathcal{B}({\cal H}) \;\text{ and }\;
        \phi, \xi \in {\cal H},
\end{equation}
where $\dag$ is the Hermitian adjoint with respect to the inner product $\langle \cdot,
\cdot \rangle$. One calls
\begin{equation}
	\mathcal{O}^{+} := {\rho(\mathcal{O})}^\dag
\label{eq:13}
\end{equation}
the $\rho$-\emph{adjoint} of the operator $\mathcal{O}$.
If $\rho$ is inner and implemented by a 
unitary operator ${\cal R}$ on ${\cal H}$ -- that is, 
$\rho(\mathcal{O}) = {\cal R}\mathcal{O}{\cal R}^\dag$ for any
        $\mathcal{O} \in \mathcal{B}({\cal H})$ --
then, a canonical $\rho$-product is
\begin{equation} 
\label{rho-p}
	\langle \phi, \xi \rangle_{\rho} = \langle \phi, {\cal R}\xi \rangle.
\end{equation}

The $\rho$-adjointness is not necessarily an involution.  If $\rho$ is a
$*$-automorphism (for instance, when $\rho$ is inner), then $^+$ is
an involution iff \eqref{eq:17} holds, for 
\begin{equation}
  \label{eq:18}
  ({\cal O}^+)^+ = \rho({\cal O}^+)^\dag = \rho(({\cal O}^+)^\dag) 
  = \rho(\rho({\cal O})).
\end{equation}
\begin{remark}
   The
  regularity condition in Rem.~\ref{rem:autmodul} (written as
  $\rho(b)^*=\rho^{-1}(b^*)$ for any $b=a^*\in \A$) is equivalent to
  the $\rho$-adjointness $a^+:=\rho(a)^*$ being an involution, for 
  \begin{equation}
    \label{eq:16}
    (a^+)^+ = (\rho(a)^*)^+ = \left(\rho(\rho(a)^*)\right)^* =
    \left(\rho(\rho^{-1}(a^*))\right)^*= (a^*)^*=a.
  \end{equation}
\end{remark}

Given a twisted spectral triple $(\A, \HH, D)_\rho$ whose
  twisting-automorphism $\rho$
  coincides with an automorphism of ${\cal B}({\cal H})$, any
  choice of the unitary $\cal R$ implementing this automorphism  induces a natural
  twisted inner product \eqref{rho-p} on $\HH$. These products are useful to define a
  gauge invariant fermionic action.

\subsection{Twisted fermionic action}
\label{subsec:fermionaction}

The \emph{fermionic action} for a real  
spectral triple $({\cal A},{\cal H,D};J,\Gamma)$ is \cite{CCM,Barrett:2007vf}
$S(\D_{\omega}) :=
    \mathfrak{A}_{\D_\omega}(\tilde\xi, \tilde\xi)$,
where 
\begin{equation}
\label{eq:4}
	\mathfrak{A}_{\D_\omega}(\phi,\xi)
	:= \langle J\phi, \D_{\omega}\xi \rangle,
	\quad \phi,\psi \in \cal H
\end{equation}
is a bilinear form defined by the covariant Dirac operator $\D_{\omega} 
:= \D + \omega + \epsilon'J\omega J^{-1}$~\cite{Connes:1996fu},
with $\omega$ is a self-adjoint element of the set of generalised one-forms
\begin{equation}
\label{eq:3}
	\Omega^1_\D({\cal A})
	:= \left\{ \sum\nolimits_i a_i[\D, b_i], \quad a_i, b_i \in \cal A \right\}\!;
\end{equation}
while $\tilde{\psi}$ is a Gra{\ss}mann vector in the Fock space $\tilde{\cal H}_+$ 
of classical fermions, corresponding to the positive eigenspace 
${\cal H}_+ \subset {\cal H}$ of the grading $\Gamma$, that is,
	\begin{equation}
\label{eq:H_+}
	\tilde{\cal H}_+ := \{ \tilde{\xi}, \,\, \xi \in {\cal H}_+ \}, 
	\quad \text{ where} \quad{\cal H}_+ := \{ \xi\in \HH, \,\, \Gamma\xi = \xi \}.
	\end{equation}
The fermionic action is invariant under a gauge transformation, that
is  the simultaneous 
adjoint action of the group ${\cal U}({\cal A})$ of unitaries of $\A$, both on $\HH \ni \psi$, 
\begin{equation}
  \label{eq:20}
  ({\sf Ad}\,u)\psi  := u \psi u^* = u (u^*)^\circ \psi = uJuJ^{-1} \quad
  u\in {\cal U}({\cal A})
\end{equation}
and on the covariant Dirac operator: $\D_\omega\to ({\sf Ad}\,u)  \D_\omega ({\sf Ad}\,u)^\dag$. 
        \begin{remark}
\label{rem:bilingrassman}
          The form~\eqref{eq:4} is antisymmetric in
          $KO$-dim $2,4$ (lemma.~\ref{lemma:antisymm}~below), 
         so  $\frak A_{\D_\omega}(\xi,
          \xi)$ vanishes when evaluated on vectors. However, it is non-zero when
          evaluated on Gra{\ss}mann vectors \cite[\S I.16.2]{CM0}.
In particular, for the spectral triple of
Standard Model (of $KO$-dim $2$), the fermionic action 
contains the coupling of matter with fields (scalar, gauge, and
gravitational).
   \end{remark}

In twisted spectral geometry, the fermionic action is 
constructed~\cite{a5} substituting~$\D_{\omega}$~with a twisted
  covariant Dirac operator 
\begin{equation}
\label{eqw:twistfluct}
	\D_{\omega_{\rho}} := \D + \omega_{\rho} + \epsilon'J\omega_{\rho}J^{-1},
\end{equation}
where $\omega_\rho$ is an element of the set of twisted one-forms \cite{CM},
\begin{equation}
\label{eq:8}
	\omega_{\rho} \in \Omega_D^1(\A,\rho)
	:= \left\{ \sum\nolimits_j a_j[\D,b_j]_{\rho},
 	\quad a_j, b_j \in {\cal A} \right\}\!,
\end{equation}
such that $\D_{\omega_\rho}$ is
self-adjoint;{\footnote{The domain of $\D_{\omega_\rho}$ coincides with the
    one of $\cal D$ (being $\omega_\rho + {\cal J}\omega_\rho {\cal
      J}^{-1}$ in ${\cal B}({\cal H})$). By Kato-Relish theorem,
    ${\cal D}_{\omega_\rho}$ is selfadjoint iff $\omega_\rho +
    J\omega_\rho J^{-1}$ is selfadjoint. In \cite{a4} we required $\omega_\rho +
    J\omega_\rho J^{-1}$ to be selfadjoint without necessarily imposing the
    self-adjointness of
    $\omega_\rho$. This is discussed in detail after Lem.~\ref{lem:3.1} below.}} 
    and by replacing the inner product 
with the $\rho$-product~\eqref{3.2}. 
Instead
of~\eqref{eq:4}, one thus considers the bilinear form 
\begin{equation}
\label{Sfrho}
	{\mathfrak A}^\rho_{D_{\omega_\rho}}\!(\phi, \xi)
	:= \langle J\phi, \D_{\omega_{\rho}}\xi \rangle_{\rho}.
\end{equation} 

A gauge transformation is given by  the same action (\ref{eq:20}) of
${\cal U}({\cal A})$ on
$\HH$, but the Dirac operator transforms in the following `twisted' manner~\cite{LM2}:
\begin{equation}
\label{eq:9}
	\D_{\omega_\rho} \to ({\sf Ad}\,\rho(u))\D_{\omega_\rho}({\sf Ad}\,u^*).
\end{equation}
The r.h.s. of \eqref{eq:9} is still a twisted covariant
Dirac operator $D_{\omega_\rho^u}$ where~\cite[Prop.~4.2]{LM2}
\begin{equation}
	\label{LM2Prop4.2}
		 \omega^u_\rho
		:= \rho(u) \left([D,u^*]_\rho + \omega_\rho u^*\right).
	\end{equation}
The transformation $\omega_\rho \to\omega_\rho^u$ is the twisted
version of the law of transformation of the gauge potential in
noncommutative geometry \cite{Connes:1996fu}.
\newpage

In case the twist $\rho$ is compatible with the real structure in the
sense of~\eqref{eq:2.11} for some unitary $\cal R$,
the bilinear form \eqref{Sfrho} is invariant under the
simultaneous transformation~(\ref{eq:20})--(\ref{eq:9})~\cite[Prop.~4.1]{a5}. 
However, the antisymmetry  of the form ${\mathfrak A}^\rho_{\D_{\omega_\rho}}$ 
is
not guaranteed, unless one restricts to the positive eigenspace of ${\cal R}$, that is
\begin{equation}
\label{eq:2.16}
	\mathcal{H_R} := \{ \xi \in \mathsf{Dom}\,\D\, \quad {\cal R}\xi = \xi \}.
\end{equation}
This has been discussed in~\cite[Prop.~4.2]{a5} and led to the following:
\begin{definition}
	For a real twisted spectral triple $(\A, \HH, \D;J)_\rho$, the fermionic action
	is
\begin{equation}
	S_\rho(\D_{\omega_\rho}) 
	:= \frak A^\rho_{\D_{\omega_\rho}}\!(\tilde \xi, \tilde \xi),
\label{eq:10}
\end{equation}
where $\tilde \xi$ is the Gra{\ss}mann vector associated to $\xi \in {\cal H_R}$.
\end{definition}

In the spectral triple of SM, the restriction to ${\cal H}_+$ is there to solve the 
fermion doubling problem~\cite{LMMS97}. It also selects out the physically meaningful elements
of ${\cal H}=L^2(\M,{\cal S})\otimes \HH_{\cal F}$, 
that is, those spinors whose chirality in $L^2(\M,{\cal S})$ coincides 
with their chirality as elements of the finite-dimensional
Hilbert space $\HH_{\cal F}$.
In the twisted case, the restriction to ${\cal H_R}$ is there to guarantee the 
antisymmetry of the bilinear form $\frak A^\rho_{\D_{\omega_\rho}}$. 
However, the eigenvectors of $\cal R$ may
not have a well-defined chirality. If fact, they cannot have it when
the twist comes from the grading (see \S \ref{sec:2.2} below), since the unitary ${\cal R}$ 
implementing the twist (given in \eqref{eq:12})
anticommutes with the chirality $\Gamma=\text{diag}\,({\mathbb I}_{\HH_+}, -{\mathbb I}_{\HH_-})$, 
so that
\begin{equation}
	\HH_+\cap \HH_{\cal R} = \left\{0\right\}\!.
 \label{eq:11}
\end{equation}

From a physical standpoint, by restricting to $\HH_{\cal R}$ rather than
$\HH_+$, one loses a clear interpretation of the elements of the Hilbert space: a
priori, an element of $\HH_{\cal R}$ is not physically meaningful, since its
chirality is ill-defined. However, we show in what follows that -- at least in two examples: 
a manifold and the almost-commutative geometry of electrodynamics -- the 
restriction to ${\cal H_{\cal R}}$ is actually meaningful, for it allows to obtain the Weyl and Dirac equations in the 
lorentzian signature, even though one starts with a riemannian manifold.
\smallskip

Before that, we conclude this section with
two easy but useful lemmas. The first recalls how the
symmetry properties of the bilinear form ${\frak A}_D = \langle J\cdot,
D\cdot\rangle$ do not depend on the
explicit form of the Dirac operator, but solely on the signs
$\epsilon', \epsilon''$ in~(\ref{eq:KOdim}). The second stresses that
once restricted to $\HH_{\cal R}$, the bilinear forms \eqref{eq:4}
and \eqref{Sfrho} differ only by a sign.

\begin{lemma}
\label{lemma:antisymm}
Let $J$ be an antilinear isometry on the Hilbert space 
$(\HH,\langle \cdot,\cdot\rangle)$ such that $J^2=\epsilon \mathbb I$, and
 $D$ a self-adjoint operator on $\HH$ such that $JD = \epsilon'DJ$. Then
\begin{equation}
\langle J\phi, D \xi\rangle = \epsilon\epsilon' \langle J\xi, D
\phi\rangle, \qquad \forall \phi, \xi\in\HH.
\label{eq:15}
\end{equation}
\end{lemma}
\begin{proof}
 By definition, an antilinear isometry satisfies $\langle J
  \phi, J \xi\rangle=\overline{ \langle \phi, \xi\rangle} = \langle
  \xi, \phi \rangle$. Thus,
  \begin{equation*}
    \langle J
  \phi, D \xi\rangle =  \epsilon \langle J
  \phi, J^2 D \xi\rangle =   \epsilon\langle J D \xi,
  \phi\rangle=\epsilon\epsilon'\langle D  J\xi,
  \phi\rangle=\epsilon\epsilon'\langle  J\xi,
  D\phi\rangle. \qedhere
  \end{equation*}
\end{proof}
\noindent In particular, for $KO$-dimension $2,4$ one has $\epsilon=-1,
\epsilon'=1$, so $\frak A_{\D}$ is antisymmetric. The same is true for
$\frak A_{\D_\omega}$ in \eqref{eq:4}, because the covariant operator $\D_\omega$ satisfies the same rules of sign (\ref{eq:KOdim})~as~$\D$.
\begin{lemma}
\label{lem:bilinearform}
	Given $D$, and a unitary $\cal R$ compatible with $\J$ in the
        sense of \eqref{eq:2.11}, one has 
\begin{equation}
	{\mathfrak A}_D^\rho(\phi, \xi)
	= \epsilon'''\,{\mathfrak A}_D(\phi, \xi),
	\qquad \forall \phi, \xi \in {\cal H_R}.
\end{equation}
\end{lemma}

\begin{proof}
	For any $\phi, \xi \in {\cal H_R}$, we have	
\begin{equation}
	{\mathfrak A}_D^\rho(\phi, \xi)
	= \langle J\phi, {\cal R}D\xi \rangle
	= \langle {\cal R}^\dag J\phi, D\xi \rangle
	= \epsilon'''\langle J{\cal R}^\dag\phi, D\xi \rangle
	= \epsilon'''\langle J\phi, D\xi \rangle,
\end{equation}
where we used \eqref{eq:2.11} as
${\cal R}^\dag J = \epsilon'''J{\cal R}^\dag$ and 
\eqref{eq:2.16} as ${\cal R}^\dag\phi = \phi$.
\end{proof}

\subsection{Minimal twist by grading}
\label{sec:2.2}

The twisted spectral triples recently employed  in physics are built by minimally twisting a usual
spectral triple $(\A, \HH, \cal D)$. The idea is
to substitute   the commutator $[{\cal D},\cdot]$ with a twisted one $[{\cal D},\cdot]_{\rho}$, while keeping the Hilbert space and the Dirac operator
intact, because they encode the fermionic content of the theory and
there is, so far, no experimental indications of extra fermions
beyond those of the SM. 
However,
for the spectral triples relevant for physics,
$[{\cal D},\cdot]$ and  $[{\cal D},\cdot]_{\rho}$ cannot be
simultaneously bounded \cite[\S 3.1]{LM1}. So in order to be able to
twist the commutator, one needs to play with the only object
that remaining available, namely the algebra.

\begin{definition}[from~\cite{LM1}]
\label{defi:minimaltwist}
	A \emph{minimal twist} of a spectral triple $(\cal A, H, D)$ by
	a unital $*$-algebra $\mathcal{B}$ is a twisted spectral triple 
	$(\cal A\otimes B, H, D)_{\rho}$ where the initial 
	representation $\pi_0$ of $\cal A$ on $\cal H$ is related to the 
	representation $\pi$ of $\cal A \otimes \cal B$ on $\cal H$ 
	by
\begin{equation}
	\pi(a \otimes \mathbb{I}_{\mathcal{B}}) = \pi_0(a), \qquad 
	\forall a \in {\cal A}
\end{equation}
where $\mathbb{I}_{\mathcal{B}}$ is the identity of the algebra $\cal B$.
\end{definition}

If  the initial spectral triple is graded, 
a natural minimal twist may be obtained as follows. 
The grading  $\Gamma$ commutes with the representation of $\A$, so the latter
is actually a direct sum of two representations on the positive and
negative eigenspaces 
$\HH_+$, $\HH_-$  of $\Gamma$ (see \eqref{eq:H_+}). Therefore, one has enough space 
on $\HH=\HH_+\oplus\HH_-$ to represent twice the algebra $\A$. It is tantamount to taking ${\cal B} = 
{\mathbb C}^2$ in Def.~\ref{defi:minimaltwist}, with ${\cal
  A}\otimes\mathbb{C}^2 \simeq \A\oplus\A \ni (a,a')$ represented on
${\cal H}$ as  
\begin{equation} 
\label{repaa'}
	\pi(a,a') := {\frak p}_+\pi_0(a) + {\frak p}_-\pi_0(a') =
	\left( \begin{array}{cc}
		\pi_+(a) & 0 \\ 0 & \pi_-(a')
	\end{array} \right)\!,
\end{equation}
where
	${\frak p}_{\pm} := \frac{1}{2} \left( \mathbb{I}_\HH \pm \Gamma
        \right)$ and 
	$\pi_{\pm}(a) := \pi_0(a)_{|{\cal H}_{\pm}}$
are respectively  the projections on ${\cal H}_{\pm}$ and the restrictions on 
${\cal H}_{\pm}$ of $\pi_0$. If $\pi_{\pm}$ are faithful, then $({\cal
  A}\otimes\mathbb{C}^2, {\cal H, D})_{\rho}$ with $\rho$ the flip automorphism
\begin{equation}
\label{eq:flip}
	\rho(a,a') := (a',a), \qquad 
	\forall (a,a') \in {\cal A}\otimes\mathbb{C}^2,
\end{equation}
is indeed a twisted spectral triple, with grading $\Gamma$. Furthermore,
if the initial spectral triple is real, then so is this minimal twist, with the 
same real structure \cite{LM1}.{\footnote{The requirement
    that $\pi_\pm$ are faithful was not explicit in \cite{LM1}. If it
    does not hold, then $(\A\otimes\mathbb C^2, \HH, D)_\rho$ still
    satisfies all the properties of a twisted spectral triple, except
    that $\pi$ in (\ref{repaa'}) might not be faithful.}}

	The flip $\rho$ is a $*$-automorphism that 
        satisfies \eqref{eq:17}, and coincides on 
$\pi(\A\otimes\mathbb C^2)$ with the inner automorphism of ${\cal B}(\HH)$
implemented by the unitary
\begin{equation}
\label{eq:12}
	{\cal R} =
	\left( \begin{array}{cc}
		0 & {\mathbb I}_{\HH_+} \\ 
		{\mathbb I}_{\HH_-} & 0
	\end{array} \right)\;
\text{ with }\;
\mathbb{I}_{\HH_{\pm}} \text{ the identity operator in } \HH_{\pm}.
\end{equation}

As recalled in the next section, the canonical $\rho$-product~\eqref{rho-p} associated to the minimal twist of a closed \emph{riemannian} spin manifold
of dimension $4$ turns out to coincide with the \emph{lorentzian} 
Krein product on the space of lorentzian spinors~\cite{a5}. The aim of
this paper is to show that a similar transition from the euclidean to
the lorentzian also occurs for the fermionic action. 

\smallskip

We first investigate how this idea comes about, by studying in the
next section the
simplest example of  the minimal twist of a manifold. Then, in the
following sections, we show how to obtain the Weyl equations in the lorentzian
signature by
doubling the twisted manifold and, finally, the Dirac equation by minimally
twisting the spectral triple of electrodynamics in~\cite{W1}. 

\section{Preliminary: minimally twisted manifold}
\label{sec:Weyl}
\setcounter{equation}{0} 

We compute the fermionic action for the minimal twist of a
closed euclidean spin manifold~$\M$. Since we aim at finding back  the
Weyl and Dirac equations, we work in dimension $4$, assuming gravity
is negligible (hence the flat metric). 
 This is tantamount to choosing in (\ref{eq:KOdim})
    \begin{equation}
      \label{eq:22}
      \epsilon=-1,\quad \epsilon'=1,\quad \epsilon''=1.
    \end{equation}

\subsection{Minimal twist of a riemannian manifold}
\label{subsec:minimaltwistriema}
The minimal twist of $\cal M$ is the real, graded,
twisted  spectral triple
\begin{equation}
\label{mnfld}
	\left(C^{\infty}(\mathcal{M})\otimes\mathbb{C}^2, \;
	 L^2(\mathcal{M,S}), \; \eth \right)_{\rho}
\end{equation}
where $C^{\infty}(\mathcal{M})$ is the algebra of smooth functions on 
$\mathcal{M}$, $L^2(\mathcal{M,S})$
is the Hilbert space of square integrable spinors
with inner product ($\dmu$ the volume form)
\begin{equation}
\label{eq:12bis}
	\langle \psi, \phi \rangle = \int_{\cal M}\dmu \; \psi^\dag \phi,
	\qquad \text{for} \quad \psi, \phi \in L^2(\mathcal{M,S});
\end{equation}
 and
	$\eth := -i\gamma^\mu \partial_{\mu}$
	is the euclidean Dirac operator
        with 
$\gamma^\mu$ the self-adjoint euclidean Dirac 
matrices (see \eqref{EDirac}).
The real structure
and grading
are ($cc$ denotes complex conjugation)
\begin{equation}
\label{eqn:3.9}
	{\cal J} = i\gamma^0\gamma^2 cc
	= i\!\left(\begin{array}{cc}
		\tilde\sigma^2  & 0 \\ 0 & \sigma^2
	\end{array}\right)\!cc,
	\qquad \gamma^5 = \gamma^1\gamma^2\gamma^3\gamma^0
	= \left( \begin{array}{cc}
		\mathbb{I}_2 & 0 \\ 0 & -\mathbb{I}_2
	\end{array} \right)\!.
\end{equation}

\smallskip

The 
representation~\eqref{repaa'} of $C^\infty(\M)\otimes\mathbb C^2$ on
$L^2(\M,{\cal S}) = L^2(\mathcal{M,S})_+\oplus L^2(\mathcal{M,S})_-$ is 
\begin{equation}
\label{eq:piM}
	\pi_{\mathcal{M}}(f,f')
	= \left( \begin{array}{cc}
		f\,\mathbb{I}_2 & 0 \\ 0 & f'\,\mathbb{I}_2
	\end{array} \right)\!,
\end{equation}
where each of the two copies of
$C^{\infty}(\mathcal{M})$ acts independently and faithfully by point-wise
multiplication on the eigenspaces 
$L^2(\mathcal{M,S})_{\pm}$ of $\gamma^5$.
The automorphism $\rho$ of $C^{\infty}(\mathcal{M})\otimes\mathbb C^2$ is the~flip
\begin{equation}
	\rho(f,f') = (f',f), \qquad \forall f,f' \in C^{\infty}(\mathcal{M}).
\end{equation}
It coincides with the inner automorphism of ${\cal B}(\HH)$ implemented by the unitary 
\begin{equation}
\label{eq:definR}
{\cal R}=
\left( \begin{array}{cc}
	0 & \mathbb{I}_2 \\ \mathbb{I}_2 & 0
\end{array} \right)\!,
\end{equation}
which is nothing but the Dirac matrix $\gamma^0$
 (this choice is not unique, as will be investigated~in~\cite{Besnard:2020aa}). 
It is 
compatible with the real structure \eqref{eq:2.11}~with
\begin{equation}
  \label{eq:31}
  \epsilon'''=-1.
\end{equation}
\begin{lemma}
\label{lemma:commgamma}
 For any $a=(f, f')\in C^{\infty}(\mathcal{M})\otimes{\mathbb C}^2$ and $\mu=0,1,2,3$, one has  
 \begin{equation}
 \label{eq:commgamma}
   \gamma^\mu a = \rho(a) \gamma^\mu, \quad \gamma^\mu\rho(a) =
   a\gamma^\mu, \quad
  \gamma^\mu \J = -\epsilon'\J\gamma^\mu.
 \end{equation}
\end{lemma}
\begin{proof}
  The first equation is checked by direct calculation, 
  using the
  explicit form of $\gamma^\mu$, along with \eqref{eq:piM}
  and (writing $\rho(a)$ for $\pi_\M(\rho(a))$):
  \begin{equation}
\label{eq:twistmanifold}    \rho(a)=  \left( \begin{array}{cc}
		f'\,\mathbb{I}_2 & 0 \\ 0 & f\,\mathbb{I}_2
	\end{array} \right)\!.
  \end{equation}
The second follows from \eqref{eq:17}, the third
from \eqref{eq:KOdim}, noticing that $\cal J$ commutes with
$\partial_\mu$, having constant components:
\begin{equation*}
 0=\J\eth - \epsilon' \eth\J =i \left(\J \gamma^\mu +
    \epsilon'\gamma^\mu \J\right)\partial_\mu. 
\end{equation*}

\vspace{-.7truecm}
\end{proof}
\begin{corollary} 
\label{cor:bound}
The boundedness of the twisted
 commutator follows immediately:
  \begin{equation}
    \label{eq:test1}
    \left[\eth, a\right]_\rho 
    =-i\left( \gamma^\mu\partial_\mu a 
      -\rho(a)\gamma^\mu\partial_\mu \right) 
    =-i\gamma^\mu \left[ \partial_\mu, a \right]
    = -i\gamma^\mu(\partial_\mu a) \quad\quad \forall a\in C^\infty(\M)\otimes
    \mathbb C^2.
  \end{equation}
\end{corollary}

\subsection{Twisted fluctuation for a manifold}

Substituting, in a twisted spectral triple, $\D$ with
the twisted covariant $\D_{\omega_\rho}$ \eqref{eqw:twistfluct} is called a
\emph{twisted fluctuation}.
The
minimally twisted manifold \eqref{mnfld}
has non-vanishing self-adjoint twisted 
fluctuations~\eqref{eqw:twistfluct} of the form
\begin{equation}
\label{eq:dflucM}
	\eth_{\bf X} := \eth + {\bf X},
\end{equation}
where
\begin{equation}
\label{Xmu}
{\bf X} := -i\gamma^\mu X_{\mu}, \quad \text{ with } \quad X_{\mu} :=
f_\mu \gamma^5, \quad \text{ for some } f_\mu \in C^\infty({\cal M},\mathbb{R}).
\end{equation}
This has been shown in~\cite[Prop.~5.3]{LM1}; in contrast with the
non-twisted case, where the self-adjoint fluctuation of $\eth$
always vanishes, irrespective of the dimension of the manifold 
$\cal M$~\cite{Connes:1996fu}. 
\smallskip

In \cite{LM1} the self-adjointness of $\eth_{\bf X}$ was guaranteed
imposing the selfadjointness of $\omega_\rho + {\cal J}\omega_\rho{\cal
  J}^{-1}$, but not necessarily the one of 
$\omega_\rho$. One may worry that the
non-vanishing of $\bf X$ is an artefact of this choice, 
and that $\bf X$ might actually vanish as soon as 
$\omega_\rho=\omega_\rho^\dag$. The following lemma clarifies this point.

\begin{lemma} 
\label{lem:3.1} 
The twisted one-forms $\omega_\rho$ \eqref{eq:3}
and the twisted fluctuations $\omega_\rho + \J\omega_\rho \J^{-1}$ of a
minimally twisted four-dimensional closed euclidean manifold are all of the kind
    	\begin{align}
	\label{eq:omrgo}	
& \omega_\rho =  -i\gamma^\mu W_\mu, & \text{ with } & \quad W_\mu = {\sf diag} \left( h_\mu\mathbb{I}_2, \, h'_\mu\mathbb{I}_2 \right)\!, \\
\label{eq:omrgo2}
& \omega_\rho + \J\omega_\rho \J^{-1} = -i\gamma^\mu X_\mu, & \text{ with } & \quad X_\mu = {\sf diag} \left( f_\mu \mathbb{I}_2, \, f'_\mu \mathbb{I}_2 \right)\!,
	\end{align}
where $h_\mu, \, h'_\mu\in C^\infty(\M)$, $f_\mu = 2\Re h_\mu$ and 
$f'_\mu = 2\Re h'_\mu$. They are self-adjoint, respectively, iff
\begin{align}
  \label{eq:29}
  h'_\mu =-\bar h_\mu, \quad {and} \quad 
   f'_\mu = - f_\mu.
\end{align}
\end{lemma}

\begin{proof}
By Lem. \ref{lemma:commgamma} and its corollary, one obtains  for $a_i:=(f_i,f'_i), b_i:=(g_i,g'_i)\in C^{\infty}(\mathcal{M})\otimes{\mathbb C}^2$, 
\begin{equation*}
\omega_\rho = \sum_i
        b_i[\eth,a_i]_\rho = -i \gamma^\mu \sum_i
        \rho(b_i)(\partial_\mu a_i) = -i\gamma^\mu \sum_i\left(\begin{array}{cc} 
	g'_i\mathbb{I}_2 \!\!\!&\!\!\! 0 \\ 0\!\!\! &\!\!\! g_i\mathbb{I}_2 
\end{array} \!\!\right)\left(\!\! 
\begin{array}{cc} 
		(\partial_\mu f_i)\mathbb{I}_2 \!\!\!&\!\!\! 0 \\ 0\!\!\! &\!\!\! (\partial_\mu f_i')\mathbb{I}_2 
\end{array} \!\!\right),
	\end{equation*}
which is of the form \eqref{eq:omrgo} with
$h_\mu := \sum_i g_i'(\partial_\mu f_i)$ and 
$h'_\mu := \sum_i g_i(\partial_\mu f_i')$. 
	The adjoint  is
	\begin{equation}
		\omega^\dag_\rho = iW_\mu^\dag\gamma^\mu  
		= i\gamma^\mu \rho(W_\mu^\dag),
	\end{equation}
where the last equality follows from (\ref{eq:commgamma}), applied to
$W_\mu$ viewed as an element of
$C^{\infty}(\mathcal{M})\otimes{\mathbb C}^2$. Thus,
$\omega_\rho$ is self-adjoint iff $\gamma^\mu\rho(W_\mu^\dag)
        = -\gamma^\mu W_\mu$,
        that is, going back to the explicit form of $\gamma^\mu$,
        \begin{equation}
\label{eq:gsym}
          \sigma^\mu \bar h_\mu = -\sigma^\mu h'_\mu,\quad\text{ and }\quad
 \tilde\sigma^\mu \bar h'_\mu = -\tilde \sigma^\mu h_\mu.
        \end{equation}
Multiplying the first equation by $\sigma^\lambda$ and using 
${\sf Tr}(\sigma^\lambda\sigma^\mu) = 2\delta^{\mu\lambda}$ yields the
first part of  \eqref{eq:29}. Obviously the
latter implies  both eqs. \eqref{eq:gsym}. Hence,
		$\omega_\rho = \omega_\rho^\dag$ is equivalent to the
                first eq.~\eqref{eq:29}.

Further, we have
	\begin{equation}
		{\cal J}\omega_\rho{\cal  J}^{-1} = {\cal J}(-i\gamma^\mu  W_\mu){\cal J}^{-1} 
	=  i{\cal J}(\gamma^\mu  W_\mu){\cal J}^{-1} 
	=-i\gamma^\mu {\cal J}W_\mu{\cal J}^{-1} = -i\gamma^\mu  W_\mu^\dag,
	\end{equation}
	using 
       ${\cal J}\gamma^\mu  = -\gamma^\mu
        {\cal J}$  (from \eqref{eq:22} and \eqref{eq:commgamma}), 
        along with $\J W^\mu= W_\mu^\dag \J$
(from (\ref{eqn:3.9})
         and the explicit form (\ref{eq:omrgo}) of $W_\mu$).
Therefore, 
	\begin{equation}
		\omega_\rho + {\cal J}\omega_\rho {\cal J}^{-1}
		= -i\gamma^\mu (W_\mu + W_\mu^\dag),
	\end{equation}
	which is nothing but (\ref{eq:omrgo2}), identifying
$X_\mu := W_\mu + W_\mu^\dag = {\sf diag}
		\left( (h_\mu + \bar h_\mu)\mathbb{I}_2,
			( h'_\mu + \bar h'_\mu)\mathbb{I}_2 \right)\!$.
One checks as above that $\omega_\rho + \J \omega_\rho \J^{-1}$ is
self-adjoint iff the second equation of \eqref{eq:29} holds.
\end{proof}

\noindent Consequently, imposing that $\omega_\rho\neq 0$ be self-adjoint, that is
imposing \eqref{eq:29} with $h_\mu\neq 0$, does not imply that $X_\mu$
vanishes (it does vanish only if $h_\mu$ is purely imaginary). In
other words,  as long as $h_\mu\notin i{\mathbb  R}$, the 
self-adjointness of $\omega_\rho$ does not forbid a non-zero twisted
fluctuation.

\subsection{Gauge transformation}
\label{sec:3.1.2}

For a minimally twisted manifold, not only is the fermionic action \eqref{eq:10}
invariant under a gauge transformation \eqref{eq:20},\eqref{eq:9},
but so  is the operator 
${\cal D}_{\omega_\rho}$ (in dimension $0,4$)
\cite[Prop.~5.4]{LM2}. We check it explicitly by studying how
 the field $h_\mu$ parametrising $\omega_\rho$ in \eqref{eq:omrgo} transforms.
\smallskip

A unitary of $C^\infty(\M) \otimes\mathbb C^2$ is
$u:=(e^{i\theta},e^{i\theta'})$ with $\theta, \theta'\in
C^\infty(\M,\mathbb R)$. It (and its twist) acts on $\HH$
according to \eqref{eq:piM} as (we omit the symbol of
representation)  
\begin{equation}
\label{eq:pi_M(u)}
	u = 
	\left( \begin{array}{cc}
		e^{i\theta}\mathbb{I}_2 & 0 \\ 0 & e^{i\theta'}\mathbb{I}_2
	\end{array} \right)\!, \qquad \rho(u) =
	\left( \begin{array}{cc}
		e^{i\theta'}\mathbb{I}_2 & 0 \\ 0 & e^{i\theta}\mathbb{I}_2
	\end{array} \right)\!.
\end{equation}

\begin{proposition}
	Under a gauge transformation with unitary $u \in
        C^{\infty}(\mathcal{M})\otimes\mathbb C^2$, the fields 
	$h_\mu$ and $h'_\mu$ parametrising the twisted one-form $\omega_\rho$ in 
	\eqref{eq:omrgo} transform as
	\begin{equation}
		h_\mu \to h_\mu -i\partial_\mu\theta, \qquad
		h'_\mu \to h_\mu' -i\partial_\mu\theta'\!.
	\end{equation}
\end{proposition}
\begin{proof}
Under a gauge transformation,  the twisted one-form $\omega_\rho$ is
mapped to (see \eqref{LM2Prop4.2})
	\begin{equation*}
	\begin{split}
		\omega^u_\rho 
		& = -i\rho(u) \left( [\gamma^\mu \partial_\mu,u^*]_\rho + 
		\gamma^\mu  W_\mu u^*\right)
	 = -i\rho(u)\gamma^\mu (\partial_\mu + W_\mu)u^*
		= -i\gamma^\mu (u\partial_\mu u^* + W_\mu),
	\end{split}
	\end{equation*}
	where we used \eqref{eq:test1} for $a=u^*$, namely
	\begin{equation}
	\label{Id1}
		[\gamma^\mu \partial_\mu,u^*]_\rho 
     = \gamma^\mu (\partial_\mu u^*),
	\end{equation}
as well as (\ref{eq:commgamma}) for $a=u$, together with
$uW_\mu u^*=W_\mu$ since $u$ commutes with $W_\mu$.
Therefore, $W_\mu \to W_\mu + u\partial_\mu u^*$,
	which with the explicit representation of $W_{\mu}$~\eqref{eq:omrgo} and 
	$u$~\eqref{eq:pi_M(u)} reads
	\begin{equation*}
		\left( \begin{array}{cc}
			h_\mu\mathbb{I}_2 & 0 \\ 0 & h_\mu'\mathbb{I}_2
		\end{array} \right) \longrightarrow
		\left( \begin{array}{cc}
			(h_\mu-i\partial_\mu\theta)\mathbb{I}_2 & 0 \\
			0 &  (h_\mu'-i\partial_\mu\theta')\mathbb{I}_2
		\end{array} \right)\!,
	\end{equation*}

\vspace{-.5truecm}
\end{proof}

Although $h_\mu$, $h'_\mu$ transform in a nontrivial
        manner, their real parts $\frac 12 f_\mu, \frac 12 f'_\mu$ remain invariant.
 This explains why the
fluctuation $\bf X$ in \eqref{eq:omrgo} is invariant
under a gauge transformation \eqref{eq:9}. 
Furthermore, by simultaneously transforming spinors according to
\eqref{eq:20}}, the twisted fermionic action is invariant, by
construction. So one expects
that any $\psi\in{\cal H}_R$ is unchanged under the adjoint
action of $\text{Ad} \,u$. This is true, as one
checks from \eqref{eqn:3.9} that $u{\cal J}u
{\cal J}^{-1}=\mathbb I$ for any unitary~$u$. 

\subsection{Twisted fermionic action for a manifold}
\label{subsec:fermioncactionmanif}

Let us 
first work out 
the positive eigenspace $\HH_{\cal R}$ \eqref{eq:2.16}  for ${\cal
  R}=\gamma^0$ as in \eqref{eq:definR}.
\begin{lemma}
\label{lem3.1}
	An eigenvector $\phi \in {\cal H_R}$
is of the form
$\phi := \begin{pmatrix} \varphi \\ \varphi \end{pmatrix}$
where
$\varphi$ is a Weyl spinor.
\end{lemma}
\begin{proof}
The
$+1$-eigenspace of $\gamma^0$ is spanned by
	$\upsilon_1 = \left( \begin{smallmatrix} 1 \\
            0 \end{smallmatrix}\right)\! \otimes\left( \begin{smallmatrix} 1 \\ 1 
	\end{smallmatrix} \right) $, 
	$\upsilon_2 = \left( \begin{smallmatrix} 0 \\ 1 
	\end{smallmatrix} \right)\otimes\left( \begin{smallmatrix} 1 \\ 1 
	\end{smallmatrix} \right)$. Therefore, a generic vector $\phi
      = \phi_1\upsilon_1 + \phi_2\upsilon_2$ in
      ${\cal H_R}$ is as in the lemma,	
 with $\varphi := \begin{pmatrix} \phi_1 \\ \phi_2 \end{pmatrix}$.
\end{proof}

\smallskip

We now compute the fermionic action
\eqref{eq:10} for a minimally twisted manifold.

\smallskip

\begin{proposition}
\label{prop:actspecmanif}
	Let $\eth_{\bf X}$ be the twist-fluctuated Dirac operator \eqref{eq:dflucM}.  
The bilinear form \eqref{Sfrho} restricted to ${\cal H}_{\cal R}$
(antisymmetric by lemma \ref{lem:bilinearform}) is
\begin{equation}
\label{3.1}
{\frak A}^\rho_{\eth_X}(\phi, \xi) 
	= 2\int_{\cal M}\!\!\! \emph{d}\upmu\;   \bar\varphi^\dag\sigma_2
	\left( if_0 - \textstyle\sum_{j=1}^3\sigma_j\partial_j \right) 
			\zeta ,
\end{equation}
where $\varphi$, $\zeta$ are, respectively, the Weyl components of the Dirac 
spinors $\phi$, $\xi\in{\cal H}_{\cal R}$, and $f_0$ is the zeroth
component of the twisted fluctuation $f_\mu $ in 
\eqref{Xmu}.
\end{proposition}

\begin{proof}
	One has
\begin{align}
\label{eq:JPhi}
	{\cal J} \phi 
& 	= i\gamma^0\gamma^2 \, cc \!
	\left( \begin{array}{c} \varphi \\ \varphi \end{array} \right)
	= i\!\left( \begin{array}{cc} 
		\tilde\sigma^2 & 0 \\ 0 & \sigma^2
	\end{array} \right) \!\!
	\left( \begin{array}{c} \bar \varphi \\ \bar \varphi \end{array} \right)
	= i\!\left( \begin{array}{cc}
		\tilde\sigma^2\,\bar \varphi \\ \sigma^2\,\bar\varphi
	\end{array} \right)\!, \\[5pt]
\label{eq:DPhi}
	\eth \xi 
&	= -i\gamma^\mu \partial_\mu \!
	\left( \begin{array}{c} \zeta \\ \zeta \end{array} \right)
	= -i\!\left( \begin{array}{cc}
		0 & \sigma^\mu \\ \tilde\sigma^\mu & 0
	\end{array} \right) \!\!
	\left( \begin{array}{c}	
		\partial_\mu\zeta \\ \partial_\mu\zeta
	\end{array}\right)
	= -i\!\left(\begin{array}{c}
		\sigma^\mu\partial_\mu\zeta \\ \tilde\sigma^\mu\partial_\mu\zeta
	\end{array}\right)\!, \\[5pt]
\label{eq:Xxi}
	{\bf X} \xi 
&	= -i\gamma^\mu X_{\mu} \!
	\left( \begin{array}{c} \zeta \\ \zeta \end{array} \right)
	= -i\!\left( \begin{array}{cc}
		0 & \sigma^\mu \\ \tilde\sigma^\mu & 0
	\end{array} \right) \!\!
	\left( \begin{array}{cc}
		f_\mu\mathbb{I}_2 & 0 \\ 0 & -f_\mu\mathbb{I}_2
	\end{array} \right) \!\!
	\left( \begin{array}{c} \zeta \\ \zeta \end{array} \right)
	= -i\!\left( \begin{array}{r}
		-f_\mu\sigma^\mu\zeta \\ f_\mu\tilde\sigma^\mu\zeta
	\end{array} \right)\!.
\end{align}
Hence, noticing that  $(\tilde\sigma^2)^\dag =-i\sigma_2$ and
${\sigma^2}^\dag=i\sigma_2$ (see Appendix~\ref{appA}), and using 
\begin{equation}
\label{eq:sumsigma}
\sigma^\mu + \tilde\sigma^\mu = 2{\mathbb
  I}_2 \delta^{\mu 0},\qquad
\sigma^\mu - \tilde\sigma^\mu = -2i \delta^{\mu j} \sigma_j,
\end{equation}
one gets
\begin{align}
 \label{derniere1}
\frak A_{\eth}(\phi, \xi) =	\langle {\cal J} \phi, \eth \xi
  \rangle =
& -\left(\begin{array}{cc}
		\bar\varphi^\dag \tilde\sigma^{2\dag} ,
		\bar\varphi^\dag \sigma^{2\dag}
                \end{array} \right.) 
	\left(\begin{array}{c}
		\sigma^\mu \partial_\mu \zeta \\
		\tilde\sigma^\mu\partial_\mu \zeta
	\end{array}\!\! \right)\!\\
 &= i\int_{\cal M}\!\!\dmu\,\bar\varphi^\dag \sigma_2
(\sigma^\mu- \tilde\sigma^\mu) \partial_\mu \zeta = 2 \int_{\cal M}\!\!\dmu\,\bar\varphi^\dag \sigma_2
	 \underset{j=1}{\overset{3}{\textstyle\sum}}\sigma_j\partial_j \zeta;\\[4pt]
\label{derniere2}
\frak A_{\bf X}(\phi, \xi)=\langle {\cal J} \phi, {\bf X} \xi \rangle
&=
 -\left(\!\begin{array}{cc}
		\bar\varphi^\dag \tilde\sigma^{2\dag} ,
		\bar\varphi^\dag \sigma^{2\dag}
                \end{array}\!\!\! \right.) 
	\left(\!\!\begin{array}{c}
		-f_\mu\sigma^\mu \zeta \\
		f_\mu\tilde\sigma^\mu\ \zeta
	\end{array}\!\! \right)\!\\
 &= -i\int_{\cal M}\!\!\dmu\,\bar\varphi^\dag \sigma_2 f_\mu
		(\sigma^\mu+ \tilde\sigma^\mu) \partial_\mu \zeta = -2i \int_{\cal M} \!\!\emph{d}\upmu 
	 f_0 \bar\varphi^\dag \sigma_2 \,\zeta,
\end{align} 
From Lem.~\ref{lem:bilinearform} and \eqref{eq:31} follows
\begin{equation}
\label{eq:arho}
	{\frak A}^\rho_{\eth_{\bf X}}(\phi, \xi)
	= -	{\frak A}_{\eth_{\bf X}}(\phi, \xi)= -	{\frak
          A}_{\eth}(\phi, \xi) -{\frak A}_{\bf X}(\phi, \xi).
\end{equation}
Hence the result.
\end{proof}

The fermionic action is then obtained by substituting $\phi=\xi$ in
\eqref{3.1} and replacing the components $\zeta$ of $\xi$ by the
associated Gra{\ss}mann variable $\tilde \zeta, \tilde\varphi$:
\begin{equation}
  \label{eq:56}
  S_\rho(\eth_{\bf X}) =
{\frak A}^\rho_{\eth_{\bf X}}(\tilde\xi, \tilde\xi)   = 2\int_{\cal M} \dmu \left[ \tilde{\bar\zeta}^\dag\sigma_2
	\left(if_0
	- \sum_{j=1}^3\sigma_j\partial_j \right) \tilde\zeta \right]\!.
\end{equation}

The striking fact about \eqref{eq:56}
is the disappearance of the derivative in the $x_0$ direction,
and the appearance, instead,  of the zeroth component  of the real field $f_\mu$
parametrising the twisted fluctuation ${\bf X}$. This derivative, however,
can be restored interpreting $-if_0\zeta$ as $\partial_0\zeta$,
i.e.~assuming
\begin{equation}
  \label{eq:33}
  \zeta(x_0, x_i) = \exp(- if_0 x_0) \, \zeta(x_i)
\end{equation}
with $f_0$ independent of $x_0$. Denoting by $\sigma^\mu_M=\{\mathbb I_2, \sigma_j\}$ the
upper-right components of the minkowskian Dirac matrices (see
\eqref{MDirac}), the integrand in the fermionic action then reads
(with summation on the $\mu$ index)
\begin{equation}
-\tilde{\bar\zeta}^\dag \sigma^2_M
\left( \sigma^\mu_M\partial_\mu\right)\tilde\zeta,
\label{eq:034}
\end{equation}
which reminds of the Weyl lagrangian densities~\eqref{eq:Weyl}:
\begin{equation}
{\cal L}_M^r = i\Psi_r^\dag \left(
 {\sigma}_M^{\mu}\partial_{\mu} \right)\Psi_r, 
\end{equation}
but with the $\sigma_M^2$ matrix, that prevents to simultaneously identify $\tilde\zeta$ with
$\Psi_r$ and $-\tilde{\bar\zeta}^\dag\sigma^2_M$ with~$i\Psi_r^\dag$. 

To make such an
identification 
possible, one needs more spinorial degrees of freedom. They are
obtained in the next section, multiplying the manifold by a two point space.

\section{Doubled manifold and Weyl equations}
\label{sec:doublem}
\setcounter{equation}{0} 

In constructing a spectral triple for electrodynamics, the authors 
of~\cite[\S3.2]{W1} first consider, as an intermediate step, the product of a 
manifold with the finite-dimensional spectral triple
\begin{equation}
\A_{\cal F} = {\mathbb C}^2,\quad
\HH_{\cal F} = \mathbb C^2,\quad
D_{\cal F} = 0.\label{eq:35}
\end{equation}
This model describes a 
$U(1)$ gauge theory, but 
fails to describe classical electrodynamics for two reasons,
 discussed at the end of~\cite[\S 3]{W1}: 
first,  the
finite Dirac operator is zero so  the electrons are 
	massless; second,  ${\cal H_F}$ is not big enough 
	to capture the required spinor degrees-of-freedom. 

However, none of the above arises as an issue if one wishes to obtain the Weyl 
lagrangian, since the Weyl fermions are massless anyway, and they only need half of 
the spinor degrees-of-freedom as compared to the Dirac fermions.

\subsection{Minimal twist of a two-point almost-commutative geometry}
\label{subsec:minimaltwistdoublema}

The product -- in the sense of spectral triple -- of a four-dimensional closed euclidean manifold $\M$
with the two-point space \eqref{eq:35} is
\begin{equation}
\label{eq:MxF_X}
{\cal A} = C^\infty({\cal M}) \otimes \mathbb{C}^2, \quad
	{\cal H} = L^2({\cal M,S}) \otimes \mathbb{C}^2, \quad
	D = \eth \otimes \mathbb{I}_2,
\end{equation}
with real structure $ J = {\cal J} \otimes J_{\cal F}$ and grading
$\Gamma = \gamma^5 \otimes \gamma_{\cal F} $, where $\eth$, $\cal J$,
$\gamma^5$ are as in~\eqref{eqn:3.9},  while 
\begin{equation}
\label{eq:JG}
	J_{\cal F} = \left(\begin{array}{cc} 0 & 1 \\ 1 & 0 \end{array}\right)cc,
	\qquad \gamma_{\cal F} = 
	\left(\begin{array}{cc} 1 & 0 \\ 0 & -1 \end{array}\right)\!,
\end{equation}
in the orthonormal basis $\{e,\bar e\}$ of $\mathcal{H_F} = \mathbb{C}^2$. The algebra ${\cal A} \ni a :=
(f,g)$ acts on ${\cal H}$ as 
\begin{equation}
\label{eq:p0_X}
	\pi_0(a) :=
	\left(\begin{array}{cc}
		f\mathbb{I}_4 & 0 \\ 0 & g\mathbb{I}_4
	\end{array}\right)\!, \qquad \forall f, g \in C^\infty({\cal M}).
\end{equation}

\smallskip

Following \S\ref{sec:2.2}, the minimal twist
of \eqref{eq:MxF_X} is given by the algebra ${\cal A} \otimes \mathbb{C}^2$, 
acting on $\HH$ as
\begin{equation}
\label{eq:mintwst_X}
	\pi(a,a') =
	\left(\begin{array}{cccc}
		f\mathbb{I}_2 & 0 & 0 & 0 \\ 0 & f'\mathbb{I}_2 & 0 & 0 \\ 
		0 & 0 & g'\mathbb{I}_2 & 0 \\ 0 & 0 & 0 & g\mathbb{I}_2
	\end{array}\right) =:
	\left(\begin{array}{cc}
		F & 0 \\ 0 & G' 
	\end{array}\right)\!,
\end{equation} 
for $a:=(f,g),\; a':=(f',g') \in {\cal A}$; with twist 
\begin{equation}
\label{eq:mintwst_Xbis}
	\pi(\rho(a,a'))=\pi(a',a)= \left(\begin{array}{cccc}
		f'\mathbb{I}_2 & 0 & 0 & 0 \\ 0 & f\mathbb{I}_2 & 0 & 0 \\ 
		0 & 0 & g\mathbb{I}_2 & 0 \\ 0 & 0 & 0 & g'\mathbb{I}_2
	\end{array}\right)=:
	\left(\begin{array}{cc}
		F' & 0 \\ 0 & G 
	\end{array}\right)\!.
\end{equation}
In both of the equations above, we have denoted
\begin{equation}
\label{eq:4.10}
\begin{split}
&	F := \pi_{\mathcal{M}}(f,f'), \qquad 
	F' := 
 \pi_{\mathcal{M}}(f',f), \\
&	G := \pi_{\mathcal{M}}(g,g'), \qquad \,
	G' := 
 \pi_{\mathcal{M}}(g',g),
\end{split}
\end{equation}
where $\pi_\M$ is the representation \eqref{eq:piM} of
$C^{\infty}(\mathcal{M})\otimes\mathbb{C}^2$ on $L^2(\M,{\cal S})$.

\subsection{Twisted fluctuation of of a doubled manifold}
\label{subsec:twistfluc2man}

We begin with some notations and a  technical lemma. Following \eqref{Xmu} and \eqref{eq:4.10}, given any $Z_\mu=\pi_\M(f_\mu,
f_\mu')$ with $f_\mu, f'_\mu\in C^\infty(\M)$, we denote $Z_\mu'=\pi_\M(f'_\mu,f_\mu)$ and\begin{equation}
\label{eq:notaZ}
	\mathbf{Z} := -i\gamma^\mu  Z_\mu, \quad 	\mathbf{Z}' :=
        -i\gamma^\mu  Z_\mu', \quad \bar{\mathbf{Z}}
        :=	-i\gamma^\mu  \bar{Z}_{\mu}.
 \end{equation} 
Notice that $\bar{\bf Z }$ \emph{is
  not} the complex conjugate of ${\bf Z}$, since in (\ref{eq:notaZ}) the complex
conjugation acts neither on $i$ nor on the Dirac matrices. This guarantees that
$\;\bar{}\;$ and $\;'\;$
 commute not only for $Z_\mu$, i.e. 
${\overline{Z'_\mu}}=(\bar Z_\mu)'=\pi_\M(\bar{f'_\mu}, \bar{f_\mu})$, but also for $\bf Z$, i.e.
\begin{equation}
  \label{eq:45}
  {\bf (\bar Z)'}={\bf \overline{Z'}}.
\end{equation}
The notation ${\bf \bar Z'}$
is thus unambiguous, and denotes indistinctly the two members of \eqref{eq:45}.

\begin{lemma} For any $F, G, Z_\mu$ as in \eqref{eq:4.10}, \eqref{eq:notaZ}, one has
\label{lem:notatio}
  \begin{align}
    \label{eq:4.11}
F[\eth,G]_{\rho} = -i\gamma^\mu F'\partial_{\mu}G,\qquad 
	{\cal J}\mathbf{Z}{\cal J}^{-1} = \bar{\mathbf{Z}},\qquad {\bf
          Z}^\dag = -{\bf \bar Z'}\!.
\end{align}
\end{lemma}
\begin{proof}
Eq. (\ref{eq:test1}) for $a=G$ yields  $[\eth,G]_{\rho}=-i\gamma^\mu\partial_\mu G$, while \eqref{eq:commgamma} for $a=F'$
  gives
  \begin{equation}
F\gamma^\mu  = \gamma^\mu F'.
\label{eq:94}
\end{equation}
Thus $F[\eth, G]_\rho = -i F\gamma^\mu\partial_\mu G = -i\gamma^\mu
F'\partial_\mu G$.
The second equation  \eqref{eq:4.11}  follows from
\begin{equation}
  \label{eq:41}
  \J {\bf Z}\J^{-1}= i\J\gamma^\mu Z_\mu \J^{-1}=- i\gamma^\mu \J
  Z_\mu \J^{-1}= -i\gamma^\mu \bar Z_\mu ={\bf \bar Z},
\end{equation}
where we used (\ref{eq:commgamma}) as well as (recalling that in
$KO$-dimension $4$, one has $\J^{-1}=-\J$)
\begin{small}
  \begin{align}
    \J Z_\mu \J^{-1} &= -i\left(\begin{array}{cc} \tilde\sigma^2&0\\
                                  0&\sigma^2\end{array}\right) cc \left(\begin{array}{cc}
                                                                          f_\mu\, \mathbb I_2&0\\
                                                                          0&f'_\mu\, \mathbb I_2\end{array}\right) i\left(\begin{array}{cc} \tilde\sigma^2&0\\
                                                                                                                            0&\sigma^2\end{array}\right) cc,\\
                     &= -\left(\begin{array}{cc} \tilde\sigma^2&0\\
                                 0&\sigma^2\end{array}\right)\!\!\left(\begin{array}{cc}
                                                                         \bar f_\mu\, \mathbb I_2&0\\
                                                                         0&\bar f'_\mu\, \mathbb I_2\end{array}\right)\!\!\left(\begin{array}{cc} \bar{\tilde\sigma}^2&0\\
                                                                                                                                  0&\bar\sigma^2\end{array}\right) =\left(\begin{array}{cc}
                                                                                                                                                                            \bar f_\mu\, \mathbb I_2&0\\
                                                                                                                                                                            0&\bar f'_\mu\, \mathbb
                                                                                                                                                                               I_2\end{array}\right) =\bar Z_\mu,
  \end{align}
\end{small}

\noindent  noticing that $\bar{\tilde\sigma}^2 =\tilde\sigma^2$ and
$\bar\sigma^2=\sigma^2$, so that
$\tilde\sigma^2\bar{\tilde\sigma}^2 =
\sigma^2\bar\sigma^2=-\mathbb I_2$. The third equation in \eqref{eq:4.11} follows from 
    \begin{equation}
{\bf Z}^\dag = iZ_\mu^\dag \gamma^\mu = i \bar Z_\mu \gamma^\mu =
i\gamma^\mu (\bar Z_\mu)' =i\gamma^\mu \bar{Z_\mu'} = -{\bf \bar Z'},
\label{eq:37}
\end{equation}
where we notice that $Z_\mu^\dag = \bar Z_\mu$, from the explicit form 
\eqref{eq:piM} of $\pi_\M$, then use \eqref{eq:94}.
\end{proof}

With this lemma, it is easy to compute the twisted fluctuation  $\omega_\rho + \J \omega_\rho  \J^{-1}$ 
for a generic twisted 1-form
\begin{equation}
\omega_\rho := \pi(a,a') \left[ \eth\otimes\mathbb{I}_2,\;\pi(b,b')\right]_{\rho}
\end{equation} 
for $a=(f,g)$,  $a'=(f',g')$, $b=(v, w)$, $b'=(v', w')$ in $\A$.
\begin{lemma}
\label{prop:twistfluctweyl}
 One has 
 \begin{equation}
    \label{eq:51bis}
    \omega_\rho + \J \omega_\rho \J^{-1} = {\bf X}\otimes {\mathbb I}_2
    +i{\bf Y}\otimes \gamma_{\cal F} ,
\end{equation}
with ${\bf X}=-i\gamma^\mu X_\mu$, ${\bf Y} =-i\gamma^\mu Y_\mu$ for
\begin{equation}
  \label{eq:40}
  X_\mu=\pi_\M(f_\mu, f'_\mu), \qquad Y_\mu=\pi_\M(g_\mu, g'_\mu),
\end{equation}
where $f_\mu, f'_\mu$ and $g_\mu, g'_\mu$ denote, respectively, the real and the
imaginary parts
of 
\begin{equation}
 \label{eq:47}
 z_\mu := f'\partial_\mu v +\bar g\partial_\mu\bar w'\!, 
 \quad \text{and} \quad z'_\mu =
 f\partial_\mu v' +\bar g'\partial_\mu\bar w'\!.
\end{equation}
\end{lemma}
\begin{proof}
  Define
  \begin{equation}
    \label{eq:38}
    V:=\pi_\M(v, v'),\quad V':=\pi_\M(v', v),\quad
    W:=\pi_\M(w,w'),\quad W'=\pi_\M(w', w).
  \end{equation}
From \eqref{eq:mintwst_X}--\eqref{eq:mintwst_Xbis}, one gets 
\begin{equation} 
\label{5.11}
\left[ \eth \otimes \mathbb{I}_2, \, \pi(b,b') \right]_{\rho} = 
\left( \begin{array}{cccc}
		[\eth,V]_{\rho} & 0 \\ 0 & [\eth,W']_{\rho} 
	\end{array} \right)\!,
\end{equation}
so that, for $(a, a')$ as in \eqref{eq:mintwst_X} and using
(\ref{eq:4.11}) one finds
\begin{equation}
\label{eq:1form2point}
\omega_\rho:
 = \left( \begin{array}{cc}
		F & 0 \\  0 & G'
	\end{array} \right) \!\!
	\left( \begin{array}{cc}
		[\eth,V]_{\rho} & 0 \\
		0 & [\eth,W']_{\rho}
	\end{array} \right) =
 \left( \begin{array}{cc}
		-i\gamma^\mu P_{\mu} & 0\\
		 0 & -i\gamma^\mu Q_{\mu}'  
	\end{array} \right) =
 \left( \begin{array}{cc}
{\bf P}& 0\\
		 0 & {\bf Q'} 
	\end{array} \right)\!,
\end{equation}
with
\begin{equation}
  \label{eq:39}
  P_\mu := F'\partial_\mu V, \qquad Q'_\mu:= G\partial_\mu W'\!.
\end{equation}
The explicit form of the real structure and its inverse,
\begin{equation} \label{J}
J=	{\cal J} \otimes J_F
=\left( \begin{array}{cc}
		0 &  {\cal J}  \\ {\cal J} & 0 
	\end{array} \right)\!, \qquad J^{-1}=\left( \begin{array}{cc}
		0 &  {\cal J}^{-1}  \\ {\cal J}^{-1} & 0 
	\end{array} \right)\!,
\end{equation}
along with the second equation \eqref{eq:4.11}, yield
\begin{equation} 
\label{eq;3-5.18}
\begin{split}
	J\omega_\rho J^{-1}
 = \left( \begin{array}{cc}
		{\cal J} {\bf Q'} {\cal J}^{-1} & 0 \\
		 0 & {\cal J}{\bf P} {\cal J}^{-1}
	\end{array} \right) = 
	\left( \begin{array}{cc}
		\bar{\mathbf{Q}}' & 0 \\ 
		 0 & \bar{\mathbf{P}} \end{array} \right)\!.
\end{split}
\end{equation}
Summing up \eqref{eq:1form2point} and \eqref{eq;3-5.18}, 
one obtains \eqref{eq:51}
\begin{equation}
  \label{eq:49}
  \omega_\rho + \J\omega_\rho \J^{-1} = \left(
    \begin{array}{cc}
      {\bf Z} & 0 \\ 0 &{\bf \bar Z}
    \end{array}\right)\!,
\end{equation}
where
${\bf Z}:= {\bf P} + {\bf \bar Q'}=-i\gamma^\mu Z_\mu$
with
\begin{equation}
\label{sac1bis}
Z_\mu = P_{\mu} + \bar Q_{\mu}'  = {F}'\partial_{\mu}{V} + \bar
                       G\partial_{\mu}\bar W' = \left( \begin{array}{cc}
		( f'\partial_\mu v + \bar g\partial_\mu \bar w')\mathbb{I}_2 & 0 \\
		0 & (f\partial_\mu v' + \bar g'\partial_\mu \bar w)\mathbb{I}_2
	\end{array} \right)
\end{equation}
(the last equation follows from the explicit form \eqref{eq:38} of $V,
W'$ and \eqref{eq:4.10} of $F', G$).
By \eqref{eq:47}), this reads as
\begin{equation}
Z_\mu=\pi_\M(z_\mu, z'_\mu) =
\pi_\M(f_\mu, f'_\mu) + i\pi_\M(g_\mu, g'_\mu) = X_\mu + iY_\mu.
\label{eq:42}
\end{equation}
Similarly, ${\bf \bar Z}=-i\gamma^\mu \bar Z_\mu$ with $\bar Z_\mu =
X_\mu - iY_\mu$. Hence, \eqref{eq:49} yields
 \begin{equation}
    \label{eq:51}
    \omega_\rho + \J \omega_\rho \J^{-1} =\left(
      \begin{array}{cc}
       -i\gamma^\mu (X_\mu + i Y_\mu) &0\\ 
       0 & -i\gamma^\mu(X_\mu - iY_\mu)\end{array}\right)\!,
        \end{equation}
which is nothing but \eqref{eq:51bis}.
\end{proof}

\smallskip

\begin{proposition}
\label{prop:twistfluctweylsa}
The self-adjoint twisted fluctuations of 
the Dirac operator of the doubled
manifold are parametrised by two real fields $f_\mu$ and $g_\mu$ in $C^\infty(\M,
\mathbb R)$, and are of the form
\begin{equation}
\label{eq:twstflct}
	\eth_{\bf X} \otimes \mathbb{I}_2 \; + \; g_\mu\gamma^\mu\otimes \gamma_{\cal F}
\end{equation}
where $\eth_X$ is the
twisted-covariant operator \eqref{eq:dflucM} of a manifold.
\end{proposition}
\begin{proof} A generic twisted fluctuation \eqref{eq:49} (adding a
  summation index $i$ and redefining 
${\bf Z}\!\!=\!\!\sum_i\!{\bf Z}_i$) is self-adjoint
iff ${\bf Z}={\bf Z}^\dag$ and ${\bf \bar Z}={\bf \bar Z}^\dag$.
By \eqref{eq:45}, and the third equation in \eqref{eq:4.11}, both
conditions are equivalent to $\bf Z = -\bf \bar Z'$, that is 
$-i\gamma^\mu \left(Z_{\mu} + \bar{Z}'_{\mu}\right) =0$. As discussed
below \eqref{eq:gsym}, this is equivalent to $Z_\mu=-\bar Z'_\mu$. From  \eqref{sac1bis}, this last condition is
      equivalent to $z_\mu = -\bar z_\mu'$, that is 
      \begin{equation}
        \label{eq:44}
        f_\mu= -f'_\mu, \;\text{ and } \; g_\mu = g'_\mu.
      \end{equation}
Substituting in \eqref{eq:40}, one obtains
\begin{equation}
        \label{eq:34}
X_\mu=\pi_\M(f_\mu, -f_\mu)= f_\mu\gamma^5,\qquad Y_\mu =
\pi_\M(g_\mu, g_\mu)= g_\mu {\mathbb I}_4, 
 \end{equation}
so that (\ref{eq:51bis}) gives
\begin{equation}
  \label{eq:46}
  \omega_\rho + \J \omega_\rho \J^{-1} = -i\gamma^\mu
  f_\mu\gamma^5\otimes{\mathbb I}_2
  +g_\mu\gamma^\mu\otimes\gamma_{\cal F}.
\end{equation}
The result follows adding $\eth\otimes {\mathbb I}_2$.
\end{proof}

Self-adjointness directly
 into the bold notation: by \eqref{eq:44}, ${\bf  X}\otimes \mathbb I_2 + i{\bf Y}\otimes\gamma_{\cal
  F}$ is self-adjoint iff ${\bf X'} = -{\bf X}$ and ${\bf Y'} = {\bf
   Y}$. Since ${\bf X} =\bar{\bf X}$, ${\bf Y} =\bar{\bf Y}$ by construction,
this is equivalent by the third equation \eqref{eq:4.11} to 
 ${\bf X}={\bf X}^\dag$ 
 and ${\bf Y}=-{\bf Y}^\dag$\!.

\subsection{Weyl equations from the twisted fermionic action}
\label{subsec:Weyl}

We show that the action
defined by the component $\eth_{\bf X}\otimes {\mathbb I}_2$ of the
twisted covariant Dirac operator \eqref{eq:twstflct} of the doubled
manifold (i.e. we assume that $g_\mu=0$) yields the Weyl
equations. Non vanishing $g_\mu$ will be taken into account in the
spectral triple of electrodynamics. 

Following the choice made in \eqref{eq:definR}, we take as a unitary implementing the action of $\rho$ on ${\cal H}$
\begin{equation}
{\cal R} = \gamma^0 \otimes \mathbb{I}_2.
\label{eq:58}
\end{equation}

It has eigenvalues $\pm1$ and is compatible with the real structure in
the sense of (\ref{eq:2.11}) with $\epsilon'''=-1$. 
A generic element $\eta$ in the $+1$-eigenspace $\cal H_R$
is 
\begin{equation}
\label{eq:eta}
	\eta = \phi \otimes e + \xi \otimes \bar e, \qquad \text{with}\quad
        \phi := \begin{pmatrix} \varphi \\ \varphi \end{pmatrix},\quad
	\xi := \begin{pmatrix} \zeta \\ \zeta \end{pmatrix}\!,
\end{equation} 
where $\phi,\xi \in L^2({\cal M,S})$ are Dirac the eigenspinors of
$\gamma^0$ (lemma \ref{lem3.1}), with Weyl 
components~$\varphi, \zeta$.

\begin{proposition}
\label{prop:3.6}  

The twisted fermionic action induced by $\eth\otimes \mathbb I_2$ on the doubled manifold is 
\begin{equation}
  \label{eq:57}
  S_\rho(\eth_{\bf X}\otimes \mathbb I_2) 
  = 2\;\frak A^\rho_{\eth_{\bf X}}(\tilde\phi,\tilde \xi) 
  = 4\int_{\cal M}\dmu \left[ \bar{\tilde\varphi}^\dag\sigma_2 \left( if_0-
\sum_{j=1}^3 \sigma_j\partial_j \right)\tilde\zeta \right]\!.
\end{equation}
\end{proposition}

\begin{proof}
	For $\eta, \eta' \in {\cal H_R}$ given by \eqref{eq:eta}, remembering
        that $J_{\cal F}e = \bar e$ and $J_{\cal F}\bar e = e$, one has
\begin{equation*}
	J\eta = {\cal J}\phi\otimes\bar e + {\cal J}\xi\otimes e, \quad 
(\eth_X\otimes\mathbb{I}_2)\eta' 
	= \eth_X\phi'\otimes e + \eth_X\xi'\otimes\bar e.
\end{equation*}
 So, Lem.~\ref{lem:bilinearform} with $\epsilon'''=-1$ yields
\begin{align} 
\label{eq:bilidoubleman1}
	{\frak A}^\rho_{\eth_{\bf X}\otimes\mathbb{I}_2}(\eta,\eta') &
	= -\langle J\eta, (\eth_{\bf X}\otimes\mathbb{I}_2)\eta' \rangle 
	= - \langle {\cal J}\phi, \eth_X\xi' \rangle
	- \langle {\cal J}\xi, \eth_X\phi' \rangle,\\
\label{eq:bilidoubleman2}
& = -\frak A_{\eth_{\bf X}}(\phi, \xi') - \frak A_{\eth_{\bf X}}(\xi, \phi')
= \frak A^\rho_{\eth_{\bf X}}(\phi, \xi') + \frak A^\rho_{\eth_{\bf
  X}}(\xi, \phi'),
\end{align}
where the first inner product is in $\HH$ and the second in $L^2(\M,S)$.
The action is then obtained substituting $\eta'=\eta$ and promoting
$\zeta$, $\varphi$ to
Gra{\ss}mann variables. The antisymmetric bilinear form $\frak A^\rho_{\eth_X}$
becomes symmetric when evaluated on Gra{\ss}mann variables (as in the
proof of \cite[Prop.~4.3]{W1}), hence
\begin{equation}
  \label{eq:81}
  	{\frak A}^\rho_{\eth_{\bf
            X}\otimes\mathbb{I}_2}(\tilde\eta,\tilde\eta)=2\frak
        A^\rho_{\eth_{\bf X}}(\tilde\phi, \tilde\xi).
\end{equation}
The result
then follows from Prop.~\ref{prop:actspecmanif}.
\end{proof}

Identifying the physical Weyl spinors as
 \begin{equation}
\psi := \tilde\zeta, \quad \psi^\dag := \pm
	i\bar{\tilde\varphi}^\dag\sigma_2
\label{eq:48}        
     \end{equation}  
(the sign is discussed below), the lagrangian density in
the action \eqref{eq:57} becomes
\begin{equation}
{\cal L}=\mp 4i\psi^\dagger 	\left(if_0
  - \textstyle\sum_{j}\sigma_j\partial_j
\right)\psi.
\label{eq:24}
\end{equation}
The Euler-Lagrange equation for $\psi^\dag$ yields the equation of motion
\begin{equation}
\left(if_0 - \textstyle\sum_{j}\sigma_j\partial_j
 \right)\psi=0.\label{eq:111}
 \end{equation}
\begin{proposition}
\label{Prop:Weyl}
  For $f_0$ a non-zero constant, a plane wave solution of \eqref{eq:111} coincides 
   with the left handed solutions of the Weyl equation with
  momentum  $p_0=-f_0$, or with the right handed
solution with momentum $p=f_0$. 
\end{proposition}
\begin{proof}
A plane-wave solution \eqref{eq:50} of \eqref{eq:111} satisfies
$(f_0+  \textstyle\sum_{j}\sigma_jp_j)\psi_l=0$. This is
equivalent to the first eq.\eqref{eq:52} with $p_0=-f_0$, or to the second \eqref{eq:52}~with~$p_0=f_0$.
\end{proof}

One may also identify directly the lagrangian \eqref{eq:24}
with the Weyl Lagrangians ${\cal L}_M^l$, ${\cal L}_M^r$~\eqref{eq:WeyllagL}.  
Choosing the minus sign in \eqref{eq:48} (that is the plus sign in \eqref{eq:24}), 
then ${\cal L}$ coincides (up to a global factor $4$) with ${\cal L}_M^l$ as soon as one imposes
$\partial_0\psi = if_0\psi$ (meaning, for a plane wave solution, $p_0=-f_0$).
Choosing instead the plus sign, 
then ${\cal L}$
coincides with 
${\cal L}_M^r$, as soon as one imposes 
$\partial_0\psi = -if_0\psi$ (meaning $p_0=f_0$). 

Prop.~\ref{Prop:Weyl} gives weight to the observation
made after Prop.\ref{prop:actspecmanif}: identifying $x_0$ with the
time coordinate of Minkowski spacetime, then 
the fermionic  action $S_\rho(\eth\otimes\mathbb I_2)$ of a
twisted doubled manifold - without fluctuation - yields the
spatial part of the Weyl equations (that is the lagrangian \eqref{eq:57}
with $f_0=0$). For a non-zero but constant $f_0$, the twisted fluctuation does not only bring back a fourth
component, but allows its interpretation as a time direction. 
It also  provides a clear interpretation of $f_0$ as the $0^\text{th}$
component of the momentum, that is an energy.

Even though the lagrangian density is lorentzian,
one may argue 
the action is not the Weyl one, for the manifold over which one
integrates is still riemannian. We come back  to this in the conclusion.

In these two examples - manifold and
  doubled manifold - the main difference between the twisted
and the usual fermionic actions does not lay so much in the twist of
the inner product than in the restriction to different
subspaces.
Indeed, by lemma \ref{lem:bilinearform} the twist of the inner product just
amounts to a global sign. As stressed in the following
remark,  this
is the restriction to $\HH_{\cal R}$ instead of $\HH_+$ that explains the
change of signature.
\begin{remark}
The disappearance of $\partial_0$ has no analogous in the non
  twisted case. In that case,  $\psi\in\HH_+$ and there is no
  fluctuation $\bf X$,  so that
  \begin{itemize}
  \item for a manifold, the usual fermionic action $\langle {\cal J}\tilde\psi, \eth \tilde\psi\rangle$ vanishes since $\eth \psi\in\HH_-$ while
    $\cal J \psi\in~\HH_+$;

\item for a doubled manifold, $\HH_+$ is
    spanned by $\left\{\xi\otimes e, \phi\otimes \bar e\right\}$ with
    $\xi=\left(
      \begin{array}{c}
        \zeta \\0
      \end{array}\right)$, $\phi=\left(
      \begin{array}{c}
        0\\ \varphi\end{array}\right)$ . Then 
    \begin{equation}
      S(\eth\otimes\mathbb I_2)=2 \langle {\cal J}\tilde \phi,
      \eth\tilde\xi\rangle = -2\int_\M \dmu\, \tilde{\bar\varphi}^\dag\sigma^2 \tilde\sigma^\mu\partial_\mu\tilde\zeta.
    \end{equation}
   By \eqref{eq:48}, the integrand is
    the euclidean
    version ${\cal L}_E^l:=i\Psi_l^\dag \tilde\sigma^\mu\partial_\mu \Psi_l$ of the Weyl
    lagrangian~${\cal L}_M^l$.
  \end{itemize}
\end{remark}

Following the result of \S \ref{sec:3.1.2}, one expects that the
field $f_\mu$ remains invariant under a gauge transformation. In order
not to make the paper too long, we do not check this here,
but we will do it for the spectral triple of electrodynamics in \S \ref{subsec:gaugetransformED}.
We will also give there the meaning of the other field $g_\mu$ that
parametrises the twisted fluctuation in
Prop. \ref{prop:twistfluctweylsa}. As in the non-twisted case, this
will identify with the $U(1)$ gauge field of electrodynamics. 

\section{Minimal twist of electrodynamics and Dirac equation}
\label{sec:electrody}
\setcounter{equation}{0} 

We first introduce the spectral triple of electrodynamics (as formalised 
in~\cite{W1,W}), then write down its minimal twist 
(\S\ref{subsec:mintwised}) following the recipe prepared 
in~\S\ref{sec:2.2}. We compute the twisted fluctuation in
\S\ref{sec:4.3}.
Gauge transformations are investigated
in \S\ref{subsec:gaugetransformED}: in addition to the
$X_\mu$ field encountered already for the minimal twist of the
(doubled) manifold, we obtain a $U(1)$ gauge field. Finally,
we compute the fermionic action in \S\ref{sec:Dirtaceq} and 
derive the lorentzian Dirac equation.

\subsection{Minimal twist of electrodynamics}
\label{subsec:mintwised}

The spectral triple 
of electrodynamics is the product of a riemannian manifold $\M$ (
still assumed to be four-dimensional) by a two-point space like
\eqref{eq:35}, except that $\D_{\cal F}$ is no longer zero (since fermions
are massive). In order to satisfy the axioms of
noncommutative geometry, this forces to enlarge $\HH_{\cal F}$ from $\mathbb C^2$ to $\mathbb C^4$
(see \cite{W1,W} for details). Hence
\begin{equation*}
	{\cal A}_{\text{ED}} = C^{\infty}(\mathcal{M}) \otimes \mathbb{C}^2,
	\quad {\cal H} = L^2(\mathcal{M,S}) \otimes \mathbb{C}^4,
	\quad \D = \eth \otimes \mathbb{I}_4 + \gamma^5 \otimes
        \D_{\cal F}; \quad
 	J = {\cal J} \otimes J_{\cal F}, \quad 
	\Gamma = \gamma^5 \otimes \gamma_{\cal F},
\label{ED}
\end{equation*}
where  $\eth$, ${\cal J}$, $\gamma^5$ are as in~(\ref{eqn:3.9}), 
$d\in\mathbb{C}$ is a constant parameter, and
\begin{equation}
\label{eq:stEDF}
	D_{\cal F} =
\left(\begin{array}{cccc}
	0 & d & 0 & 0 \\ \bar{d} & 0 & 0 & 0 \\ 0 & 0 & 0 & \bar{d} \\ 0 & 0 & d & 0
\end{array}\right)\!, \;
	J_{\cal F} =
\left(\begin{array}{cccc}
	0 & 0 & cc & 0 \\ 0 & 0 & 0 & cc \\ cc & 0 & 0 & 0 \\ 0 & cc & 0 & 0
\end{array}\right)\!, \;
	\gamma_{\cal F} = 
\left(\begin{array}{cccc}
	1 & 0 & 0 & 0 \\ 0 & -1 & 0 & 0 \\ 0 & 0 & -1 & 0 \\ 0 & 0 & 0 & 1
\end{array}\right)\!,
\end{equation}
written in an orthonormal basis
$\{e_L,e_R,\overline{e_L},\overline{e_R}\}$
 of 
$\HH_{\cal F}= \mathbb{C}^4$.
The algebra ${\cal A}_{\text{ED}} \ni a := (f,g)$ acts on ${\cal H}$ as
\begin{equation}
\label{p0}
	\pi_0(a) :=
	\left(\begin{array}{cccc}
		f\mathbb{I}_4 & 0 & 0 & 0 \\ 0 & f\mathbb{I}_4 & 0 & 0 \\ 
		0 & 0 & g\mathbb{I}_4 & 0 \\ 0 & 0 & 0 & g\mathbb{I}_4
	\end{array}\right)\!, \qquad \forall f, g \in C^\infty({\cal M}).
\end{equation}
Inner fluctuations are parametrised by a single $U(1)$~gauge field
$Y_\mu\!\in\! C^\infty({\cal M},\mathbb{R})$~\cite[(4.3)]{W1}:
\begin{equation}
\label{Ymu}
	D \to D_\omega = D + \gamma^\mu \otimes B_\mu, \qquad
	B_\mu := {\sf diag}(Y_\mu,Y_\mu,-Y_\mu,-Y_\mu);
\end{equation}
carrying an adjoint action of a unitary 
$u:=e^{i\theta} \in C^\infty({\cal M},U(1))$ on $D_\omega$, implemented by
\begin{equation}
\label{U1act}
	Y_\mu \to Y_\mu -iu\partial_\mu u^* = Y_\mu -\partial_\mu\theta,
	\qquad \theta \in C^\infty({\cal M},\mathbb{R}).
\end{equation}
Computing the action (fermionic and bosonic, via the spectral action
formula), one gets that this fields is the $U(1)$ gauge potential of electrodynamics.
\smallskip

A minimal twist is obtained by replacing ${\cal A}_{\text{ED}}$ 
by
$\A={\cal A}_{\text{ED}}\otimes\mathbb{C}^2$
along with 
its flip automorphism $\rho$~\eqref{eq:flip},
with the representation $\pi_0$ of ${\cal A}$ 
defined by~(\ref{repaa'}). Explicitly, 
\begin{equation}
\Gamma = \gamma^5 \otimes \gamma_{\cal F} =
	\left(\begin{smallmatrix}
		\mathbb{I}_2 & 0 \\ 0 & -\mathbb{I}_2
	\end{smallmatrix}\right) \otimes 
	\left(\begin{smallmatrix}
		1 & 0 & 0 & 0 \\ 0 & -1 & 0 & 0 \\ 0 & 0 & -1 & 0 \\ 0 & 0 & 0 & 1
	\end{smallmatrix}\right) =
	\left(\begin{smallmatrix}
		\mathbb{I}_2 & 0 & 0 & 0 & 0 & 0 & 0 & 0 \\
		0 & -\mathbb{I}_2 & 0 & 0 & 0 & 0 & 0 & 0 \\
		0 & 0 & -\mathbb{I}_2 & 0 & 0 & 0 & 0 & 0 \\
		0 & 0 & 0 & \mathbb{I}_2 & 0 & 0 & 0 & 0 \\
		0 & 0 & 0 & 0 & -\mathbb{I}_2 & 0 & 0 & 0 \\
		0 & 0 & 0 & 0 & 0 & \mathbb{I}_2 & 0 & 0 \\
		0 & 0 & 0 & 0 & 0 & 0 & \mathbb{I}_2 & 0 \\
		0 & 0 & 0 & 0 & 0 & 0 & 0 & -\mathbb{I}_2
	\end{smallmatrix}\right)\!,
\end{equation}
so that  the projections 
${\frak p}_{\pm} = \frac{1}{2}(\mathbb{I}_{16}\pm\Gamma)$ on the
eigenspaces ${\cal H}_{\pm}$ of $\HH$ are
\begin{equation}
	{\frak p}_+  = \mathsf{diag}
	(\mathbb{I}_2, 0_2, 0_2, \mathbb{I}_2,
	0_2, \mathbb{I}_2, \mathbb{I}_2, 0_2), \quad 
	{\frak p}_-  = \mathsf{diag}
	(0_2, \mathbb{I}_2, \mathbb{I}_2, 0_2,
	\mathbb{I}_2, 0_2, 0_2, \mathbb{I}_2).
\end{equation}
Therefore, for $(a, a')\in \A$, where $a:=(f,g)$, $a'\!:=(f'\!,g')$ with $f,g,f'\!,g' \in
C^{\infty}(\mathcal{M})$, one has
\begin{equation}
\label{5.4}
	\pi(a,a') ={\frak p}_+ \pi_0(a) + {\frak p}_- \pi_0(a') =
	\left(\begin{smallmatrix}
		f\mathbb{I}_2 & 0 & 0 & 0 & 0 & 0 & 0 & 0 \\
		0 & f'\mathbb{I}_2 & 0 & 0 & 0 & 0 & 0 & 0 \\
		0 & 0 & f'\mathbb{I}_2 & 0 & 0 & 0 & 0 & 0 \\
		0 & 0 & 0 & f\mathbb{I}_2 & 0 & 0 & 0 & 0 \\
		0 & 0 & 0 & 0 & g'\mathbb{I}_2 & 0 & 0 & 0 \\
		0 & 0 & 0 & 0 & 0 & g\mathbb{I}_2 & 0 & 0 \\
		0 & 0 & 0 & 0 & 0 & 0 & g\mathbb{I}_2 & 0 \\
		0 & 0 & 0 & 0 & 0 & 0 & 0 & g'\mathbb{I}_2
	\end{smallmatrix} \right) =:
\begin{pmatrix}
		F & 0 & 0 & 0 \\ 0 & F' & 0 & 0 \\ 0 & 0 & G' & 0 \\ 0 & 0 & 0 & G
	\end{pmatrix}\!,
\end{equation}
where $F, F', G$ and $G'$ are as in (\ref{eq:4.10}). The image of $(a,
a')\in\A$ under the flip $\rho $ is represented by
\begin{equation}
\label{5.4b}
	\pi(\rho(a,a')) = \pi(a',a) =
	\left( \begin{array}{cccc}
		F' & 0 & 0 & 0 \\ 0 & F & 0 & 0 \\ 0 & 0 & G & 0 \\ 0 & 0 & 0 & G'
	\end{array} \right)\!.
\end{equation}
In agreement with \eqref{eq:definR}, we choose as unitary ${\cal R \in
  B(H)}$ implementing the twist
\begin{equation}
\label{Red}
	{\cal R} = \gamma^0 \otimes \mathbb{I}_4 =
	\left( \begin{array}{cc}
		0 & \mathbb{I}_2 \\ \mathbb{I}_2 & 0
	\end{array} \right) \otimes \mathbb{I}_4.
\end{equation}
It is compatible with the real structure in the sense of \eqref{eq:2.11} with $\epsilon''' = -1$, as before.

\subsection{Twisted fluctuation of the Dirac operator}
\label{sec:4.3}

The twisted commutator $[D, a]_\rho$ being linear in $D$,  we treat separately the free part $\eth
\otimes \mathbb{I}_4$ and the finite part $\gamma^5 \otimes
D_{\cal F}$ of the Dirac operator.
The results are summarised in Prop.~\ref{prop:twistflucED}.

\subsubsection{The free part}
\label{subsubsec:freepart}

We show (Prop.~\ref{prop:4.1} below) that self-adjoint twisted fluctuations of  
 $\eth\otimes\mathbb{I}_4$ are parametrised by two real 
fields: $X_{\mu}$ 
 arising from the minimal twist of a manifold~\eqref{eq:dflucM} and
 the $U(1)$~gauge field $Y_{\mu}$ of
electrodynamics. To arrive there, we need a couple of lemmas.

\begin{lemma}
\label{lem:4.2}
For $a  = (f,g)$, $b=(v,w)$ in $\A_{\emph{ED}}$, and similar definition for $a' ,b'$, one has
\begin{equation} 
\label{w_M}
\omega_{\rho_{\cal M}} := 
	\pi(a,a')\left[\eth\otimes\mathbb{I}_4,\,\pi(b,b')\right]_{\rho} =
	\left( \begin{array}{cccc}
		\mathbf{P} & 0 & 0 & 0 \\ 0 & \mathbf{P}' & 0 & 0 \\ 
		0 & 0 & \mathbf{Q}' & 0 \\ 0 & 0 & 0 & \mathbf{Q}
	\end{array} \right)\!,
\end{equation}
where we use the notation \eqref{eq:notaZ} for
\begin{align}
  \label{eq:59}
P_{\mu} := F'\partial_{\mu}V, \quad P_{\mu}' := F \partial_{\mu}V', \quad
Q_{\mu} : = G'\partial_{\mu}W, \quad  Q_{\mu}': = G \partial_{\mu}W' ,
\end{align}
with $F,F',G,G'$ as in \eqref{eq:4.10}, and $V,V',W,W'$ as in \eqref{eq:38}.
\end{lemma}

\begin{proof}
	Using \eqref{5.4}--\eqref{5.4b} written for $(b,b')$, one computes 
\begin{equation} 
\label{5.110}
\left[ \eth \otimes \mathbb{I}_4, \, \pi(b,b') \right]_{\rho} =:\left( \begin{array}{cccc}
		[\eth,V]_{\rho} & 0 & 0 & 0 \\ 0 & [\eth,V']_{\rho} & 0 & 0 \\ 
		0 & 0 & [\eth,W']_{\rho} & 0 \\ 0 & 0 & 0 & [\eth,W]_{\rho}
	\end{array} \right)\!,
\end{equation}
The result follows multiplying by \eqref{5.4}, then using \eqref{eq:4.11}.
\end{proof}

\begin{lemma}
\label{lemma:twist1formed}
With the same notations as in Lem.~\ref{lem:4.2}, 1one has
\begin{equation} 
\label{5.16}
{\cal Z}:=	\omega_{\rho_{\cal M}} + J\omega_{\rho_{\cal M}}J^{-1}= \left( \begin{array}{cccc}
		\mathbf{Z} & 0 & 0 & 0 \\ 0 & \mathbf{Z'} & 0 & 0 \\
		0 & 0 & \bar{\mathbf{Z}} & 0 \\ 0 & 0 & 0 & \bar{\mathbf{Z'}}
	\end{array} \right)\!, 
\end{equation}
with 
$\mathbf{Z}  := \mathbf{P} + \bar{\mathbf{Q}}'$, 
		$\mathbf{Z'}  :=  \mathbf{P}' + \bar{\mathbf{Q}}$,
		$\bar{\mathbf{Z}} : =  \bar{\mathbf{P}} +
                \mathbf{Q}'$, and 
		$\bar{\mathbf{Z'}} : =  \bar{\mathbf{P}}' + \mathbf{Q}$.
\end{lemma}

\begin{proof}
From \eqref{w_M}, Lem.~\ref{lem:notatio} and the explicit form of $J={\cal J}\otimes J_{\cal F}$ with $J_{\cal F}$ as in 
\eqref{eq:stEDF}, one gets 
\begin{small}
  \begin{equation}
    \label{5.18}
    \begin{split}
      J\omega_{\rho_{\cal M}}J^{-1} & = \left( \begin{array}{cccc} 0 &
          0 & {\cal J} & 0 \\ 0 & 0 & 0 & {\cal J} \\ {\cal J} & 0 & 0
          & 0 \\ 0 & {\cal J} & 0 & 0
                                               \end{array} \right) \!\!\!
                                             \left( \begin{array}{cccc}
                                                      \mathbf{P} & 0 & 0 & 0 \\ 0 & \mathbf{P}' & 0 & 0 \\
                                                      0 & 0 & \mathbf{Q}' & 0 \\ 0 & 0 & 0 & \mathbf{Q}
                                                    \end{array} \right) \!\!\!
                                                  \left( \begin{array}{cccc}
                                                           0 & 0 & {\cal J}^{-1} & 0 \\ 0 & 0 & 0 & {\cal J}^{-1} \\ {\cal J}^{-1} & 0 & 0 & 0 \\ 0 & {\cal J}^{-1} & 0 & 0
                                                         \end{array}\right),\\[4pt]	
                                                       & =
                                                       \left( \begin{array}{cccc}
                                                                {\cal J}\mathbf{Q}'{\cal J}^{-1} & 0 & 0 & 0 \\ 0 & {\cal J}\mathbf{Q}{\cal J}^{-1} & 0 & 0 \\
                                                                0 & 0
                                                                                                     &
                                                                                                       {\cal
                                                                                                       J}\mathbf{P}{\cal
                                                                                                       J}^{-1}
                                                                                                         &
                                                                                                           0
                                                                \\ 0 &
                                                                       0
                                                                                                     &
                                                                                                       0
                                                                                                         &
                                                                                                           {\cal
                                                                                                           J}\mathbf{P}'{\cal
                                                                                                           J}^{-1}
                                                              \end{array} \right) = 
                                                            \left( \begin{array}{cccc}
                                                                     \bar{\mathbf{Q}}' & 0 & 0 & 0 \\ 0 & \bar{\mathbf{Q}} & 0 & 0 \\
                                                                     0 & 0 & \bar{\mathbf{P}} & 0 \\ 0 & 0 & 0 & \bar{\mathbf{P}}'
                                                                   \end{array} \right)\!.
                                                               \end{split}
                                                             \end{equation}
                                                           \end{small}

\noindent Adding up with \eqref{w_M}, the result follows.
\end{proof}

\begin{proposition}
\label{prop:4.1}
A self-adjoint twisted fluctuation~\eqref{5.16} of the 
	free Dirac operator ${\eth \otimes \mathbb{I}_4}$ is of the form
\begin{equation}
\label{5.20a}
	{\cal Z} = {\bf X} \otimes \mathbb{I}'
	+ i{\bf Y} \otimes \mathbb{I}'',
\end{equation}
where ${\bf X}=-i\gamma^\mu X_\mu$, ${\bf Y}=-i\gamma^\mu Y_\mu$, $\;\mathbb{I}' :=
 \mathsf{diag}(1,-1,1,-1),\,\mathbb{I}'' := \mathsf{diag}(1,1,-1,-1)$ with 
\begin{equation}
  \label{eq:65}
  	X_{\mu} := f_\mu \gamma^5, \qquad Y_{\mu} :=
        g_{\mu}\mathbb{I}_4,\qquad f_\mu,\,g_\mu \in
C^{\infty}(M,\mathbb{R}).
\end{equation}
\end{proposition}

\begin{proof}
From~\eqref{5.16}, it follows that ${\cal Z}$ is self-adjoint iff,
$\mathbf{Z} = \mathbf{Z}^\dag$, 
		$\mathbf{Z'} = \mathbf{Z'}^\dag$, 
		$\bar{\mathbf{Z}} = \bar{\mathbf{Z}}^\dag$ and 
		$\bar{\mathbf{Z'}} = \bar{\mathbf{Z'}}^\dag$.
	From \eqref{eq:45} and the third equation \eqref{eq:4.11}, these four
conditions are equivalent to 
$\mathbf{Z} = -\bar{\mathbf{Z}'}$, i.e.
\begin{equation}
Z_\mu = - \bar Z_\mu'.\label{eq:95}
\end{equation}
By lemma \ref{lemma:twist1formed}, one knows that
 \begin{equation}
   \label{eq:28}
   Z_\mu = P_\mu + \bar Q'_\mu =
   \begin{pmatrix}
     z^\mu \mathbb I_2&0\\0& z'_\mu\mathbb I_2
   \end{pmatrix}
\end{equation}
with $z_\mu= f'\partial_\mu v + \bar g\partial_\mu \bar w'$ and $z'_\mu=
f\partial_\mu v' + \bar g'\partial_\mu \bar w$. Denoting $f_\mu, g_\mu$ the real
and imaginary part of $z_\mu$  and (similarly for $z'_\mu$), then \eqref{eq:95}
is equivalent to  $f'_\mu=-f_\mu$ and $g'_\mu=g_\mu$, that is 
\begin{equation}
\label{eq:zmued}
	Z_\mu = 
	\left( \begin{array}{cc} 
		(f_\mu + ig_\mu)\mathbb{I}_2 & 0 \\ 0 & (-f_\mu + ig_\mu)\mathbb{I}_2
	\end{array} \right).
\end{equation}
In other terms, $Z_\mu = X_\mu +iY_\mu$ with  $X_\mu := f_\mu\gamma^5$, $Y_\mu
:= g_\mu\mathbb{I}_4$.

Going back to \eqref{5.16}, one obtains 
\begin{equation}
\label{eq:4.31}
\begin{split}
  {\cal Z} &= \left(\begin{array}{cccc}
                     \mathbf{Z} & 0 & 0 & 0 \\ 0 & -\bar{\mathbf{Z}} & 0 & 0 \\
                     0 & 0 & \bar{\mathbf{Z}} & 0 \\ 0 & 0 & 0 &
                                                                 -\mathbf{Z}
                   \end{array}\right) =
                 \left(\begin{array}{cccc}
                         -i\gamma^{\mu}Z_{\mu} & 0 & 0 & 0 \\
                         0 & i\gamma^{\mu}\bar{Z}_{\mu} & 0 & 0 \\
                         0 & 0 & -i\gamma^{\mu}\bar{Z}_{\mu} & 0 \\
                         0 & 0 & 0 & i\gamma^{\mu}Z_{\mu}
                       \end{array}\right)\\
& = \left( \begin{array}{cccc}
		-i\gamma^\mu  (X_\mu +iY_\mu) & 0 & 0 & 0 \\
		0 & i\gamma^\mu  (X_\mu -iY_\mu) & 0 & 0 \\
		0 & 0 & -i\gamma^\mu  (X_\mu -iY_\mu) & 0 \\ 
		0 & 0 & 0 & i\gamma^\mu (X_\mu +iY_\mu)
	\end{array} \right) \\
	& = -i\gamma^\mu  X_\mu\otimes\mathbb I'
      + i(-i\gamma^\mu  Y_\mu)\otimes\mathbb I''.
\end{split}
\end{equation}

\vspace{-.5truecm}
\end{proof}

\begin{remark}
 Imposing the self-adjointness of the twisted one-form
  $\omega_{\rho_\M}$  amounts to
  \begin{equation}
    \label{eq:61}
    {\bf P}^\dag ={\bf P}, \qquad {\bf Q}^\dag ={\bf Q}.
  \end{equation}
 This implies -- but is not equivalent -- to imposing the self-adjointness of
  $\omega_{\rho\M} + \J\omega_{\rho\M}\J^{-1}$\!,
  \begin{equation}
    \label{eq:62}
    {\bf Z}^\dag ={\bf Z}.
  \end{equation}
 As discussed below 
 Lem.~\ref{lem:3.1} for the minimal twist of a manifold, the relevant
 point is that the stronger condition \eqref{eq:61} does not imply
  that the twisted fluctuation $\cal Z$ be zero. The final form of the
  twist-fluctuated operator is the same, whether one requires
  \eqref{eq:61} or \eqref{eq:62}.
\end{remark}

\subsubsection{The finite part}
\label{subsubsec:finitepart}

In the spectral triple of electrodynamics, the finite part $\gamma^5 \otimes D_{\cal F}$ of the Dirac operator $D$~\eqref{ED}
does not fluctuate~\cite{W1}, for 
it commutes with the representation $\pi_0$~(\ref{p0}) of
${\cal A}_{\text{ED}}$.
The same is true for the minimal twist of electrodynamics.
\begin{proposition}
\label{prop:4.2}
	The finite Dirac operator $\gamma^5\otimes \D_{\cal F}$ has no twisted fluctuation.
\end{proposition}

\begin{proof}
	With the representations \eqref{5.4}--\eqref{5.4b}, 
        one calculates that
\begin{equation*}
\begin{split}
\left[ \gamma^5\otimes \D_{\cal F},\,\pi(a,a')
        \right]_{\rho}=&(\gamma^5\otimes \D_{\cal F})\,\pi(a,a') 
        - \pi(a',a)\,(\gamma^5\otimes \D_{\cal F})\\
& = \left( \begin{smallmatrix}
		0 & d\gamma^5 & 0 & 0 \\ \bar{d}\gamma^5 & 0 & 0 & 0 \\
		0 & 0 & 0 & \bar{d}\gamma^5 \\ 0 & 0 & d\gamma^5 & 0
	\end{smallmatrix} \right) \!\!
	\left( \begin{smallmatrix}
		F & 0 & 0 & 0 \\ 0 & F' & 0 & 0 \\ 0 & 0 & G' & 0 \\ 0 & 0 & 0 & G
	\end{smallmatrix} \right) -
	\left( \begin{smallmatrix}
		F' & 0 & 0 & 0 \\ 0 & F & 0 & 0 \\ 0 & 0 & G & 0 \\ 0 & 0 & 0 & G'
	\end{smallmatrix} \right) \!\!
	\left( \begin{smallmatrix}
		0 & d\gamma^5 & 0 & 0 \\ \bar{d}\gamma^5 & 0 & 0 & 0 \\
		0 & 0 & 0 & \bar{d}\gamma^5 \\ 0 & 0 & d\gamma^5 & 0
	\end{smallmatrix} \right) \\[3pt]
& = \left( \begin{matrix}
		0 & d[\gamma^5,F'] & 0 & 0 \\ \bar{d}[\gamma^5,F] & 0 & 0 & 0 \\
		0 & 0 & 0 & \bar{d}[\gamma^5,G] \\ 0 & 0 & d[\gamma^5,G'] & 0
	\end{matrix} \right) = 0,
\end{split}
\end{equation*}
where $F$, $F'$, $G$, $G'$ \eqref{eq:4.10} being diagonal, commute with 
$\gamma^5$.
\end{proof}
\bigskip

The results of this section summarise as follows:

\begin{proposition}
\label{prop:twistflucED}
The Dirac
  operator
  ${\cal D} = {\eth\otimes\mathbb{I}_4} + {\gamma^5\otimes {\cal D}_{\cal F}}$ of
  electrodynamics, under the minimal twist \eqref{5.4}--\eqref{Red}, 
  twist-fluctuates to
  \begin{equation}
    \label{eq:4.34}
 {\cal D}_{\cal Z} := {\cal D} + {\cal Z}, 
  \end{equation}
  where $\cal Z$ is given by Prop.~\ref{prop:4.1}.
\end{proposition}

 \begin{remark}
   Expectedly, substituting $\rho = \mathsf{Id}$, one returns to the
   non-twisted case: the triviality of $\rho$ is
   tantamount to equating~\eqref{5.4} with~\eqref{5.4b}, that is to
   identify the `primed' functions with their
   `un-primed' partners. Hence, ${\bf Z'}={\bf Z}$.
   Imposing self-adjointness, the third eq.  \eqref{eq:4.11} gives
   ${\bf Z}=-{\bf \bar Z}$. Going back to \eqref{eq:zmued}, this
   yields $f_\mu=0$. Therefore, $X_\mu$ vanishes and remains only the
   $U(1)$ gauge field $\bf Y$.
 The latter is
 \begin{equation}
   i{\bf Y}\otimes \mathbb I''=\gamma^\mu Y_\mu\otimes \mathbb
   I''=\gamma^\mu \otimes g_\mu\mathbb I'',
 \end{equation}
and coincides with
 the gauge potential  $\gamma^\mu\otimes B_\mu$ of the spectral triple
 of electrodynamics\eqref{Ymu} in the non-twisted case. 
 \end{remark}

\subsection{Gauge transformation}
\label{subsec:gaugetransformED}

We discuss
 the transformation of the fields  $\bf X$ and $\bf Y$ parametrizing the twisted
 fluctuation $\cal Z$, along the lines of \S\ref{sec:3.1.2}. 
A unitary $u$ of $\A_\text{ED}\otimes \mathbb C^2$ is of the form
$u=(v,v')$, where $v:=(e^{i\alpha},e^{i\beta})$, $v':=(e^{i\alpha'},e^{i\beta'})$ are 
unitaries 
of ${\cal A}_\text{ED}$, with $\alpha,\alpha',\beta,\beta'\in C^\infty({\cal M},\mathbb{R})$.
It (and its twist) acts on $L^2(\M,S)\otimes \mathbb C^4$ as
\begin{equation}
\label{eq:piuu'}
	\pi(u)= 
	\left(\begin{array}{cccc}
		A & 0 & 0 & 0 \\ 0 & A' & 0 & 0 \\ 0 & 0 & B' & 0 \\ 0 & 0 & 0 & B
	\end{array}\right)\!, \quad 	\pi(\rho(u)) = \pi(v',v) =
	\left( \begin{array}{cccc}
		A' & 0 & 0 & 0 \\ 0 & A & 0 & 0 \\ 0 & 0 & B & 0 \\ 0 & 0 & 0 & B'
	\end{array} \right)\!,
\end{equation}
where we denote
\begin{equation}
\label{rep:A,B}
\begin{split}
&	A := \pi_{\mathcal{M}}(e^{i\alpha},e^{i\alpha'}), \qquad 
	A' := \rho(A) = \pi_{\mathcal{M}}(e^{i\alpha'},e^{i\alpha}), \\
&	B := \pi_{\mathcal{M}}(e^{i\beta},e^{i\beta'}), \qquad \,
	B' := \rho(B) = \pi_{\mathcal{M}}(e^{i\beta'},e^{i\beta}).
\end{split}
\end{equation}

\begin{proposition}
	Under a gauge transformation \eqref{eq:9}, $\bf X$ remains invariant while 
$\bf Y$ is mapped to 
\begin{equation}
  \label{eq:107}
 -i\gamma^\mu\left( Y^\mu - 
  \begin{pmatrix}
    \partial_\mu \theta \mathbb I_2 &0 \\ 0& \partial_\mu\theta'
    \mathbb I_2
  \end{pmatrix}\right)
\end{equation}
for $\theta:=\alpha - \beta'$, $\theta'=\alpha'-\beta$.
\end{proposition}

\begin{proof}
	Since $\gamma_{\cal F}\otimes D_{\cal F}$ twist-commutes with
        the algebra, in the transformation \eqref{LM2Prop4.2} of the gauge potential
        it is enough to consider  $\eth
        \otimes \mathbb{I}_4$. So $\omega_{\rho_\M}$ in \eqref{5.16}
        transforms to
	\begin{equation}
	\label{eq:4.45}
	\begin{split}
\omega_{\rho_{\cal M}}^w
		 &= \rho(u) \left( [\eth\otimes\mathbb{I}_4, u^*]_\rho 
		+ \omega_{\rho_{\cal M}}u^* \right) 
		 =\rho(u)\left( \eth\otimes\mathbb{I}_4 
		+ \omega_{\rho_{\cal M}} \right)u^*\!,
	\end{split}
	\end{equation}
	where we used 
$	[\eth\otimes\mathbb{I}_4, u^*]_\rho
		= (\eth\otimes\mathbb{I}_4)u^*$ as in ~\eqref{Id1}.
	By~\eqref{eq:piuu'} and Lem.~\ref{lem:4.2}, this
        transformation writes
	\begin{equation*}
		\left( \begin{array}{cccc}
			{\bf P} &0 &0 &0 \\0 & {\bf P'} & 0& 0\\ 
			0&0 & {\bf  Q'} &0 \\ 0& 0&0 & {\bf Q}	
		\end{array} \right) \to
		\left( \begin{array}{cccc}
			A'(\eth + {\bf P})\bar A &0 &0 &0 \\
			0& A(\eth +{\bf  P'})\bar A' &0 &0 \\ 
			0& 0& B(\eth+ {\bf Q'})\bar B' &0 \\ 
			0&0 &0 & B'(\eth + {\bf Q})\bar B
		\end{array} \right)\!,
	\end{equation*}
    Since $A'$, $B'$ twist-commute with $\gamma^\mu$ and $A$ commutes
    with $P_\mu$ (and $B$ with $Q_\mu$), one has 
that $P_\mu$ is mapped to  $P_\mu + A\partial_\mu\bar A$ and
	$Q'_\mu$ to $Q_\mu + B'\partial_\mu\bar B'$.
Thus $Z_\mu =  P_\mu+\bar Q'_\mu$ in \eqref{eq:95} is mapped to 
$	
				Z_\mu + \left(A\partial_\mu\bar A +
                                  \bar B'\partial_\mu B'\right)$.
With the representations \eqref{eq:28} of $Z_\mu$
and  \eqref{rep:A,B}
	of $A, B$, this means
	\begin{equation*}
		\left( \begin{array}{cc}
			z_\mu\mathbb{I}_2 & 0 \\ 0 &  z'_\mu\mathbb{I}_2
		\end{array} \right) \longrightarrow
		\left( \begin{array}{cc}
			\left( z_\mu-i\partial_\mu\theta \right) \mathbb{I}_2 & 0 \\ 
			0 &  (z'_\mu -i\partial_\mu\theta') \mathbb{I}_2
		\end{array} \right)\!.
	\end{equation*} 
The result follows remembering that $X^\mu$ and $Y^\mu$ are the real
and imaginary parts of $Z^\mu$.
\end{proof}

	By imposing that both $\cal Z$ and its gauge transform are
        self-adjoint, that is by lemma \ref{lem:notatio}: $z'_\mu = -\bar z_\mu$ and $z'_\mu -
        i\partial_\mu\theta' =- \overline{z_\mu
          -i\partial_\mu\theta}$, one is forced to identify
        $\theta'=\theta+\text{constant}$. 
	Then \eqref{eq:107} means that $Y_\mu=g_\mu\mathbb{I}_4$ undergoes the  transformation
	\begin{equation}
	\label{eq:4.52}
		g_\mu \to g_\mu - \partial_\mu\theta, \qquad 
		\theta \in C^\infty({\cal M},\mathbb{R}).
	\end{equation}
This is a $U(1)$ gauge field, formally similar to the one \eqref{U1act} of the
(euclidean) non-twisted case. 
By computing the twisted fermionic action, we show
that this actually identifies with the $U(1)$ of electromagnetism, but
now in lorentzian signture.
\begin{remark}
 For $\theta'-\theta$ a non-zero constant, the gauge transformation
 preserves the selfadjointness of the twisted fluctuation, even
 though $u$ is not invariant by the twist. This is because 
 such an $u$ satisfies the
 weaker condition for preserving selfadjointness - pointed out in
 \cite[\S5.1]{LM2} - namely
 $\rho(u)^*u$ twist-commutes with $\cal D$.  
\end{remark}
\subsection{Lorentzian Dirac equation from twisted fermionic action}
\label{sec:Dirtaceq}

To calculate the action, we first  identify
the eigenvectors of the unitary $\R$ implementing the~twist.
 
\begin{lemma}
	Any $\eta$ in the positive eigenspace 
	$\mathcal{H_R}$~\eqref{eq:2.16} of the unitary ${\cal R}$~\eqref{Red}
	is of the form
\begin{equation}
\label{6.3}
	\eta = \phi_1 \otimes e_L + \phi_2 \otimes e_R 
	+ \xi_1 \otimes \overline{e_L} + \xi_2 \otimes \overline{e_R}, \quad
\end{equation}
where
$\phi_{k=1, 2} := \begin{pmatrix} \varphi_k \\
  \varphi_k \end{pmatrix}$ and 
$\xi_{k=1, 2} := \begin{pmatrix} \zeta_k \\
  \zeta_k \end{pmatrix}$	
are Dirac spinors with Weyl
        components $\varphi_k, \zeta_k$.
\end{lemma}
\begin{proof}
	${\cal R}$
	has eigenvalues $\pm1$ and its eigenvectors corresponding to the eigenvalue 
	$+1$ are:
\begin{equation*}
\begin{split}
	\varepsilon_1 = \upsilon_1 \otimes e_L, \quad
	\varepsilon_2 = \upsilon_2 \otimes e_L, \quad
&	\varepsilon_3 = \upsilon_1 \otimes e_R, \quad
	\varepsilon_4 =	\upsilon_2 \otimes e_R,	\\
	\varepsilon_5 = \upsilon_1 \otimes \overline{e_L}, \quad
	\varepsilon_6 = \upsilon_2 \otimes \overline{e_L}, \quad
&	\varepsilon_7 = \upsilon_1 \otimes \overline{e_R}, \quad
	\varepsilon_8 = \upsilon_2 \otimes \overline{e_L},
\end{split}
\end{equation*}
	where $\upsilon_1 := 
	\left( \begin{smallmatrix} 1 \\ 0 \end{smallmatrix}
        \right)\otimes 	\left( \begin{smallmatrix} 1 \\ 1 \end{smallmatrix}
        \right), \;
	\upsilon_2 := 
	\left( \begin{smallmatrix} 0 \\ 1 \end{smallmatrix}
        \right)\otimes	\left( \begin{smallmatrix} 1 \\ 1\end{smallmatrix}
        \right)$ denote 
	the eigenvectors of $\gamma^0$. Thus, 
\begin{equation*}
\begin{split}
	\eta \enskip  & = \enskip \textstyle\sum_{j = 1}^8 \lambda_j \varepsilon_j 
        = \enskip
	(\lambda_1 \upsilon_1 + \lambda_2 \upsilon_2) \otimes e_L +
 	(\lambda_3 \upsilon_1 + \lambda_4 \upsilon_2) \otimes e_R +
	(\lambda_5 \upsilon_1 + \lambda_6 \upsilon_2) \otimes \overline{e_L} +
	(\lambda_7 \upsilon_1 + \lambda_8 \upsilon_2) \otimes \overline{e_R}, \\ & = \enskip
	\phi_1 \otimes e_L + \phi_2 \otimes e_R + 
	\xi_1 \otimes \overline{e_L} + \xi_2 \otimes \overline{e_R}, \end{split}
\end{equation*} 
with
	$\varphi_1 := \begin{pmatrix} \lambda_1 \\ \lambda_2 \end{pmatrix}\!,\;
	\varphi_2 := \begin{pmatrix} \lambda_3 \\ \lambda_4 \end{pmatrix}\!,\;
	\xi_1 := \begin{pmatrix} \lambda_5 \\ \lambda_6 \end{pmatrix}\!,\;
	\xi_2 := \begin{pmatrix} \lambda_7 \\ \lambda_8 \end{pmatrix}\!.$
\end{proof}

The following lemma is
useful to compute the
contribution of $\gamma^5\otimes {\cal D}_F$ and $\bf Y$
to the  action.
\begin{lemma}
\label{lem:4.6}
	For Dirac spinors $\phi :=
	\begin{pmatrix} \varphi \\ \varphi \end{pmatrix}$, $\xi :=
	\begin{pmatrix} \zeta \\ \zeta\end{pmatrix}$ in
        $L^2(\mathcal{M,S})$, one has
	\vspace{-5pt}
\begin{equation}
\label{6.4.2}
{\frak A}_{i{\bf Y}}(\phi, \xi) = 2i\int_{\cal M} \emph{d}\upmu \;
	\bar\varphi^\dag \sigma_2 \left(\textstyle\sum_{j}\sigma_j g_j\right)\zeta, 
\qquad  {\frak A}_{\gamma^5}(\phi, \xi)= -2 \int_{\cal M} \emph{d}\upmu \,
	\bar\varphi^\dag \sigma_2 \zeta.
\end{equation}
\end{lemma}

\begin{proof}
Using \eqref{eq:65} for $Y_\mu$ and \eqref{EDirac} for
the Dirac matrices, one gets
\begin{equation*}
\begin{split}
	i {\bf Y} \xi
& = \gamma^\mu Y_\mu \left( \begin{array}{c} \zeta \\ \zeta \end{array} \right) 
  = \left( \begin{array}{cc} 
  		0 & \sigma^\mu \\ \tilde\sigma^\mu & 0 
	\end{array} \right) \!\!
	\left( \begin{array}{cc} 
		g_\mu\mathbb{I}_2 & 0 \\ 0 & g_\mu\mathbb{I}_2 
	\end{array} \right) \!\!
  	\left( \begin{array}{c} \zeta \\ \zeta \end{array} \right)
  = \left( \begin{array}{c}
  		g_\mu\sigma^\mu \zeta \\ g_\mu\tilde\sigma^\mu \zeta
	\end{array} \right)\!.
\end{split}
\end{equation*}
Along with \eqref{eq:JPhi}, recalling that $\sigma^{2\dag} =
i\sigma_2$ and $\tilde\sigma^{2\dag} = -i\sigma_2$, yields
\begin{equation*}
 \begin{split}
 {\frak A}_{i{\bf Y}}(\phi, \xi) =	({\cal J} \phi)^\dag (i{\bf Y} \xi) & = -i\left( \begin{array}{c}
 		\tilde\sigma^2 \bar \varphi \\ \sigma^2 \bar\varphi 
 	\end{array} \right)^{\!\!\dag} \!\!
 	\left( \begin{array}{r} 
 		g_\mu \sigma^\mu \zeta\\ g_\mu \tilde\sigma^\mu \zeta 
 	\end{array} \right)
   = -i\int_\M \emph{d}\upmu\;	 \bar\varphi^\dag \left( \tilde\sigma^{2\dag} \sigma^\mu+\sigma^{2\dag} \tilde\sigma^\mu\right) g_\mu\zeta, \\
 &    = \;\;\int_\M \emph{d}\upmu	\;	\bar\varphi^\dag \sigma_2
     	(-\sigma^\mu + \tilde\sigma^\mu) g_\mu \zeta \quad
     = \; 2i \;\int_\M \emph{d}\upmu\;	\bar\varphi^\dag \sigma_2 \left(\textstyle\sum_{j}\sigma_j g_j\right)\zeta,
 \end{split}
 \end{equation*}
where we used \eqref{eq:sumsigma} and obtained the first equation of
\eqref{6.4.2}. The second one follows from
\begin{equation*}
{\frak A}_{\gamma^5}(\phi, \xi) =	({\cal J} \phi)^\dag (\gamma^5 \xi)
 = -i \left( 
   \begin{array}{c}
     \tilde\sigma^2 \bar\varphi \\ \sigma^2 \bar\varphi
   \end{array} \right)^{\!\!\dag} \!\!
	\left( \begin{array}{r} \zeta \\ -\zeta \end{array} \right)
 = -i \int_\M \emph{d}\upmu\;\left(\bar\varphi^\dag \tilde\sigma^{2\dagger} \zeta 
	- \bar\varphi^\dag \sigma^{2\dagger} \zeta\right)
  = -2 \int_\M \emph{d}\upmu\;\bar\varphi^\dag \sigma_2 \zeta.
\end{equation*}
\end{proof}

\begin{proposition}
\label{prop:eq-Dirac}
The fermionic action of the minimal twist of 
electrodynamics is the integral
\begin{align*}
& S_\rho({\cal D}_{\cal Z}) = {\frak A}^\rho_{D_{\cal
  Z}}(\tilde\eta,\tilde\eta) = 4\!\int_{\cal M}\!\!\!\emph{d}\upmu\; 	{\cal L}
\end{align*}
of the lagrangian density
\begin{equation}
\label{eq:lagDir1}
	{\cal L}:= 
	{\bar{\tilde\varphi}}_1^\dag \sigma_2 \left( if_0- \textstyle\sum_j\sigma_j\frak D_j \right) \tilde\zeta_1 -
	{\bar{\tilde{\varphi}}}_2^\dag \sigma_2 \left( if_0 + \textstyle\sum_j\sigma_j\frak D_j \right) \tilde\zeta_2 +
	\left( \bar d \bar{\tilde{\varphi}}_1^\dag \sigma_2 \tilde\zeta_2 + 
		d \bar{\tilde{\varphi}}_2^\dag \sigma_2 \tilde\zeta_1 \right),
\end{equation}
with
${\frak D}_\mu := \partial_\mu - ig_\mu$
the covariant derivative associated to the electromagnetic four-potential~\eqref{eq:4.52}.
\end{proposition}
\begin{proof}
	Let  ${\frak A}^\rho_{{\cal D}_{\cal Z}}$ be
	 the antisymmetric bilinear form~\eqref{Sfrho}
         defined by the twisted-covariant Dirac operator 
	~\eqref{eq:4.34}. It breaks down into four terms:
\begin{equation}
\label{6.8}
	{\frak A}^\rho_{D_{\cal Z}} = 
	{\frak A}^\rho_{\eth \otimes \mathbb{I}_4} + 
	{\frak A}^\rho_{{\bf X} \otimes \mathbb{I}'} +
	{\frak A}^\rho_{i{\bf Y} \otimes \mathbb{I}''} + 
	{\frak A}^\rho_{\gamma^5 \otimes D_{\cal F}}.
\end{equation} 
For $\eta,\eta' \in\cal H_R$ as in \eqref{6.3} one gets
\begin{equation}
\begin{split}
	J \eta
& = {\cal J} \phi_1 \otimes \overline{e_L} + 
	{\cal J} \phi_2 \otimes \overline{e_R} + 
	{\cal J} \xi_1 \otimes e_L +
	{\cal J} \xi_2 \otimes e_R, \\[3pt]
	(\eth \otimes \mathbb{I}_4) \eta'
& = \eth \phi'_1 \otimes e_L + \eth \phi'_2 \otimes e_R
	+ \eth \xi'_1 \otimes \overline{e_L} 
	+ \eth \xi'_2 \otimes \overline{e_R}, \\[3pt]
	({\bf X} \otimes \mathbb{I}') \eta'
& = {\bf X} \phi'_1 \otimes e_L - {\bf X} \phi'_2 \otimes e_R
	+ {\bf X} \xi'_1 \otimes \overline{e_L} 
	- {\bf X} \xi'_2 \otimes \overline{e_R}, \\[3pt]
	(i{\bf Y} \otimes \mathbb{I}'') \eta'
& = i{\bf Y} \phi'_1 \otimes e_L + i{\bf Y} \phi'_2 \otimes e_R
	- i{\bf Y} \xi'_1 \otimes \overline{e_L} 
	- i{\bf Y} \xi'_2 \otimes \overline{e_R}, \\[3pt]
	(\gamma^5 \otimes D_{\cal F}) \eta'
& = \gamma^5 \phi'_1 \otimes \bar{d}e_R + \gamma^5 \phi'_2 \otimes de_L
	+ \gamma^5 \xi'_1 \otimes d\;\overline{e_R} 
	+ \gamma^5 \xi'_2 \otimes \bar{d}\;\overline{e_L}
\end{split}
\end{equation}
where the first and last equations come from the explicit forms \eqref{eq:stEDF}
of $J_{\cal F}$ and $D_{\cal F}$, while the third and fourth follow
from the explicit form \eqref{5.20a} of $\bf X$ and $\bf Y$.
These equations allow to reduce each of the four terms in \eqref{6.8}
to a bilinear form on $L^2(\M, \cal S)$ rather than on the tensor product
$L^2(\M,{\cal S})\otimes \mathbb C^4$. More precisely, recalling Lem.~\ref{lem:bilinearform} with
        $\epsilon'''=-1$ (and noticing that $\eth\otimes\mathbb I_4$, ${\bf X}\otimes \mathbb
        I'$, $i{\bf Y}\otimes \mathbb I''$, $\gamma^5\otimes {\cal
          D}_{\cal F}$ are all selfadjoint), one computes:
\begin{align} 
\nonumber
{\frak A}^\rho_{\eth \otimes \mathbb{I}_4}(\eta, \eta') & =-{\frak A}_{\eth \otimes \mathbb{I}_4}(\eta, \eta')=-
	\langle J\eta, (\eth \otimes \mathbb{I}_4)\eta' \rangle, \\
\nonumber & =
-	\langle {\cal J} \phi_1, \eth \xi'_1 \rangle 
-	\langle {\cal J} \phi_2, \eth \xi'_2 \rangle 
-	\langle {\cal J} \xi_1, \eth \phi'_1 \rangle 
-	\langle {\cal J} \xi_2, \eth \phi'_2 \rangle ,\\
\label{eq:bilinear1}&=-\frak A_\eth(\phi_1, \xi'_1 )  
-	\frak A_\eth(\phi_2, \xi'_2) 
-	\frak A_\eth(\xi_1, \phi'_1) 
-	\frak A_\eth(\xi_2, \phi'_2 );\\[4pt]
\nonumber	
{\frak A}^\rho_{{\bf X} \otimes \mathbb{I}'}(\eta, \eta') & 
=-{\frak A}_{{\bf X} \otimes \mathbb{I}'}(\eta, \eta') =-
	\langle J\eta, ({\bf X} \otimes \mathbb{I}')\eta' \rangle, \\ 
\nonumber
& =
	-\langle {\cal J} \phi_1, {\bf X} \xi'_1 \rangle +
	\langle {\cal J} \phi_2, {\bf X} \xi'_2 \rangle -
	\langle {\cal J} \xi_1, {\bf X} \phi'_1 \rangle +
	\langle {\cal J} \xi_2, {\bf X} \phi'_2 \rangle, \\
\label{eq:bilinear2}
&=-\frak A_{\bf X}(\phi_1, \xi'_1) +  \frak A_{\bf X}(\phi_2, \xi'_2) -
	\frak A_{\bf X}(\xi_1, \phi'_1) +
	\frak A_{\bf X}(\xi_2,\phi'_2);\\[4pt]
 \nonumber
  {\frak A}^\rho_{i{\bf Y} \otimes \mathbb{I}''}(\eta,
  \eta') & = -{\frak A}_{i{\bf Y} \otimes \mathbb{I}''}(\eta, \eta')=-
	\langle J\eta, (i{\bf Y} \otimes \mathbb{I}'')\eta' \rangle, \\
  \nonumber & =
	\langle {\cal J} \phi_1, i{\bf Y} \xi'_1 \rangle +
	\langle {\cal J} \phi_2, i{\bf Y} \xi'_2 \rangle -
	\langle {\cal J} \xi_1, i{\bf Y} \phi'_1 \rangle -
	\langle {\cal J} \xi_2, i{\bf Y} \phi'_2 \rangle, \\
\label{eq:bilinear3}
& =
	\frak A_{i\bf Y} (\phi_1, \xi'_1) +
	\frak A_{i\bf Y}(\phi_2,\xi'_2) -
	\frak A_{i\bf Y} (\xi_1, \phi'_1)- 
	\frak A_{i\bf Y}(\xi_2,\phi'_2);
  \\[4pt]
\nonumber
	{\frak A}^\rho_{\gamma^5 \otimes D_{\cal F}}(\eta, \eta') & =-{\frak A}_{\gamma^5 \otimes D_{\cal F}}(\eta, \eta') =-
	\langle J\eta, (\gamma^5 \otimes D_{\cal F})\eta' \rangle \\ 
\nonumber
& =
	-\bar d\langle {\cal J} \phi_1, \gamma^5 \xi'_2 \rangle -
	d\langle {\cal J} \phi_2, \gamma^5 \xi'_1 \rangle -
	d\langle {\cal J} \xi_1, \gamma^5 \phi'_2 \rangle -
	\bar d\langle {\cal J} \xi_2, \gamma^5 \phi'_1 \rangle,  \\
\label{eq:bilinear4}
 & =
	-\bar d\,\frak A_{\gamma^5}(\phi_1, \xi'_2) -
	d\,\frak A_{\gamma^5}(\phi_2, \xi'_1) -
	d\,\frak A_{\gamma^5}(\xi_1, \phi'_2) -
	\bar d\,\frak A_{\gamma^5}(\xi_2, \phi'_1 ).
\end{align}
Substituting $\eta=\eta'$,  then going to Gra{\ss}mann variables, the sum of
\eqref{eq:bilinear1}, \eqref{eq:bilinear2}, and \eqref{eq:bilinear4} is
\begin{equation}
\begin{split}
-2\,\frak A_\eth(\tilde \phi_1, \tilde \xi_1 )  -2\,\frak A_{\eth}(\tilde \phi_2, \tilde \xi_2) 
-2\,\frak A_{\bf X}(\tilde \phi_1, \tilde \xi_1)+2\,\frak A_{\bf X}(\tilde
\phi_2, \tilde \xi_2)\\
 - 	2\bar d\,\frak A_{\gamma^5}(\tilde\phi_1, \tilde\xi_2) -
	2d\,\frak A_{\gamma^5}(\tilde\phi_2, \tilde\xi_1);
\label{eq:64}
\end{split}
\end{equation}
where we used that $\frak A_\eth, \frak A_{\bf X}$ and $\frak
A_{\gamma^5}$ are antisymmetric on vectors (by Lem.~\ref{lemma:antisymm}, 
since $\eth, {\bf X}, \gamma^5$ all commute
with $\J$: $\eth$ and $\gamma^5$ by \eqref{eq:KOdim}
in $KO$-dim $4$, $\bf X$ by \eqref{eq:4.11}), and so symmetric
when evaluated on Gra{\ss}mann variables. On the other hand,
\eqref{eq:bilinear3} is symmetric on vectors (since $i{\bf Y}$
anticommutes with $\cal J$), while antisymmetric in Gra{\ss}mann
variables, so that \eqref{eq:bilinear3} is equal to 
\begin{equation}
  \label{eq:19}
  	2\frak A_{i\bf Y} (\tilde\phi_1, \tilde\xi_1) +
	2\,\frak A_{i\bf Y}(\tilde\phi_2, \tilde\xi_2).
\end{equation}

The lagrangian \eqref{eq:lagDir1} follows substituting all the bilinear forms in
\eqref{eq:64} and \eqref{eq:19} with their explicit expressions given
in  \eqref{derniere1}, \eqref{derniere2} and in Lem.~\ref{lem:4.6}.\end{proof}

In order to get Dirac equations, we have to possibilities for
identifying the physical spinors:
\begin{align}
\label{eq:spinors}\text{ either}\quad	&\Psi = \begin{pmatrix} \psi_l \\ \psi_r \end{pmatrix} 
	:= \begin{pmatrix} \tilde\zeta_1 \\
          \tilde\zeta_2 \end{pmatrix}, & &\Psi^\dag 
	= \left( \begin{array}{cc} \psi_l^\dag, & \psi_r^\dag \end{array} \right) := 
	\left( \begin{array}{cc} 
		-i\bar{\tilde{\varphi}}_1^\dag\sigma_2, & i\bar{\tilde\varphi}_2^\dag\sigma_2 
	\end{array} \right);\\ 
\label{eq:spinors1}
\text{ or }\quad &\Psi' = \begin{pmatrix} \psi'_l \\ \psi'_r \end{pmatrix} 
	:= \begin{pmatrix} \tilde\zeta_2 \\
          \tilde\zeta_1 \end{pmatrix},& & {\Psi'}^\dag 
	= \left( \begin{array}{cc} {\psi'_l}^\dag, & {\psi'_r}^\dag \end{array} \right) := 
	\left( \begin{array}{cc} 
		i\bar{\tilde{\varphi}}_2^\dag\sigma_2, & -i\bar{\tilde\varphi}_1^\dag\sigma_2 
	\end{array} \right)\! .
\end{align}
Imposing the complex parameter $d$ to be purely imaginary as
$d=im, m\in\mathbb R^*$,
   (in agreement with the non-twisted
case~\cite[Rem.~4.4]{W1}), the lagrangian \eqref{eq:lagDir1} becomes
\begin{align}
\label{eq:lagDir2}	
\text{ either}\;	&{\cal L}
	= i\psi_l^\dag \left(if_0 - \textstyle\sum_j\sigma_j\frak D_j \right)\psi_l
	+ i\psi_r^\dag \left(if_0 + \textstyle\sum_j\sigma_j\frak D_j \right)\psi_r
	+ m \left( \psi_l^\dag\psi_r + \psi_r^\dag\psi_l \right)\!,\\
\label{eq:lagDir3}
\text{ or }&{\cal L}'
	= i{\psi'}_r^\dag \left(if_0 -\textstyle\sum_j\sigma_j\frak D_j \right)\psi'_r
	+ i{\psi'}_l^\dag \left(if_0 + \textstyle\sum_j\sigma_j\frak D_j \right)\psi'_l
	+ m \left( {\psi'_r}^\dag\psi'_l + {\psi'_l}^\dag\psi'_r \right)\!.
\end{align}
The Euler-Lagrange equations for $\psi_l^\dag, \psi_r^\dag$ and
${\psi'_l}^\dag, {\psi'_r}^\dag$ yield the equation of motion
\begin{align}
  \label{eq:36}
 & i\left(if_0-\textstyle\sum_j \sigma_j {\frak D}_j\right)\psi_l + m\psi_r =0 , \quad
   i\left(if_0+ \textstyle\sum_j\sigma_j {\frak D}_j\right)\psi_r + m\psi_l =0 ,\\
  \label{eq:36bis}
 \text{ and }\quad & i\left(if_0+ \textstyle\sum_j \sigma_j {\frak D}_j\right)\psi'_l + m\psi'_r =0 , \quad
   i\left(if_0-\textstyle\sum_j \sigma_j {\frak D}_j\right)\psi'_r + m\psi'_l =0.
\end{align}

Which identification \eqref{eq:spinors} or \eqref{eq:spinors1} is meaningful is  fixed by the sign of $m$.
\begin{proposition}
  \label{prop:Diracfinal}
If $m<0$ (resp. $m>0$), then a plane wave solution of
\eqref{eq:36}  (resp. \eqref{eq:36bis}) coincides with a plane wave
solution of the Dirac equation with electromagnetic
potential $g_\mu$, in lorentzian signature and within Weyl 
temporal gauge (i.e. $\frak D_0=\partial_0$), with momentum $p$ such that
$p_0=-f_0$ (resp. $p_0=f_0$). 
\end{proposition}
\begin{proof}
A plane wave solution \eqref{eq:planewave} of \eqref{eq:36}  satisfies
  \begin{equation}
    \label{eq:93}
    i\left(if_0 +i \textstyle\sum_j \sigma_j(p_j + g_j)\right) \psi_{l} =
    -m\psi_{r}, \quad   i\left(if_0 -i \textstyle\sum_j \sigma_j(p_j + g_j)\right) \psi_{r} =
    -m\psi_{l}.
  \end{equation}
For $f_0=-p_0$, this is equivalent to the system of equations \eqref{eq:82}   satisfied by a plane wave solution
of the Dirac equation with mass $-m>0$, having previously substituted in 
\eqref{eq:80} the spatial derivative
$\partial_j$ with the covariant one $\frak D_j$. 
Similarly, a plane wave solution of \eqref{eq:36bis}~satisfies
\begin{equation}
  \label{eq:60}
    i\left(if_0 -i \textstyle\sum_j \sigma_j(p_j + g_j)\right) \psi'_{l} =
    -m\psi'_{r}, \quad   i\left(if_0 +i \textstyle\sum_j \sigma_j(p_j + g_j)\right) \psi'_{r} =
    -m\psi'_{l}.
\end{equation}
For $f_0=p_0$, this is equivalent to the Dirac equations \eqref{eq:82}
for mass $m>0$. \end{proof}

\noindent Identifying $x^0$ with the time direction $t$ of Minkowski space, then
$p_0$ is the energy of the plane wave. As for the double manifold, the $0^\text{th}$
component of the twisted fluctuation of the spectral triple of
electrodynamics gets interpreted as an energy.  

As for the Weyl equations, one may directly identify the lagrangian
density \eqref{eq:lagDir1} of the twisted fermionic action of
\emph{euclidean} electrodynamic with the \emph{lorentzian} Dirac lagrangian \eqref{L_M}
 (with covariant
        derivative $\frak D_\mu$, in the temporal
       gauge $\frak
D_0=\partial _0$): 
\begin{enumerate}
\item[-] either considering \eqref{eq:spinors}, and imposing
  $\partial_0\psi = if_0\psi$, so that \eqref{eq:lagDir2} coincides with
  \eqref{L_M} ;

  \item[-]  or using \eqref{eq:spinors1} and imposing
  $\partial_0\psi =- if_0\psi$, so that
  \eqref{eq:lagDir3} coincides with
 - \eqref{L_M}.
\end{enumerate}
\medskip

\begin{remark}
  The physical interpretation of $f_0, g_\mu$ is gauge invariant. From
  \eqref{eq:piuu'}, one gets
   \begin{small}\begin{equation}
    \label{eq:102}
    U:= \pi(u)J\pi(u)J^{-1}=
    \begin{pmatrix}
      A&0&0&0\\
      0&A'&0&0\\
      0&0&B'&0\\
      0&0&0&B
    \end{pmatrix}
    \begin{pmatrix}
      \bar B'&0&0&0\\
      0&\bar B&0&0\\
      0&0&\bar A&0\\
      0&0&0&\bar A'
    \end{pmatrix} = \begin{pmatrix}
      \Theta&0&0&0\\
      0&\Theta '&0&0\\
      0&0&\bar \Theta&0\\
      0&0&0&\bar \Theta'
    \end{pmatrix}
  \end{equation}
\end{small}

  \noindent where $\Theta:=\text{diag}(e^{i\theta}\,e^{i\theta'})$, $\Theta':=\text{diag}(e^{i\theta'}\,e^{i\theta})$
  with $\theta, \theta'$ as in \eqref{eq:107}.  Imposing the
  gauge transformation to preserve selfadjointness, that
  is
  $\theta=\theta'$  (disregarding the constant), then $U$ is simply the
  multiplication by a phase. This means that $U\eta$ is still in $\cal
  H_{\cal R}$, so that the computation of the fermionic action $\frak
  A^\rho_{D_{\rho(U)\cal ZU}}(\widetilde{U\eta}, \widetilde{U\eta})$ is
  similar as above.
\end{remark}

\subsection{Identification of the physical degrees of freedom}
\label{rem:identify}
The relation between the components $\xi:=\left(
  \begin{array}{c}
\zeta\\ \zeta
  \end{array}\right)$, $\phi:=\left(
  \begin{array}{c}
  \varphi\\ \varphi
  \end{array}\right)$ of the eigenvector $\eta$ of
$\cal R$ and the
physical spinors $\Psi=
\begin{pmatrix}
  \psi_l\\ \psi_r
\end{pmatrix}
$,  $\Psi^\dagger =\begin{pmatrix}
  \psi_l^\dagger\\ \psi_r^\dagger
\end{pmatrix}$ is encoded within the rule of identification 
 \eqref{eq:48}  (with the sign discussed below 
Prop. \ref{Prop:Weyl}) for the double manifold, that we write
 equivalently as
  \begin{equation}  
\label{eq:defspinphys}
     \Psi = \tilde\xi, \quad \Psi^\dag =  i({\cal
       J}\tilde\phi)^\dag; 
  \end{equation}
and the rules 
 (\ref{eq:spinors}, \ref{eq:spinors1})  for the spectral triple of electrodynamics, that
 writes equivalently
 \begin{equation}
    \label{eq:71}
    \begin{array}{ll}
\Psi =\tilde \Xi &\Psi^\dag = i({\cal J}\tilde\phi)^\dag\\[4pt]
\Psi' =\gamma^0\tilde \Xi &\Psi'^\dag = i({\cal J}\tilde\phi)^\dag\gamma^0=-i({\cal J}\gamma^0\tilde\phi)^\dag
    \end{array}
 \quad \text{ with }
   \quad \Xi:=\left(
  \begin{array}{c}
    \zeta_1\\ \zeta_2
  \end{array}\right), \quad \phi:=\left(
  \begin{array}{c}
   \varphi_1\\ \varphi_2
  \end{array}\right).
 \end{equation}
In any case, the physical spinors are completely determined by  the
projection $\eta_+$ of $\eta$ on the $+1$ eigenspace ${\cal H}_+$ of the
grading operator, that is
\begin{align}
  \label{eq:96}
  &\eta_+=\begin{pmatrix}
    \varphi\\0
  \end{pmatrix}\otimes e + \begin{pmatrix}
    0\\\zeta
  \end{pmatrix}\otimes \bar e &\text{projecting }\eqref{eq:eta},\\
&\eta_+=  \begin{pmatrix}
    \varphi_1\\0
  \end{pmatrix}\otimes e_L + \begin{pmatrix}
    0\\\varphi_2
  \end{pmatrix}\otimes e_R + \begin{pmatrix}
    0\\ \zeta_1
  \end{pmatrix}\otimes \overline{e_L} + \begin{pmatrix}
   \zeta_2\\0
  \end{pmatrix}\otimes \overline{e_R}& \text{ projecting } \eqref{6.3}.
\end{align}
This is similar to the non-twisted case, where the physical spinors
are determined by an eigenvector in ${\cal H}_+$.

\section{Lorentz invariance}
\label{sec:Lorentzinv}
\setcounter{equation}{0}

So far, our results do not say anything on the components $f_i$ of the
twisted fluctuation for $i=1,2,3$, because they do not appear in the lagrangian \eqref{eq:lagDir1}.
Since $f_0$ identifies with an energy, it is tempting to identify
$f_i$ with a momentum.  This is actually achieved by acting with Lorentz
transformations on the twisted fermionic action. 

More precisely, we first define in (\S \ref{subsec:Lorentzinv}) a
action of Lorentz boosts on the twisted spectral triple which leaves
 the twisted fermionic action  invariant. 
We then investigate the action from the point of view of a boosted
observers, both for the double manifold in \S \ref{subsec:boostWeyl} and
for electrodynamics in \S\ref{subsec:boostDirac}. In both cases, we
obtain equations of motion in which the components $f_i$ of the twisted
fluctuation gets interpreted as a momenta.

\subsection{Lorentz invariance of the twisted fermionic action}
\label{subsec:Lorentzinv}

As recalled in appendix \ref{appC}, the Dirac equation on Minkowski spacetime  is invariant under the
action \eqref{eq:boost}
of boosts  simultaneously on spinors and on the Dirac operator.
From a mathematical point of view, this  action makes sense on an
euclidean spin manifold $\M$ as well:  although this might  seem physically
non-relevant at first sight, we let boosts act on euclidean
spinors and on the euclidean Dirac operator as
\begin{align}
\label{eq:67}
 \phi \; &\to \;\phi^\Lambda:= S[\Lambda]\phi \quad\qquad\quad \forall \phi\in L^2(\M,{\cal S}), \\
\label{eq:67-bis} \eth \; &\to \;  \eth^\Lambda:= S[\Lambda] \;\eth \; S[\Lambda]^{-1}.
 \end{align}

As an element of ${\cal B}(L^2(\M,S))$, the boost operator
$S[\Lambda]$  is acted
upon by the inner autormorphism $\rho$ induced by ${\cal R}=\gamma^0$
given in
        \eqref{eq:definR}, namely
\begin{equation}
\label{3.7}
	\rho(S[\Lambda]) =\gamma^0\begin{pmatrix} \Lambda_- & 0_2 \\ 0_2 & \Lambda_+ \end{pmatrix} \gamma^0
	= \begin{pmatrix} \Lambda_+ & 0_2 \\ 0_2 & \Lambda_- \end{pmatrix}.
\end{equation}
Since $\Lambda_+$, $\Lambda_-$ are
inverse of one another and selfadjoint, one has
\begin{equation}
  \label{eq:72}
  \rho(S[\Lambda])=S[\Lambda]^{-1},\quad  S[\Lambda]^+=S[\Lambda]^{-1}.
\end{equation}
 
 \begin{lemma}
\label{prop:jboost}
The real structure $\cal J$
\eqref{eqn:3.9} twist-commutes with boosts:
\begin{equation}
\label{eq:JcomLamThet}	
{\cal J}S[\Lambda]= S[\Lambda]^{-1} {\cal J}.
\end{equation}
\end{lemma}	
	\begin{proof}
Since $\sigma_2$ anticommutes with $\sigma_1$, $\sigma_3$ and
commutes with itself, one has 
        \begin{equation}
          \label{eq:74}
          ({\bf n}.\boldsymbol\sigma)\sigma_2 
	= \sigma_2(-n_1\sigma_1 +n_2\sigma_2- n_3\sigma_3) = -\sigma_2\overline{({\bf n}.\boldsymbol\sigma)},
        \end{equation}
where we use $\sigma_1, \sigma_3=\bar\sigma_1, \bar\sigma_3$, $\bar\sigma_2=-\sigma_2$. Hence
	$\Lambda_\pm\sigma_2 
	= \sigma_2\bar\Lambda_\mp$.
With   $\J =\text{diag}(-\sigma_2, \sigma_2)cc$, one gets
\begin{equation*}
\label{eq:jslambda}
	\J S[\Lambda] = 
	\left( \begin{array}{cc} 
		-\sigma_2\bar\Lambda_- & 0 \\ 0 & \sigma_2\bar\Lambda_+
	\end{array} \right)\!cc\, =
	\left( \begin{array}{cc} 
		-\Lambda_+\sigma_2 & 0 \\ 0 & \Lambda_-\sigma_2
	\end{array} \right)\!cc= \quad S[\Lambda]^{-1}\J .
\end{equation*}

\vspace{-.95truecm}
\end{proof}

\medskip
 The inner product on
$L^2(\M,S)$ is not invariant by  \eqref{eq:67}, the 
 twisted
product is:
\begin{equation}
\label{eq:refboostprod}
\langle S[\Lambda]\phi, S[\Lambda]\xi \rangle_{\rho}
  = \langle \phi, S^+[\Lambda]\, S[\Lambda]\xi \rangle_{\rho}
  = \langle \phi, S[\Lambda]^{-1}S[\Lambda]\xi \rangle_{\rho}   = \langle \phi, \xi \rangle_{\rho}
\end{equation}
 for any $\psi, \phi \in
        L^2(\M,S)$.  This is not a surprise, being the twisted
product  the Krein product of lorentzian spinors (see \S\ref{sec:secfermion}).
Yet, the bilinear form
${\frak A}^\rho_\eth$ is not invariant:
\begin{equation*}
  \label{eq:5}
    {\frak A}^\rho_{\eth^\Lambda}(\phi^{\Lambda}, \xi^\Lambda)
=  \langle \J\,S[\Lambda]\phi ,\eth_\Lambda
    S[\Lambda]\xi\rangle_\rho,\\
=  \langle S[\Lambda]^{-1} \,\J\phi ,S[\Lambda]\eth
    \xi\rangle_\rho=\langle\J\phi ,S[\Lambda]^2\eth
    \xi\rangle_\rho\neq {\frak A}^\rho_{\eth}(\phi, \xi).
\end{equation*}
This can be corrected by making boosts act on the physical
spinors $\Psi$, $\Psi^\dag$ \eqref{eq:defspinphys}. Namely,
\begin{align}
  \label{eq:79}
&\Psi\to   S[\Lambda]\Psi =S[\Lambda]\tilde \zeta,\\ 
  \label{eq:84}
    &\Psi^\dagger \to  \Psi^\dag S[\Lambda]^\dag = i
    ({\cal J} \tilde\phi)^\dag S[\Lambda]^\dag = i
    (S[\Lambda]{\cal J}\tilde\phi)^\dag =  i({\cal J}
    S[\Lambda]^{-1}\tilde\phi)^\dag.
\end{align} 
Consequently, in order to ``boost the fermionic action'', instead of $\phi^\Lambda$ one should consider
\begin{equation}
  \label{eq:26}
  \phi^{-\Lambda}:=S[\Lambda]^{-1}\phi.
\end{equation}
As a matter of fact, one checks that 
  \begin{align}
    \label{eq:21}
   {\frak A}^\rho_{\eth^\Lambda}(\phi^{-\Lambda}, \xi^\Lambda)
&=  \langle \J\,S[\Lambda]^{-1}\phi ,\eth_\Lambda
    S[\Lambda]\xi\rangle_\rho,=  \langle S[\Lambda] \,\J\phi ,S[\Lambda]\eth
    \xi\rangle_\rho
   =  \langle \J\phi,  \eth\xi \rangle_\rho=  {\frak A}^\rho_{\eth}(\phi, \xi),
  \end{align}
and the same holds true for the operator
\begin{equation}
\eth_X^\Lambda:=S [\Lambda] \,\eth_X \, S[\Lambda]^{-1}= \eth^\Lambda +
{\bf X}^\Lambda
\label{eq:32}
\quad\text{ with }\quad
{\bf X}^\Lambda:= S[\Lambda]\,{\bf X}\,S[\Lambda]^{-1},
\end{equation}
obtained by the action of boosts on the twisted-covariant Dirac
operator $\eth_X$. Therefore
\begin{proposition}
\label{prop:lorentz}
  The twisted fermionic action on an euclidean manifold
  \eqref{eq:56} is invariant under the boost action
  \begin{equation}
    \xi\to \xi^\Lambda,\quad
    \phi\to\phi^{-\Lambda}, \quad \eth_X\to \eth_X^\Lambda,
    \label{eq:88}
  \end{equation}
  followed by the identification $\phi^\Lambda=\zeta^\Lambda$, that is
  \begin{equation}
    \label{eq:27}
    {\frak A}^\rho_{\eth_X}(\tilde \xi, \tilde \xi)= \frak A^\rho_{\eth_X^\Lambda}(\tilde \xi^{-\Lambda}, \tilde\xi^\Lambda).
  \end{equation}
\end{proposition}

Our claim is that the right
hand side of the equation above is the action as seen from a boosted
observer. Of course, in order to get the Weyl and Dirac equations, one needs to
double the manifold as before, 
then add a mass matrix. Still, the main
features of the boosting are visible on \eqref{eq:27}. In
particular, by computing explicitly the bilinear form $\frak
A^\rho_{\eth_{\bf X}^\Lambda}$, one
sees all the components $f_\mu$ of the twisted fluctuation appearing
in the action. To this aim we use the following notations for the
boosted spinors.
\begin{definition}
\label{defi:notaboost}
 Given $\xi=\left(\begin{array}{c}\zeta\\\zeta \end{array}\right)$,
$\phi=\left(\begin{array}{c}\varphi\\\varphi\end{array}\right)$ in
$\HH_{\cal R}$, we write $\varphi_{l,r}$, $\zeta_{l,r}$ be the
  components of
  \begin{align*}
   \xi^\Lambda=S[\Lambda]\xi =\left(
      \begin{array}{c}
        \Lambda_- \zeta\\ \Lambda_+\zeta 
      \end{array}\right)=: \left(
    \begin{array}{c}
      \zeta_l\\ \zeta_r 
    \end{array}\right),\qquad
  {\cal J}\phi^{-\Lambda}=
      S[\Lambda]	{\cal J} \phi  
      = \!\left( \begin{array}{cc}
                   -\Lambda_-	\sigma_2\,\bar \varphi\\ \Lambda_+\sigma_2\,\bar\varphi
                 \end{array} \right)\! =:  \left( \begin{array}{cc}
                                                    \bar \varphi_l\\\bar\varphi_r
                                                  \end{array} \right) .
  \end{align*}
\end{definition}

\begin{proposition} 
\label{prop:propboost}
Let $\sigma^\mu_\Lambda:=\Lambda_-\sigma^\mu\Lambda_-$ and
$\tilde\sigma^\mu_\Lambda:=\Lambda_+\tilde\sigma^\mu\Lambda_+$. Then
\begin{equation}
  \label{eq:55}
   \frak
   A^\rho_{\eth_X^\Lambda}(\phi^{-\Lambda},\xi^\Lambda)=-i\int_{\cal
     M}\!\!\!\emph{d}\upmu\;  {\bar\varphi_l}^\dag \, \tilde\sigma^\mu_\Lambda
   (\partial_\mu + f_\mu)\zeta_l+ 
 {\bar\varphi_r}^\dag \sigma^\mu_\Lambda(\partial_\mu -  f_\mu)\zeta_r.
\end{equation}
\end{proposition}
\begin{proof} 
Defining
\begin{align}
\gamma^\mu_\Lambda:= S [\Lambda]
\gamma^\mu \, S[\Lambda]^{-1}&=\left(
  \begin{array}{cc}
0&    \Lambda_-\sigma^\mu \Lambda_- \\\Lambda_+\tilde\sigma^\mu\Lambda_+&0
  \end{array}\right)
=\left(
  \begin{array}{cc}
0&   \sigma^\mu_\Lambda\\\tilde\sigma^\mu_\Lambda &0
  \end{array}\right),
\label{eq:30}
\end{align}
one has (remembering that $\partial_\mu$ and $\gamma^5$ commutes with $S[\Lambda]$)
\begin{equation*}
\label{eq:DPhiBoost}
	\eth^\Lambda \xi^\Lambda 
	= -i\gamma^\mu_\Lambda \partial_\mu \xi^\Lambda
	= -i\!\left(\begin{array}{c}
		\sigma^\mu_\Lambda\partial_\mu\zeta_r \\[4pt] \tilde\sigma^\mu_\Lambda\partial_\mu\zeta_l
	\end{array}\right)\quad \text{ and }\quad
	{\bf X}^\Lambda \xi^\Lambda 
	= -i\gamma^\mu_\Lambda f_\mu \gamma^5 \xi^\Lambda
	= -i f_\mu\left( \begin{array}{r}
		-\sigma^\mu_\Lambda \zeta_r \\
                           \tilde\sigma^\mu_\lambda \zeta_l
	\end{array} \right)\!.
\end{equation*}
Since $({\cal
  J}\phi)^\Lambda={\cal J} \phi^{-\Lambda}$, one~gets with \eqref{eq:definR} 
\begin{align*}
  &\frak A^\rho_{\eth_\Lambda}(\phi^{-\Lambda}, \xi^\Lambda)=\langle ({\cal  J} \phi)^\Lambda, 
 {\cal R}\eth^\Lambda \xi^\Lambda\rangle=\!-i\!\int_{\cal M}\!\!\!\!\emph{d}\upmu\left({\bar\varphi_l}^\dag,
  {\tilde\varphi_r}^\dag\right)\gamma^0\left(
  \begin{array}{c}
    \sigma^\mu_\Lambda \partial_\mu \zeta_r\\
    \bar\sigma^\mu_\Lambda\partial_\mu\zeta_l\end{array}\right)\!=\!
  -i\int_{\cal M}\!\!\!\!\emph{d}\upmu\!\left(\bar\varphi_l^\dag\tilde\sigma^\mu_\Lambda\partial_\mu \zeta_l + \bar\varphi_r^\dag\sigma^\mu_\Lambda\partial_\mu \zeta_r\right)\!\!,\\
&\frak A^\rho_{{\bf X}^\Lambda}(\phi^{-\Lambda},\xi^\Lambda)=\langle ({\cal  J} \phi)_\Lambda, 
  {\cal R}{\bf X}^\Lambda\xi^\Lambda\rangle=-i\!\int_{\cal M}\!\!\!\emph{d}\upmu\;\left({\bar\varphi_l}^\dag\,,
   {\bar\varphi_r}^\dag\right) \gamma^0\left(
   \begin{array}{c}
     -\sigma^\mu_\Lambda f_\mu\zeta_r\\
     \tilde\sigma^\mu_\Lambda 
     f_\mu\zeta_l
\end{array}\right)=
  -i \int_{\cal M}\!\!\!\!\emph{d}\upmu\left(\bar\varphi_l^\dag\tilde\sigma^\mu_\Lambda f_\mu \zeta_l -
  \bar\varphi_r^\dag\sigma^\mu_\Lambda f_\mu \zeta_r\right).
\end{align*}
The results follows summing these two equations.
\end{proof}

\begin{remark}
\label{rem_noboost}
One checks that for $S[\Lambda]=\mathbb I$ (no
  boost), proposition \ref{prop:propboost} gives back
  proposition~\ref{prop:actspecmanif}: one then has
 $\zeta_{l,r}= \zeta$ while  $\bar\varphi_l^\dag=-\bar\varphi^\dag \sigma_2$ and
  $\bar\varphi_r=\bar\varphi^\dag \sigma_2$,
  so by \eqref{eq:sumsigma}
\begin{align*}
 \frak
   A^\rho_{\eth_{\bf X}^\Lambda}(\phi^{-\Lambda},\xi^\Lambda)&=-i\int_{\cal M}\!\!\!\emph{d}\upmu\; -\bar\varphi^\dag\sigma_2 \,  \tilde\sigma^\mu(\partial_\mu +
   f_\mu)\zeta  + 
 \bar\varphi^\dag\sigma_2 \,\sigma^\mu
   (\partial_\mu - f_\mu)\zeta,\\
&=-i\int_{\cal M}\!\!\!\emph{d}\upmu\; \bar\varphi^\dag\sigma_2 
\left( \left(\sigma^\mu-\tilde\sigma^\mu\right)\partial_\mu - \left(\tilde\sigma^\mu + \sigma^\mu\right)  f_\mu\right)\zeta
=2\int_{\cal M}\!\!\!\emph{d}\upmu\; \bar\varphi^\dag\sigma_2 \left(
  -\sum_{j=1}^3 \sigma_j\partial_j + if_0\right)\zeta.
\end{align*}\end{remark}
The twisted fermionic action \eqref{eq:27} on an euclidean manifold, as seen from a
boosted observer, is obtained putting $\phi^{-\Lambda}=
\zeta^{-\Lambda}$ in \eqref{eq:55}, then turning the entries of the spinors into Gra\ss mann
variables. As in \S
\ref{subsec:fermioncactionmanif}, there is no enough spinor degrees of
freedom to identify a physically meaningful action.
 We thus consider
the boost of the action \eqref{eq:57} of the doubled-manifold.  

\subsection{Weyl equations for  boosted observers}
\label{subsec:boostWeyl}

In agreement with \eqref{eq:79} and \eqref{eq:26}, we define the action of a boost on
$L^2(\M,S)\otimes \mathbb C^2$ as $\phi\otimes e +  \psi\otimes \bar e
\to (S[\Lambda]^{-1}\phi)\otimes e + (S[\Lambda]\psi)\otimes \bar e$,
in such a way that $\eta\in {\cal H_R}$ in \eqref{eq:eta} is mapped to  
      \begin{equation}
        \label{eq:100}
\eta^\Lambda    =  \phi^{-\Lambda}\otimes e + \xi^\Lambda
        \otimes \bar e.
      \end{equation}
      \begin{proposition} The action of
        a double manifold \eqref{eq:57}, as seen from a boosted observer,~is 
\begin{align}
  \label{eq:57boost}
\frak A^\rho_{\eth_{\bf X}^\Lambda}(\tilde\eta^{\Lambda},\tilde \eta^\Lambda) 
  = 2\;\frak A^\rho_{\eth_{\bf X}^\Lambda}(\tilde\phi^{-\Lambda},\tilde \xi^\Lambda) 
  =-2i\int_{\cal
     M}\!\!\!\emph{d}\upmu\;  {\tilde\varphi_l}^\dag \, \tilde\sigma^\mu_\Lambda
   (\partial_\mu + f_\mu)\zeta_l+ 
 {\tilde\varphi_r}^\dag  \sigma^\mu_\Lambda(\partial_\mu -  f_\mu)\zeta_r.
\end{align}
\end{proposition}
\begin{proof}
Following
the analysis below Prop. \ref{prop:lorentz}, the twisted fermionic
action from a boosted
observer is  ${\frak A}^\rho_{\eth_{\bf X}^\Lambda\otimes \mathbb I_2}(\tilde\eta^\Lambda, \tilde\eta^\Lambda)$.
By a calculation similar to the one of Prop. \ref{prop:3.6}.
one obtains 
\begin{align}
  \label{eq:97}
 {\frak A}^\rho_{\eth_{\bf X}^\Lambda\otimes \mathbb I_2}(\eta^\Lambda, \eta'^\Lambda) = 
\frak A^\rho_{\eth_{\bf X}^\Lambda}(\phi^{-\Lambda},  \xi'^\Lambda) + \frak A^\rho_{\eth_{\bf X}^\Lambda}(\xi^\Lambda,
\phi'^{-\Lambda}).
\end{align}
 By boost invariance \eqref{eq:21}, the terms
 in the r.h.s  have the same symmetry as the corresponding
 expression without $\Lambda$, that is 
symmetric on gra\ss manian vectors. Thus, similar to \eqref{eq:81}, one~gets
    ${\frak A}^\rho_{\eth_{\bf X}^\Lambda\otimes \mathbb
      I_2}(\tilde\eta^\Lambda, \tilde \eta^\Lambda) = 2\frak A^\rho_{\eth_{\bf X}^\Lambda}(\tilde\phi^{-\Lambda},  \tilde\zeta^\Lambda)$.
 The result follows from~Prop.~\ref{prop:propboost}
\end{proof}

We identify the boosted physical degrees of freedom 
$\Psi = S[\Lambda]\tilde\zeta$,
$\Psi^\dag = i({\cal J} S[\Lambda]^{-1}\tilde\phi)^\dag$  following
(\ref{eq:79}, \ref{eq:84}). In
components (see Defi. \ref{defi:notaboost}), one has
\begin{equation}
  \label{eq:91c}
  \psi_{l,r}= \tilde \zeta_{l,r},\quad
 \psi_{l,r}^\dag = i\tilde\varphi_{l,r}^\dag.
\end{equation}
The lagrangian density in \eqref{eq:57boost} then reads
\begin{equation}
  \label{eq:85bisb}
{\cal L}_\Lambda = -2\left(\psi^\dagger_l\tilde\sigma^\mu_\Lambda(\partial_\mu + f_\mu)\psi_l
  + \psi^\dagger_r\sigma^\mu_\Lambda(\partial_\mu-f_\mu)\psi_r\right).
\end{equation}
Treating $\psi_l, \psi_r, \psi_l^\dag$ and $\psi_r^\dag$ as
independent fields, the corresponding equations of motion are
\begin{align}
  \label{eq:2}
  \tilde\sigma^\mu_\Lambda (\partial_\mu + f_\mu)\psi_l=0  ,\quad \sigma^\mu_\Lambda (\partial_\mu - f_\mu)\psi_r=0.
\end{align}

\begin{proposition}
\label{Prop:Weylboosted}
  For a constant twisted fluctuation $f_\mu$, a plane wave solution of the first 
  (resp. the second) eq. \eqref{eq:2} coincides with a plane wave solution of the
  left (resp. right) handed Weyl equation whose (dual) momentum
  $p^\sharp$ has components 
   $ p'_\nu =\Lambda_\nu^\mu p_\mu$ in the
  boosted frame, where
  \begin{equation}
\label{eq:conWeylboost}
    p_0= - f_0,\quad 
    p_j=f_j, \quad \text{ resp.}\quad
    p_0=  f_0,\quad
    p_j=-f_j.
  \end{equation}
  \begin{proof}
By (\ref{eq:85}, \ref{eq:86}), a plane wave solution \eqref{eq:50}
of the first eq. \eqref{eq:2} satisfies 
\begin{align}
 0= \tilde \sigma^\mu_\Lambda\left(-ip_\mu+ f_\mu\right)\psi_{l}  &=\left(\Lambda^0_\nu
  \tilde\sigma^\nu_M\left(-ip_0 + f_0\right)   -i\Lambda^j_\nu\tilde\sigma^\nu_M\left(-ip_j +f_j\right)\right)\psi_{l},\\
  &= - \tilde\sigma^\nu_M\left(\Lambda^0_\nu\left(ip_0-f_0\right) +
    \Lambda^j_\nu\left(p_j+if_j\right)\right)\psi_{l}.
   \label{eq:25}
\end{align}
Similarly, a plane wane solution $\psi_r$ of the second eq. \eqref{eq:2} satisfies
\begin{align}
 0= -\sigma^\nu_M \left(\Lambda^0_\nu\left(ip_0+f_0\right) +  \Lambda^j_\nu\sigma^\nu_\M\left(p_j -if_j\right)\right) \psi_{r}.
   \label{eq:87}
\end{align}
If \eqref{eq:conWeylboost} holds, these two equations become
\begin{align}
  \label{eq:43}
  &0=-(1+i)\tilde\sigma^\nu_M \left(\Lambda^0_\nu p_0 +
   \Lambda^j_\nu p_j)\right) \psi_{0l}=  -(1+i)\tilde\sigma^\nu_M \,
    p'_\nu \psi_{0l},\\
 \label{eq:43bbbis}
 &0=-(1+i)\sigma^\nu_M \left(\Lambda^0_\nu p_0+ \Lambda_\nu^jp_j \right) \psi_{0r}=-(1+i)\sigma^\nu_Mp'_\nu \psi_{0r},
\end{align}
which coincide, up to  a constant factor, with the Weyl equations of
motion \eqref{eq:52} for 
a boosted observer.
  \end{proof}
\end{proposition}

Prop. \ref{Prop:Weylboosted} is the boosted version of Prop.~\ref{Prop:Weyl}:
now the whole field $f_\mu dx^\mu$ (and not only
its $0^\text{th}$ component) identifies
with the dual $p^\sharp$ of the energy-momentum $4$-vector.
Nevertheless, the  interpretation of the lagrangian density
  \eqref{eq:85bisb} is delicate, because of the sign difference
 in \eqref{eq:conWeylboost}:
\begin{equation}
  \label{eq:63}
f_0= -i \partial_0 ,\quad f_j= i\partial _j; \quad
 \text{versus} \quad f_0= i\partial_0,\quad f_j=-i\partial_j. 
\end{equation}
Substituting the first (resp. second) of these equations in the left (resp. right) handed
part of \eqref{eq:85bisb}, one obtains
\begin{equation}
  \label{eq:68}
 2(i-1)\left(\psi^\dagger_l\tilde\sigma^\mu_M\partial_\nu'\psi_l
  + \psi^\dagger_r\sigma^\mu_M\partial'_\mu\psi_r\right) \quad \text{
  with } \partial'_\nu :=\Lambda_\nu^\mu\partial_\mu.
\end{equation}
This agrees with the equations of motion (\ref{eq:43},
\ref{eq:43bbbis}) (remembering that  $\partial'_\nu =-ip'_\nu$ and
the factor $-2$ that was ignored from
\eqref{eq:85bisb} to \eqref{eq:2}), thus suggesting that ${\cal
  L}_\Lambda$ is the sum - up to a complex factor - of
  the two Weyl lagrangians ${\cal L}_\M^l, {\cal L}_\M^r$
  (\ref{eq:WeyllagL}).
  The point is that $\psi_l, \psi_r$ comes from the
  action of $\Lambda_\mp$ on the same Weyl spinor
  $\varphi$, and this action leaves the exponential part of the plane
  wave unaltered. So $\psi_l$, $\psi_r$ should describe two
  plane waves with the same momenta, in contradiction with
  \eqref{eq:63}. We comment on this point in the conclusion.

  \begin{remark}
The no-boost limit of the action \eqref{eq:57boost} yields back the action
    \eqref{eq:57} of the double manifold (along the lines of remark
    \eqref{rem_noboost}).  As well for the lagrangian: 
    identifying 
    $\psi_l,\psi_r\to\zeta$ with $\psi$ in \eqref{eq:48},
    $\psi^\dag_l=i\bar\varphi^\dag_l\to -i\bar\varphi^\dag \sigma_2$
    with $\psi^\dag$, and
   $\psi_r^\dag=i\bar\varphi_r^\dag\to
    i\bar\varphi^\dag\sigma_2$ with $-\psi^\dag$, then
    \eqref{eq:85bisb} becomes
    $4i\left(\psi^\dag\left(if_0-\sum_j\sigma_j\partial_j\right)\right)$,
  in agreement with \eqref{eq:24} (the expression with the opposite sign is obtained
  identifying the no-boost limit of $\psi_r^\dag$ with $\psi$, and
  the one of $\psi_l^\dag$ with $-\psi$).
  \end{remark}

\subsection{Dirac equation for boosted observers}
\label{subsec:boostDirac}

A boost $S[\Lambda]$  acts on the twisted covariant Dirac operator ${\cal D}_Z$ of
electrodynamics \eqref{eq:4.34} to
give
\begin{equation}
  \label{eq:89}
  {\cal D}_{\cal Z}^\Lambda:= S[\Lambda] {\cal D}_{\cal Z}
  S[\Lambda]^{-1} = \eth^\Lambda\otimes \mathbb I_4 + \gamma^5
  \otimes D_{\cal F} + {\bf X}^\Lambda\otimes \mathbb I' + i {\bf Y}^\Lambda
  \otimes \mathbb I'',
\end{equation}
where $\eth^\Lambda$, ${\bf X}^\Lambda$ are defined in
\eqref{eq:32}, we used $S[\Lambda]\gamma^5
  S[\Lambda]^{-1}=\mathbb I$ and 
define (using notations~\eqref{eq:30}) 
\begin{equation}
    \label{eq:90}
    {\bf Y}^\Lambda := S[\Lambda]\, {\bf Y} \, S[\Lambda]^{-1} =
    -iS[\Lambda] \gamma^\mu g_\mu\mathbb I_4 S[\Lambda]^{-1}
    =-i\gamma^\mu_\Lambda g_\mu\mathbb I_4.
  \end{equation}
Similarly to what has been done for the double manifold in
\eqref{eq:100}, we make the boost acts on $L^2(\M,S)\otimes \mathbb
C^4$ in such a way that $\eta\in{\cal H_R}$ in \eqref{6.3} is mapped to 
\begin{equation}
  \label{eq:101}
  \eta^\Lambda = \phi_1^{-\Lambda}\otimes e_L +
  \phi_2^{-\Lambda}\otimes e_R + \zeta_1^\Lambda\otimes \overline{e_L}
  + \zeta_2^{\Lambda}\otimes \overline{e_R}.
\end{equation}
\begin{proposition}  The fermionic action from the minimal twist of
  electrodynamics, as seen from a boosted observer, is the integral 
  \begin{equation}
    \label{eq:103}
     \frak A^\rho_{{\cal D}_{\cal Z}}(\tilde\eta^\Lambda,
     \tilde\eta^\Lambda)=-2 \int_{\cal M}\!\!\!\!\emph{d}\upmu\;  {\cal L}_\Lambda
  \end{equation}
of the lagrangian density
    \begin{equation}
  \begin{split}
  \label{eq:92ter}
     {\cal L}_\Lambda=i \left(\tilde{\bar\varphi}_{1l}^\dag\,
        \tilde\sigma^\mu_\Lambda({\cal D}_\mu+ f_\mu)
        \,\tilde\zeta_{1l}
        +\tilde{\bar\varphi}_{1r}^\dag\,\sigma^\mu_\Lambda({\cal
          D}_\mu-  f_\mu)\tilde\zeta_{1r}\right)  
+ d\left( 
 \tilde{\bar\varphi}_{2l}^\dag {\tilde\zeta}_{1r} - \tilde{\bar\varphi}_{2r}^\dag\tilde{\zeta}_{1l}\right) \\ 
+i\left(\tilde{\bar\varphi}_{2l}^\dag\, \tilde\sigma^\mu_\Lambda({\cal D}_\mu- f_\mu) \,\tilde\zeta_{2l}
       + \tilde{\bar\varphi}_{2r}^\dag\,\sigma^\mu_\Lambda({\cal D}_\mu+
        f_\mu)\tilde\zeta_{2r}\right) + \bar d\left( 
 \tilde{\bar\varphi}_{1l}^\dag {\tilde\zeta}_{2r} - \tilde{\bar\varphi}_{1r}^\dag\tilde{\zeta}_{2l}\right) ,
  \end{split}
    \end{equation}
where ${\cal D}_\mu =\partial_\mu -i g_\mu$.
\end{proposition}
\begin{proof}
The computation is similar to the one of 
Prop. \ref{prop:eq-Dirac}. One obtains
\begin{align*}
 \frak A^\rho_{{\cal D}_{\cal Z}}(\tilde\eta^\Lambda, \tilde\eta^\Lambda)&= 2\,\frak
    A^\rho_{\eth^\Lambda}(\tilde {\phi}_1^{-\Lambda}, \tilde
    \xi_1^\Lambda ) +2\,\frak A^\rho_{\eth^\Lambda}(\tilde \phi_2^{-\Lambda},
    \tilde \xi_2^\Lambda) +2\,\frak A^\rho_{{\bf X}^\Lambda}(\tilde
    \phi_1^{-\Lambda}, \tilde \xi_1^{\Lambda})-2\,\frak A^\rho_{{\bf  X}^\Lambda}(\tilde
    \phi_2^{-\Lambda}, \tilde \xi_2^\Lambda)\\
  & - 2\frak A^\rho_{i{\bf Y}^\Lambda}
    (\tilde\phi_1^{-\Lambda}, \tilde\xi_1^\Lambda) - 2\,\frak
    A^\rho_{i{\bf Y}^\Lambda}(\tilde\phi_2^{-\Lambda},
    \tilde\xi_2^{\Lambda}) + 2\bar d\,\frak A^\rho_{\gamma^5}(\tilde\phi_1^{-\Lambda},
    \tilde\xi_2^\Lambda) + 2d\,\frak A^\rho_{\gamma^5}(\tilde\phi_2^{-\Lambda},
    \tilde\xi_1^{-\Lambda}),
  \end{align*}
where we used that the bilinear forms $\frak A^\rho_{\eth^\Lambda}$, $\frak
A^\rho_{{\bf X}^\Lambda}$, $\frak A^\rho_{\gamma^5}$ and $\frak
    A^\rho_{i{{\bf Y}^\Lambda}}$ valued on $\tilde \phi_i^{-\Lambda}$,
    $\tilde \zeta_j^\Lambda$ has the same symmetry properties of the
    corresponding expressions without $\Lambda$ (by the invariance
    \eqref{eq:21}, that holds also for ${\bf X}^\Lambda$,
    $\gamma^5$, and $i{\bf Y}^\Lambda$).  Substituting $\frak
    A^\rho_{\eth^\Lambda}$ and $\frak
    A^\rho_{{\bf X}^\Lambda}$ with their explicit form given in
    the proof of Prop. \ref{prop:propboost} and
    calculating (with ${\cal J}\phi^{-\Lambda}= ({\cal
      J}\phi)^\Lambda$ given in Def. \ref{defi:notaboost})
\begin{align}
\nonumber
 {\frak A}^\rho_{i{\bf Y}^\Lambda}(\phi^{-\Lambda},
  \zeta^\Lambda)&=\langle{\cal J} \phi^{-\Lambda}, \gamma^0i{\bf Y}^\Lambda \xi^\Lambda\rangle  = \left( \begin{array}{c}
 	\bar\varphi_l \\ \bar\varphi_r 
 	\end{array} \right)^{\!\!\dag} \gamma^0 
 	\left( \begin{array}{r} 
 		g_\mu \sigma^\mu_\Lambda \zeta_r\\ g_\mu \tilde\sigma^\mu_\Lambda \zeta_l 
 	\end{array} \right)
=\int_{\cal M}\!\!\!\!\emph{d}\upmu \,g_\mu\left(  \bar\varphi_l^\dag \tilde\sigma^\mu_\Lambda \zeta_l +
     \bar\varphi_r \sigma^\mu_\Lambda\zeta_r\right),\\
\nonumber
{\frak A}^\rho_{\gamma^5}(\phi^{-\Lambda},
  \zeta^\Lambda)&=\langle{\cal J} \phi^{-\Lambda}, \gamma^5 \xi^\Lambda\rangle  = \left( \begin{array}{c}
 	\bar\varphi_l \\ \bar\varphi_r 
 	\end{array} \right)^{\!\!\dag} \gamma^0 \gamma^5
 	\left( \begin{array}{r} 
 	\zeta_l\\\zeta_r 
 	\end{array} \right)
=-\int_{\cal M}\!\!\!\!\emph{d}\upmu \,\left( \bar\varphi_l^\dag \zeta_r -
     \bar\varphi_r^\dag \zeta_l\right),
  \end{align}
one obtains the result.
\end{proof}  

\noindent Again, taking the no-boost limit as in remark \ref{rem_noboost}, one
check that 
\eqref{eq:103} yields back the fermionic action \eqref{eq:lagDir1} for the minimal twist
of the spectral triple of electrodynamics. 
\smallskip

Boosting the rule of identification  \eqref{eq:71} in the line of  
\eqref{eq:84}, one identifies  the physical spinors
\begin{align}
  \label{eq:108}
&\Psi :=   \begin{pmatrix}  \psi_l\\ \psi_r\end{pmatrix}
  = S[\Lambda]\, \tilde \Xi = \begin{pmatrix} \tilde \zeta_{1l}\\
    \tilde\zeta_{2r} \end{pmatrix},&
&
\Psi^\dagger :=   \begin{pmatrix} \psi_l^\dag\\ \psi_r^\dag \end{pmatrix} =  i(S[\Lambda]{\cal J}\tilde\phi)^\dag =
  \begin{pmatrix}  i \tilde \varphi^\dag_{1l}\\ i\tilde\varphi^\dag_{2r}\end{pmatrix},\\[4pt]
  \label{eq:92}
 &\Psi' :=   \begin{pmatrix}
    \psi'_l\\ \psi'_r
  \end{pmatrix}
  = 
S[\Lambda]\gamma^0\tilde\Xi =   \begin{pmatrix}
    \tilde \zeta_{2l}\\ \tilde\zeta_{1r}
  \end{pmatrix},& &
{\Psi'}^\dagger := \begin{pmatrix}
    {\psi'_l}^\dagger\\ {\psi'_r}^\dagger
  \end{pmatrix}=-i(S[\Lambda]{\cal J}\gamma^0\tilde\phi)^\dag =
  \begin{pmatrix}
   - i\tilde \varphi^\dag_{2l}\\ -i\tilde\varphi^\dag_{1r},
  \end{pmatrix},
\end{align}
using Def. \ref{defi:notaboost} to write the
components. The lagrangian density \eqref{eq:92ter} becomes
  \begin{equation}
   \label{eq:92bbis}
 \begin{split}      {\cal L}_\Lambda= \psi_{l}^\dag\,\tilde\sigma^\mu_\Lambda\left({\cal D}_\mu+  f_\mu\right)  \,\psi_l
+ \psi_r^\dag\,\sigma^\mu_\Lambda\left({\cal D}_\mu+ f_\mu\right)\psi_r    
+id \left( \psi_l^\dag \psi_r + \psi_r^\dag\psi_l\right).\\ 
 - \; \psi_l^{'\dag}\, \tilde\sigma^\mu_\Lambda
\left({\cal D}_\mu- f_\mu \right)\,\psi'_l  -
\psi_r^{'\dag}\,\sigma^\mu_\Lambda
\left({\cal D}_\mu-  f_\mu\right)\psi'_r
  -i\bar d (\psi_r^{'\dag}\psi'_l +  \psi_l^{'\dag}\psi'_r).
     \end{split}
 \end{equation}
  \begin{remark}
     In the no-boost limit, $\psi_{l,r}, \psi_{l,r}^\dag$ 
 in \eqref{eq:108} coincides with $\psi_{l,r}, \psi_{l,r}^\dag$ in \eqref{eq:spinors}, and
$\psi'_{l,r}$, $\psi_{l,r}^{'\dag}$  in \eqref{eq:92}  with
 $\psi_{r,l}, \psi_{r,l}^\dag$ in \eqref{eq:spinors}. This allows to
 retrieve \eqref{eq:lagDir2} as the no-boost limit of
 \eqref{eq:92bbis},  imposing $d=im$ and taking into account the
 factor $-2$
in \eqref{eq:103} and $4$ in Prop. \ref{prop:eq-Dirac}. Conversely,   $\psi'_{l,r}, \psi{\dag}'_{l,r}$ 
 in \eqref{eq:108} and
$\psi_{l,r}$, $\psi_{l,r}^{'\dag}$  in \eqref{eq:92} coincides with
$\psi'_{l,r}, \psi_{l,r}^{'\dag}$ and $\psi'_{r,l}, \psi_{r,l}^{'\dag}$ in \eqref{eq:spinors}, allowing to
 retrieve \eqref{eq:36bis} as the no-boost limit of \eqref{eq:92bbis}.
    \end{remark}
Treating all the fields independently, one obtains the two pairs of
equations of motion
\begin{align}
  \label{eq:69}
&\tilde\sigma^\mu_\Lambda\left({\cal D}_\mu+
  f_\mu\right)  \,\psi_l  =id \psi_r , 
\quad\sigma^\mu_\Lambda\left({\cal
    D}_\mu+  f_\mu\right)\psi_r    =m \psi_l;\\[4pt] 
  \label{eq:69bbbis}
&\tilde\sigma^\mu_\Lambda\left({\cal D}_\mu- f_\mu \right)\,\psi'_l
  =- i\bar d \psi'_r,
\quad\sigma^\mu_\Lambda\left({\cal
          D}_\mu-  f_\mu\right)\psi'_r
  =-i\bar d \psi'_l.
\end{align}

The generalised energy-momentum $4$-vector $P:=p+g^\flat$ is the sum of the
energy-momentum $p$ with the musical dual of the $1$-form
$g=g_\mu dx^\mu$. In practical, this means
\begin{equation}
  \label{eq:75}
  {\cal D}_\mu\, e^{-ix^\mu p_\mu}= -i P_\mu
\end{equation}
where $P_\mu$ are the components of $P^\flat$. This leads to our final
\begin{proposition}
\label{Prop:Diracboost} For a constant fluctuation $f_\mu$,  a plane wave solution of \eqref{eq:69} ~(resp.~\eqref{eq:69bbbis})
  coincides with a plane wave solution of the Dirac equation with mass
  $m=-(1+i)\frac d2$ (resp. $m=(1+i)\frac{\bar d}2$), whose
 (dual) generalised momentum $P$ has components
 $P'_\nu=\Lambda^\mu_\nu P_\mu$ in the boosted frame, where
\begin{equation}
 P_0 = -f_0,\quad P_j=f_j,\quad\text{resp.}\quad P_0=f_0,\quad
 P_j=-f_j.
\label{eq:76}
 \end{equation}
\end{proposition}
\begin{proof}
 From \eqref{eq:85} and \eqref{eq:86}, a plane wave solution
 \eqref{eq:planewave} of \eqref{eq:69} satisfies  
 \begin{align*}
   \label{eq:77}
 id\psi_r =   \tilde\sigma_\Lambda^\mu({\cal D}_\mu +f_\mu)\psi_l &=
            \tilde\sigma_\Lambda^\mu(-iP_\mu+f_\mu)\psi_l=    \tilde\sigma_M^\nu\left(\Lambda^0_\nu\left(-iP_0 + f_0\right)
            -  i\Lambda^j_\nu\left(-iP_j+f_j\right)\right)\psi_l,
 \end{align*}
 and a similar equation with $\sigma^\mu$, inverting $\psi_l$ and
 $\psi_r$.  If the first part of \eqref{eq:76} 
holds, then these equations are equivalent  to
$\tilde\sigma_M^\nu P'_\mu\psi_l=\frac{-id}{1+i}\psi_r$ and a similar
equation for $\sigma^\mu$. These coincide
with the Dirac equation \eqref{eq:82} with mass $m=-(1+i)\frac
d2$. 
Similarly, a  plane wave solution of \eqref{eq:69bbbis} satisfies
\begin{align*}
-i\bar d\psi'_r=   \tilde\sigma_\Lambda^\mu({\cal D}_\mu -f_\mu)\psi'_l &=
            \tilde\sigma_\Lambda^\mu(-iP_\mu-f_\mu)\psi'_l
=    \tilde\sigma_M^\nu\left(\Lambda^0_\nu\left(-iP_0 - f_0\right)
            -  i\Lambda^j_\nu\left(-iP_j-f_j\right)\right)\psi'_l,
  \end{align*}
which becomes 
$\tilde\sigma_M^\nu P'_\mu\psi'_l=\frac{i\bar d}{1+i}\psi'_r$  if the
second part of \eqref{eq:76}  holds. Together with  a similar
equation for~$\sigma^\mu$, these coincide
with the Dirac equations \eqref{eq:82}, with mass $m=(1+i)\frac{\bar
  d}2$.;
\end{proof}

To guarantee a positive mass, one should impose $d=
m(i\pm 1)$ with $m\in\mathbb R^+$. Identifying the imaginary/real
axis of the complex plane
with the space/time directions of two dimensional Minkowski space, the
set of all physically acceptable values
of $d$ is
the future light-cone, while in the non boosted case it was the
imaginary axis $d=im, m\in\mathbb R$. 

\section{Conclusion and outlook}
\label{conc}

The twisted fermionic action associated to the minimal twist of a
doubled manifold 
and that
of the spectral triple of electrodynamics yields, respectively, the Weyl
and the Dirac equations in lorentzian signature, although one started
with an euclidean manifold. The $1$-form field
parametrizing the twisted fluctuation gets interpreted as an
energy-momentum four-vector. It was known that fluctuations of the
geometry 
generate the bosonic content of the theory (including the Higgs
sector). What is new here is that they generate also the
energy-momentum. In other terms, the dynamics is obtained as a
fluctuation of the geometry !

It should be checked that a
similar transition from the riemannian to the pseudo-riemannian also
takes place for the minimal twist of the Standard Model. This will be
the subject of future works, as well as the extension of these results
to curved riemannian manifolds.
\smallskip

Some points that deserve to be better understood are:

\begin{itemize}
\item Is the twisted fermionic action really lorentzian, since the
manifold $\M$ under which one integrates remains riemannian ?  Actually
this is
not a problem if one
takes as domain of integration a local chart (as in quantum
field theory: the Wick rotation is usually viewed as a local operation),
up to a change of the volume form (see \cite{b5} for
details). Nevertheless, one may hope  that the twist actually changes
the metric on the manifold, through Connes distance formula for
instance (relations
between causal structure and this distance have already been worked
out in \cite{Moretti:2003zw,Franco:2015ab,FE14}, but without taking into account the twist).

\item The twisted fermionic action is invariant under an action of the
  Lorentz group, and the equations of motions in the boosted frame
  coincide with those derived from the Weyl and Dirac equations in the
  boosted frame as well. But the boosted Lagrangians do not agree,  because of the difference of sign in
  the definition of the physical left/right spinors. As stated in the
  text, this sign difference is not
  compatible with the initial restriction to  $\HH_{\cal R}$.
To overcome this difficulty, one may relax this restriction.
 Whether this still permits to define an antisymmetric
  bilinear form, that  yields a physically meaningful action, will
  be investigated elsewhere.
\end{itemize}

In any case, the results presented here suggest an alternative attack to
the problem of extending the theory of spectral triples to
lorentzian geometries. That the twist does not fully implement the Wick
rotation (it does it only for the Hilbert space but not for the Dirac operator) is not so
relevant after all. More than being able to spectrally characterise a
pseudo-riemannian manifold, what matters most for the physics is to obtain an 
action that makes sense in a lorentzian context. The present work shows that this happens for the fermionic action. 

The spectral action
in the twisted context is still an open problem. The interpretation of the
$1$-form field $f_\mu dx^\mu$ as the energy-momentum $4$-vector might
be relevant in this context as well.

Contrary to most approaches in the literature
(e.g. \cite{Barrett:2007vf}, \cite{Du16}), we do not obtain a
lorentzian action by implementing a lorentzian
structure on the geometry. The latter somehow ``emerges'' from the riemannian one.
This actually makes sense remembering that the regularity condition imposed by Connes and Moscovici (see 
Rem.~\ref{rem:autmodul}) has its origin in Tomita's modular
theory. More precisely, the automorphism $\rho$ that defines a twisted spectral triple should
be viewed as the evaluation, at some specific value $t$, of a
one-parameter group of automorphism $\rho_t$. For the minimal twist of
spectral triples, the flip came out as the only automorphism that
makes the twisted commutator bounded. It is not yet
clear what would be the corresponding one-parameter group of automorphisms. Should it exist, this will indicate that the time
evolution in the Standard Model has its
origin in the modular group. This is precisely the content of the
thermal time hypothesis of Connes and  Rovelli \cite{Connes:1994xy}. So far, this hypothesis has
been applied to algebraic quantum field theory \cite{Martinetti:2009ff,Martinetti:2003sp}, and for general
considerations in quantum gravity \cite{Rovelli:2010fk}. Its application to the Standard
Model would be a novelty.

\section*{Acknowledgements}

\noindent
	Part of this work was first presented in Nov. 2018 during the \emph{Int. Conference on
	Noncommutative Geometry} at the Bose National
	Centre for Basic Sciences, Kolkata (India) as a part of 
	125\textsuperscript{th} anniversary celebration of Bose. DS is thankful 
	to the active participants for interesting discussions. DS also thanks W. van 
	Suijlekom for useful discussions and hospitality during May, 2019 at the 
	Math. Dpt. and IMAPP of Radboud University Nijmegen, Netherlands.

PM  F. Lizzi for pointing out mistakes in the first version of
this manuscript, and to M. Filaci for various discussions on the
Lorentz invariance.

\newpage
\appendix
\addcontentsline{toc}{section}{Appendices}
\section*{Appendix}
\setcounter{equation}{0} 

\setcounter{section}{1}
 \subsection{Gamma matrices in chiral representation}
 \label{GammaMatrices}
Let $\sigma_{j= 1,2,3}$ be the Pauli matrices:
\begin{equation}
\label{Pauli}
		\sigma_1 = \left(\begin{array}{cc} 0 & 1 \\ 1 & 0 \end{array}\right)\!,
\qquad	\sigma_2 = \left(\begin{array}{cc} 0 & -i \\ i & 0 \end{array}\right)\!,
\qquad	\sigma_3 = \left(\begin{array}{cc} 1 & 0 \\ 0 & -1 \end{array}\right)\!.
\end{equation}
In four-dimensional euclidean space, the Dirac matrices (in chiral
representation) are
\begin{equation}
\label{EDirac}
	\gamma^\mu =
	\left( \begin{array}{cc}
		0 & \sigma^\mu \\ \tilde\sigma^\mu & 0
	\end{array} \right)\!, \qquad
	\gamma^5 := \gamma^1\gamma^2\gamma^3\gamma^0 =
	\left( \begin{array}{cc}
		\mathbb{I}_2 & 0 \\ 0 & -\mathbb{I}_2 
	\end{array} \right)\!,
\end{equation}
where, for $\mu = 0,j$, we define
\begin{equation}
	\sigma^\mu := \left\{ \mathbb{I}_2, -i\sigma_j \right\}\!, \qquad
	\tilde\sigma^\mu := \left\{ \mathbb{I}_2, i\sigma_j \right\}\!.
\end{equation}
In Minkowski spacetime with signature $(+,-,-,-)$, the Dirac matrices are
\begin{equation}
\label{MDirac}
	\gamma^\mu_M =
	\left( \begin{array}{cc} 
		0 & \sigma^\mu_M \\ \bar\sigma^\mu_M & 0
	\end{array} \right)\!, \qquad
	\gamma^5_M := \gamma^1_M\gamma^2_M\gamma^3_M\gamma^0_M = -i\gamma^5,
\end{equation}
where, for $\mu = 0,j$, we define
\begin{equation}
\label{eq:sigmaMink}
	\sigma^{\mu}_M := \left\{ \mathbb{I}_2, \sigma_j \right\}\!, \qquad
	\bar{\sigma}^{\mu}_M := \left\{ \mathbb{I}_2, -\sigma_j \right\}\!.
\end{equation}

\subsection{Weyl and Dirac equations}
\label{appA}
A Dirac spinor $\Psi = \begin{pmatrix} \psi_l \\ \psi_r \end{pmatrix} \in L^2({\cal M, S})$ is the
direct sum of two Weyl spinors $\psi_l$ and $\psi_r$.
With our definition of the chiral representation, a
left handed spinor is an eigenspinor of the $+1$-eigenspace
$L^2(\M,S)_+$ of the grading operator $\gamma^5$, and a right handed
spinor an eigenspinor of the $-1$ eigenspace $L^2(\M,S)_-$ (in the
physics literature, the convention is usually opposite).
 
 The Dirac lagrangian in Minkowski spacetime is
\begin{equation}
\label{L_M}
\begin{split}
	{\cal L}_M= -\bar\Psi(\eth_M + m)\Psi
& = \left( \begin{array}{cc} \psi_l^\dag & \psi_r^\dag \end{array} \right) \!\!
	\left( \begin{array}{cc} 
		0 & \mathbb{I}_2 \\ \mathbb{I}_2 & 0
	\end{array} \right) \!\!
	\left[ \left( \begin{array}{cc} 
		0 & i\sigma_M^\mu\partial_\mu \\ i\tilde\sigma_M^\mu\partial_\mu & 0
	\end{array} \right) -m \right] \!\!
	\left( \begin{array}{c} \psi_l \\ \psi_r \end{array} \right) \\[3pt]
& = i\psi^\dag_l \tilde\sigma^\mu_M \partial_\mu\psi_l + 
	i\psi^\dag_r \sigma^\mu_M \partial_\mu\psi_r - m
	\left( \psi^\dag_l\psi_r + \psi^\dag_r\psi_l \right)\!,
\end{split}
\end{equation}
where $\bar\Psi := \Psi^\dag\gamma^0$ and $\eth= -i\gamma^\mu \partial_\mu$.
The equations of motion are derived by a variational principle, treating
$\psi_{l/r}$ and their  Hermitian conjugates $\psi^\dag_{l/r}$ as independent 
variables. In particular, the Euler-Lagrange equations for $\psi_l^\dag$, $\psi_r^\dag$
yield Dirac equations (written in components)
\begin{equation}
i\tilde\sigma^\mu_M\partial_\mu\psi_l =
m\psi_r,\qquad i\sigma^\mu_M\partial_\mu\psi_r = m\psi_l.
\label{eq:80}
\end{equation}
By \eqref{eq:sigmaMink} one retrieves the familiar form \cite[19.77]{Coleman}:
  $i\left(\partial_0 \mp \textstyle\sum_{j}
    \sigma_j\partial_j\right)\psi_{l/r} = m\psi_{r/l}$.

A plane wave solution of \eqref{eq:80} is 
\begin{equation}
\label{eq:planewave}
\Psi(x^0, x^j)=\Psi_0 \, e^{-ip_\mu x^\mu} \quad\text{ with } \quad   \Psi_0= \begin{pmatrix}
     \psi_{0l}\\ \psi_{0r}
  \end{pmatrix}
\end{equation}
where  $p_\mu:=\eta_{\mu\nu}p^\nu$ are the components of the $1$-form
$p^\sharp$, dual of the energy-momentum $4$-vector $(p^0, p^j)$
induced by the Lorentz metric
and 
$  \Psi_0$
is a constant spinor solution of 
\begin{equation}
  \label{eq:113}
   i\left(-ip_0+i \textstyle\sum_{j=1}^3 \sigma_j p_j\right)\psi_{0l} = m\psi_{0r}, \quad
  i\left(-ip_0 -i \textstyle\sum_{j=1}^3 \sigma_jp_j\right)\psi_r = m\psi_l.
\end{equation}
The components $\psi_{l/r}=\psi_{0l/r} e^{-ip_\mu x^\mu}$ of the plane
wave \eqref{eq:planewave} are solution of
\begin{equation}
\tilde\sigma^\mu_M p_\mu\psi_l =\left(p_0
  -\sum_{j=1}^3\sigma_j p_j\right)\psi_l =
m\psi_r,\qquad \sigma^\mu_Mp_\mu\psi_r =\left(p_0
  +\sum_{j=1}^3\sigma_j p_j\right)\psi_r= m\psi_l.\label{eq:82}
\end{equation}

For $m=0$, the Dirac lagrangian is the sum of two independent pieces,
the Weyl Lagrangians 
\begin{align}
  \label{eq:WeyllagL}
  	{\cal L}_M^l \!= i\psi^\dag_l
        \tilde\sigma^\mu_M \partial_\mu\psi_l=i\psi^\dag_l(\partial_0 -\sum_{j=1}^3 \sigma_j\partial_j)\psi_l,
\quad	{\cal L}_M^r
 \!= i\psi^\dag_r \sigma^\mu_M \partial_\mu\psi_r 
 = i\psi^\dag_r (\partial_0 + \sum_{j=1}^3 \sigma_j\partial_j)\psi_r,
\end{align}
 that describe Weyl fermions (massless spin-$\frac{1}{2}$ 
particle). The corresponding Weyl equations of 
motion
are  \cite[eq. 19.40, 19.41]{Coleman}:
\begin{equation}
\label{eq:Weyl}
\tilde\sigma^\mu_M \partial_\mu\psi_l
	 = (\partial_0 -
  \textstyle\sum_{j=1}^3\sigma_j\partial_j)\psi_l = 0,\qquad 
\sigma^\mu_M \partial_\mu\psi_r 
	= (\partial_0 +\textstyle\sum_{j=1}^3 \sigma_j\partial_j)\psi_r = 0.
\end{equation}
Their plane-wave solutions,
\begin{align}
  \label{eq:50}
  \psi_{l}(x^0, x^j) =\psi_{0l} \, e^{-ip_\mu x^\mu}, \quad  \psi_{r}(x^0, x^j) =\psi_{0r}\,  e^{-ip_\mu x^\mu}
\end{align}
with $\psi_{0l}, \psi_{0r}$ momentum-dependant spinors 
satisfying \eqref{eq:113} for $m=0$, are  solution of
\begin{align}
  \label{eq:52}
 (p_0 -\sum_{j=1}^3 \sigma_j p_j) \psi_{0l} =0,\quad 
 (p_0 + \sum_{j=1}^3 \sigma_j p_j )\psi_{0r} =0.
\end{align}
\subsection{Spin representation of boosts}
 \label{appC}

The spinor representation of a boost of rapidity $b/2$
in the direction $\bf n$ is given by
\begin{equation}
\label{Lambda}
	S[\Lambda] = \left( \begin{array}{cc} 
		\Lambda_+ & 0 \\ 0 & \Lambda_-
	\end{array} \right)\!, \quad \text{ where }\quad \Lambda_\pm
      := \exp(\pm{\bf a}.\boldsymbol\sigma)\; \text{ with }\,{\bf a} :=
\frac b2 {\bf n}.
\end{equation}

 Collecting the terms with even
  and odd powers in the expansion of~$\exp(\pm\bf
  a.\boldsymbol\sigma)$, one checks that
$\Lambda_\pm = \Lambda_1 \pm\Lambda_2$ 
where 
$\Lambda_1:=\left(\cosh |{\bf a}|\right) \mathbb I_2$, $\Lambda_2:=
(\sinh |{\bf a}|)\,{\bf  n}\cdot{\boldsymbol \sigma}$.  Thus  $\Lambda_+, \Lambda_-$ are both selfadjoint, and inverse of one
another. Meaning that $S[\Lambda]$ is selfadjoint but not unitary 
\begin{equation}
\label{eq:selfSL}
S[\Lambda]^\dag = S[\Lambda]\neq S[\Lambda]^{-1}. 
\end{equation}

Under such a boost, a lorentzian spinor and the lorentzian Dirac operator 
transform as
\begin{equation}
\label{eq:boost}
	\psi_M \to   S[\Lambda]\psi_M,\quad 
	\eth_M \to  S[\Lambda]\;\eth_M \;S[\Lambda]^{-1}.
\end{equation}
By construction, the spin representation of the Lorentz group is such
that
(see e.g. \cite[20.78]{Coleman})
\begin{equation}
  \label{eq:84bis}
(\tilde\sigma^\mu_M)_\Lambda :=S[\Lambda] \tilde\sigma^\mu_M S[\Lambda]^{-1}= \Lambda^\mu_\nu
\tilde\sigma^\nu_M , \quad (\sigma^\mu_M)_\Lambda= S[\Lambda] \sigma^\mu_M S[\Lambda]^{-1}= \Lambda^\mu_\nu \sigma^\nu_M 
\end{equation}
where $\left\{\Lambda^\mu_\nu\right\}$ is the matrix
representation of the Lorentz group on Minkowski space. Since $\tilde
\sigma^0=\tilde \sigma^0_M$, $\sigma^0=\sigma^0_M$ and $\tilde\sigma^j=-i\tilde\sigma^j_M$,
$\sigma^j=-i\tilde\sigma^j_M$  for $j=1,2,3$, one gets
\begin{align}
  \label{eq:85}
 &\tilde\sigma^0_\Lambda:= S[\Lambda] \,\tilde\sigma^0\, S[\Lambda]^{-1}=
\Lambda^0_\nu \,\tilde\sigma^\nu_M, &&\sigma^0_\Lambda:= S[\Lambda] \sigma^0 S[\Lambda]^{-1}=\Lambda^0_\nu \sigma^\nu_{M},\\
  \label{eq:86}
   &\tilde\sigma^j_\Lambda:= S[\Lambda] \tilde\sigma^j S[\Lambda]^{-1}=
 -i\Lambda^j_\nu\tilde\sigma^\nu_M, &&\sigma^j_\Lambda:= S[\Lambda] \tilde\sigma^j S[\Lambda]^{-1}=
  -i\Lambda^j_\nu\sigma^j_M.
\end{align}

\begin{spacing}{0.9}	

\end{spacing}

\end{document}